\documentclass[12pt]{article}

\newcommand{\levytbp}{\mu_{\tbp}}
\newcommand{\tbp}{\mathrm{3BP}}

\newcommand{\levync}{\nu}
\newcommand{\levyonepop}{\mu}
\newcommand{\ratemeasproposed}{\nu_{\mathrm{prop}}}

\newcommand{\allhyper}{\bm{\zeta}} % vector of all hyperparameters

\newcommand{\concinternal}{c}
\newcommand{\discountinternal}{\sigma}
\newcommand{\rateinternal}{\sigma}
\newcommand{\corrinternal}{\gamma}

\newcommand{\mass}{\alpha}
\newcommand{\conc}[1]{\concinternal_{#1}}
\newcommand{\discount}[1]{\discountinternal_{#1}}
\newcommand{\rate}[1]{\rateinternal_{#1}}
\newcommand{\corr}[1]{\corrinternal_{#1}}

\newcommand{\numpops}{P} % NUMBER OF POPULATIONS
\newcommand{\popidx}{p} % INDEX TO BE USED FOR POPULATIONS
\newcommand{\notpopidx}{\neg \popidx} % the other index; if p = 1, not p = 2 and vice versa
\newcommand{\samplesize}[1]{N_{#1}}
\newcommand{\genericsamplesize}{\samplesize{\popidx}}
\newcommand{\samplesizenosubscript}{N}
\newcommand{\occurrence}[1]{k_{#1}}

\newcommand{\followupidx}[1]{m_{#1}}

\newcommand{\vecoccurrence}{\bm{k}}
\newcommand{\zerovec}{\bm{0}}

\newcommand{\unitidx}{n}
\newcommand{\futureunitidx}{m}

\newcommand{\vecsamplesize}{\bm{N}}
\newcommand{\futuresamplesize}[1]{M_{#1}}
\newcommand{\genericfuturesamplesize}{\futuresamplesize{\popidx}}
\newcommand{\vecfuturesamplesize}{\bm{M}}

\newcommand{\variantfreqplain}{\theta} % plain variant frequency, no vectors or subscripts
\newcommand{\numvariants}[1]{J_{#1}}
\newcommand{\genericnumvariants}{\numvariants{\vecsamplesize}}
\newcommand{\variantidx}{\ell}
\newcommand{\variantlabel}[1]{\psi_{#1}}
\newcommand{\genericvariantlabel}{\variantlabel{\variantidx}}
\newcommand{\variantcount}[3]{x_{#1,#2,#3}}
\newcommand{\genericvariantcount}{\variantcount{\popidx}{\unitidx}{\variantidx}}
\newcommand{\genericfuturevariantcount}{\variantcount{\popidx}{\genericsamplesize+\futureunitidx}{\variantidx}}
\newcommand{\measurevariantcount}[2]{X_{#1,#2}}

\newcommand{\genericmeasurevariantcount}{\measurevariantcount{\popidx}{\unitidx}}

\newcommand{\pilotdata}{\measurevariantcount{1:2}{\vecone:\vecsamplesize}}
\newcommand{\vecvariantfreq}[1]{\bm{\variantfreqplain}_{#1}}
\newcommand{\genericvecvariantfreq}{\vecvariantfreq{\variantidx}}
\newcommand{\variantfreq}[2]{\variantfreqplain_{#1, #2}}
\newcommand{\variantfreqperpop}[1]{\variantfreqplain_{#1}} % single-population density. takes population as arg
\newcommand{\genericvariantfreq}{\variantfreq{\popidx}{\variantidx}}

\newcommand{\news}[2]{U_{ #1}^{(#2)}}
\newcommand{\obsnewsNoargs}{u}
\newcommand{\obsnews}[2]{\obsnewsNoargs_{ #1}^{(#2)}}

\newcommand{\genericnews}{\news{\vecsamplesize}{\vecfuturesamplesize}}

\newcommand{\genericfreqnews}{\news{\vecsamplesize}{\vecfuturesamplesize, \vecoccurrence}}
\newcommand{\newsparam}[2]{\lambda_{#1}^{(#2)}}

\newcommand{\genericfreqnewsparam}{\newsparam{\vecsamplesize}{\vecfuturesamplesize, \vecoccurrence}}

\newcommand{\vecone}{\textbf{1}} % vector of ones
\newcommand{\dirac}[1]{\delta_{#1}}
\newcommand{\bernoullirv}[1]{\textrm{Bernoulli}(#1)}
\newcommand{\pois}{\mathrm{Poisson}}

\newcommand{\randommeas}{\Theta}
\newcommand{\de}{\mathrm{d}}
\newcommand{\BeP}{\mathrm{BeP}}
\newcommand{\mBeP}{\mathrm{mBeP}} %multiple population bernoulli process
\newcommand{\PPP}{\mathrm{PPP}}

\newcommand{\NN}{\mathbb{N}}
\newcommand{\train}{\text{train}}

\newcommand{\propt}{\rho}
\newcommand{\boundsgeneric}{\Phi}
\newcommand{\poissonize}{\boundsgeneric}
\newcommand{\lowerbound}{\boundsgeneric_{\mathrm{lower}}}
\newcommand{\altlowerbound}{\boundsgeneric_{\mathrm{lower}}'} 
\newcommand{\upperbound}{\boundsgeneric_{\mathrm{upper}}}
\newcommand{\lowerboundproject}{\boundsgeneric_{\mathrm{lower}}^{\mathrm{project}}}
\newcommand{\upperboundproject}{\boundsgeneric_{\mathrm{upper}}^{\mathrm{project}}}
\newcommand{\lowerboundproport}{\boundsgeneric_{\mathrm{lower}}^{\mathrm{proport}}}
\newcommand{\upperboundproport}{\boundsgeneric_{\mathrm{upper}}^{\mathrm{proport}}}
\newcommand{\smallpositive}{\delta}
\newcommand{\constantgeneric}{C}
\newcommand{\constantupper}[1]{\constantgeneric_{\mathrm{upper},#1,\smallpositive}}
\newcommand{\constantlower}[1]{\constantgeneric_{\mathrm{lower},#1}}

\newcommand{\constantlowerwithotherrate}[1]{\constantlower{#1}'}

\newcommand{\constantlowerwithrateprime}[1]{\constantlower{#1}''}

\usepackage{algpseudocode}
\usepackage{amsmath, amssymb, amsthm}
\usepackage{float, fullpage, graphicx, multirow,parskip, subcaption, setspace}
\usepackage{comment}
\usepackage{url}
\usepackage{enumitem}
\usepackage{tikz}
\usepackage{bm, bbm}
\usetikzlibrary{shapes.geometric, positioning}
\usetikzlibrary{quotes, angles}
\usepackage{rotating}
\usepackage{xr, xr-hyper}
\usepackage{hyperref}
\usepackage[capitalise,noabbrev]{cleveref}

\usepackage{natbib}
\usepackage{placeins}
\usepackage{algorithm2e}
\RestyleAlgo{ruled}

\definecolor{SkyBlue}{RGB}{14, 118, 188}
\definecolor{BrightRed}{RGB}{223,82, 78}

\hypersetup{pdfborder = {0 0 0.5 [3 3]}, colorlinks = true, linkcolor = BrightRed, citecolor = SkyBlue, filecolor = BrightRed}

\DeclareMathOperator{\BER}{Bernoulli}

\DeclareMathOperator{\Pois}{Poisson}

\def\Bernoulli{\mathrm{Bernoulli}}
\def\Poisson{\mathrm{Poisson}}
\def\Beta{\mathrm{Beta}}

\def\btheta{\boldsymbol{\theta}}
\def\bTheta{\boldsymbol{\Theta}}
\def\R{\mathbb{R}}

\def\E{\mathbb{E}}

\newtheorem{myTheorem}{Theorem}
\newtheorem{myLemma}{Lemma}

\newtheorem{myExample}{Example}
\newtheorem{myProposition}{Proposition}
\newtheorem{myRemark}{Remark}
\newtheorem{myCorollary}{Corollary}

\newtheorem{desiderata}{Desideratum}

\newtheorem{desiderataalt}{Desideratum}[desiderata]
\newenvironment{desideratap}[1]{
  \renewcommand\thedesiderataalt{(#1)}
  \desiderataalt
}{\enddesiderataalt}
\makeatletter
\if@cref@capitalise
\crefname{desiderataalt}{Desideratum}{Desiderata}
\else
\crefname{desiderataalt}{desideratum}{desiderata}
\fi
\Crefname{desiderataalt}{Desideratum}{Desiderata}
\makeatother

\newtheorem{condition}{Condition}

\newtheorem{conditionalt}{Condition}[condition]
\newenvironment{conditionp}[1]{
\renewcommand\theconditionalt{(#1)}
\conditionalt
}{\endconditionalt}

\title{Double trouble: Predicting new variant counts across two heterogeneous populations}
\author{%
  Yunyi Shen 
  \thanks{EECS, MIT, \texttt{yshen99@mit.edu} }\and
   Lorenzo Masoero  \thanks{EECS, MIT, \texttt{lom@mit.edu} }\ \and %\thanks{Amazon.com,\texttt{masoerl@amazon.com}, Work not related to Amazon }\and 
  Joshua G. Schraiber \thanks{Department of Quantitative and Computational Biology, USC, \texttt{schraibe@usc.edu}} 
  \and 
  Tamara Broderick \thanks{EECS, MIT, \texttt{tbroderick@mit.edu}}
}

\begin{document}
\maketitle
% TL;DR:We consider multiple heterogeneous populations and propose a new method for predicting the number of new (shared and private) variants in a follow-up study given a pilot study. 
\begin{abstract}

	% standalone abstract for submission purposes

	Collecting genomics data across multiple heterogeneous populations (e.g., across different cancer types) has the potential to improve our understanding of disease. Despite sequencing advances, though, resources often remain a constraint when gathering data. So it would be useful for experimental design if experimenters with access to a pilot study could predict the number of new variants they might expect to find in a follow-up study: both the number of new variants shared between the populations and the total across the populations. While many authors have developed prediction methods for the single-population case, we show that these predictions can fare poorly across multiple populations that are heterogeneous. We prove that, surprisingly, a natural extension of a state-of-the-art single-population predictor to multiple populations fails for fundamental reasons. We provide the first predictor for the number of new shared variants and new total variants that can handle heterogeneity in multiple populations. We show that our proposed method works well empirically using real cancer and population genetics data.
\end{abstract}

\newpage
\section{Introduction}
\label{sec:introduction}
% !TEX root = ../more_for_less.tex

The genomic revolution has fundamentally changed our understanding of human history \citep{green2010draft, gravel2011demographic, tennessen2012evolution, lazaridis2014ancient, haak2015massive}, the genetics of complex disease \citep{lek2016analysis, karczewski2020mutational, wang2021rare, backman2021exome,chen2023genomic}, and the development of cancer \citep{cancer2008comprehensive, cancer2013cancer, kandoth2013mutational, jayasinghe2018systematic, hoadley2018cell, huang2018pathogenic}. As more and more sequencing data is gathered, novel variants continue to be discovered \citep{lek2016analysis, karczewski2020mutational} and yield a deeper insight into the evolutionary processes that have shaped the observed patterns of genetic variation \citep{agarwal2021mutation, agarwal2023relating}. Moreover, rare variants, which are typically detected in large samples, have become increasingly important in understanding the etiology of disease and are a key tool in development of novel therapeutics against common disease \citep{szustakowski2021advancing}.

Most studies of genetic variation sample from two or more groups that may share some portion of their variation, such as different human populations or different tumor types. For instance, sampling from multiple human populations is essential for discovering more variants \citep{karczewski2020mutational,chen2023genomic} and is crucial to expanding equity and inclusion in genomic medicine \citep{popejoy2016genomics, hindorff2018prioritizing, fatumo2022roadmap}. Due to the shared evolutionary history of all humans, in this setting many common variants will be shared between populations \citep{biddanda2020variant}, but differences in allele frequency between populations are key to understanding adaptation \citep{novembre2009spatial} and disease \citep{fu2011multi, assimes2021large, roselli2018multi}. Similarly, when collecting disease case and control samples to examine the enrichment of rare variants in cases, the vast majority of the variants will be shared. Rare variants shared across different type of diseases have drawn attention since they may be linked to common genetic bases of different diseases \citep{rafnar2009sequence, bojesen2013multiple, rashkin2020pan}. On the other hand, most variation between tumors in different individuals will be private, but shared variation between tumors across samples may point to the genetic basis of cancer, and thus paths toward novel treatments \citep{st2014genomic, ma2015rise,prior2020frequency,pirozzi2021implications}.

Despite large decreases in the cost of sequencing \citep{christensen2015assessing, park2016trends, schwarze2018whole}, detecting novel variants still requires a large enough set of samples to benefit from careful planning. This is especially true when attempting to expand our catalogue of variation across diverse human groups, because some of the most diverse human groups have been studied the least \citep{popejoy2016genomics}. It would therefore be of significant practical utility to have a reliable and data-informed planning method; namely, a prediction method should accurately estimate the number of novel variants in a new (follow-up) study of multiple populations, given data from a typically small, previous (pilot) study. Moreover, prediction of the number of shared or private variants seen in a follow-up study can be useful as a null model for identifying ``unusual'' variants or classes of variants that may indicate their role in biologically interesting processes, such as adaptation or disease. In all of these cases, it is crucial to account for novel variants that are shared between groups as opposed to those that are private to a single group.

Several researchers have developed approaches to predict the number of new variants that will be detected in a follow-up sequencing study, given information from a pilot study \citep{Camerlenghi2021, Masoero2022, chakraborty2019using, zou2016quantifying,gravel2014predicting,ionita2009estimating}. However, all of these methods assume samples are from a single population. So these methods cannot be used directly to predict the number of new shared variants across different populations in a follow-up study.
And because these methods do not model the possibility of shared variants between groups, single-population approaches can be expected to misestimate the number of novel variants in a follow-up study.
  %some of which explicitly modeled the impact of sequencing depth. 
%In principle, these approaches can be used to determine the benefits of sequencing additional individuals and determine the optimal sample size and sequencing coverage given a fixed budget. 

In this work, we develop an approach to predict the number of novel variants discovered in samples from multiple populations or groups. We focus on the case of two populations here, although in principle our approach can be extended to more than two. We work in the setting where a pilot study has collected genetic data in each of the two groups, and the experimenter wishes to collect additional data in a larger study.
In principle, the variant prediction challenge can be tackled in two ways: (1) one could construct an explicit model of variation within and between groups based on first principles or (2) one could employ a ``black box'' that is agnostic to the underlying biology. This work, like the single-population approaches referenced above, adopts the latter approach. We hope that, by taking the black-box approach, our method can apply to data that arise from very different processes -- ranging across tumor sequencing, case-control studies, population genetics, and other settings. 

In the one-population case, \citet{Masoero2022} showed that a Bayesian nonparametric (BNP) approach can deliver strong empirical performance. 
However, we here prove that a natural two-population extension to the method of \citet{Masoero2022} fails for fundamental reasons (\cref{sec:cant_have_it_all}). We nonetheless demonstrate how to take advantage of the flexibility of the Bayesian framework (\cref{sec:bnp_two_pops}) to develop a novel two-population BNP model (\cref{sec:our_model}) and methodology (\cref{sec:prediction_all}). We provide theory to establish the desirable power-law properties of our model in \cref{sec:powerlaw}. Finally, we use both simulations and analysis of real cancer and population genetic datasets to showcase the accuracy and robustness of our model (\cref{sec:experiments}); along the way, we highlight multi-population circumstances where using one-population approaches can be misleading.

\section{Problem setup}
\label{sec:setup}
% !TEX root = ../double_trouble_main.tex

We consider the case where we have genomic samples from a pilot study conducted in two populations, indexed by $\popidx \in \{1,2\}$,
and are interested to predict variant numbers in a follow-up study. 
Let $\samplesize{\popidx}$ be the number of pilot samples from population $\popidx$; collect the vector of pilot sample
cardinalities in $\vecsamplesize := (\samplesize{1}, \samplesize{2})$.
We assume there is a common fixed reference genome for the two populations, and we assume for convenience that
each variant is given a unique label; for instance, one might assign a random label, uniform on $[0,1]$, to a variant upon its discovery.
Take the $\unitidx$-th sample in the $\popidx$-th population; we let $\genericvariantcount$ equal 1 if the variant with label 
$\genericvariantlabel$ is observed in this sample and 0 otherwise. Let $\genericnumvariants$ denote the (scalar) total number
of unique variants observed across the $\vecsamplesize$ samples. Then the following measure will have a mass
of size 1 at each variant label where the $\unitidx$-th sample in the $\popidx$-th population differs from the reference genome:
\begin{equation} \label{eq:sample_meas}
	\genericmeasurevariantcount := \sum_{\variantidx=1}^{\genericnumvariants} \genericvariantcount \dirac{\genericvariantlabel},
\end{equation}
where $\dirac{\genericvariantlabel}$ is a Dirac mass at $\genericvariantlabel$; that is, $\dirac{\genericvariantlabel}(\variantlabel{})$ equals 1 if $\variantlabel{} = \genericvariantlabel$ and 0 otherwise. 
Equivalently, we might let $\genericvariantcount = 0$ for an imagined countable infinity of possible yet-to-be-observed variants, each assigned a new label $\genericvariantlabel$ for $\variantidx > \genericnumvariants$; then we have
\begin{equation}
	\genericmeasurevariantcount = \sum_{\variantidx=1}^{\infty} \genericvariantcount \dirac{\genericvariantlabel}. \label{eq:observation_measure_inf}
\end{equation}
To represent the pilot study, we write $1{:}\samplesize{\popidx}$ for indices $\{1,\ldots,\samplesize{\popidx}\}$. We can then write
$\measurevariantcount{\popidx}{1:\samplesize{\popidx}}$ for the pilot samples in population $\popidx$. And we let
$\pilotdata$ denote the full collection of pilot samples, across both populations.

Given the pilot study, we would like to predict (1) the total number of new variants in a follow-up study, as well as (2) the number of new variants shared between the populations.
To that end, let $\futuresamplesize{\popidx}$ be the number of new samples from $\popidx$ considered for the follow-up study; collect
the vector of follow-up study sample cardinalities in $\vecfuturesamplesize = (\futuresamplesize{1},\futuresamplesize{2})$. 
First, then, we are interested in the number of variants that (a) are not present in the first $\vecsamplesize$ pilot samples but (b) are
present in at least one of the $\vecfuturesamplesize$ follow-up samples. We call this quantity $\genericnews$, 
where the $U$ stands for ``unseen,'' and write:
\begin{equation}
	\genericnews = 
		\sum_{\variantidx=1}^{\infty} \left[ \prod_{\popidx=1}^{2} \mathbf{1}\left( \sum_{\unitidx=1}^{\genericsamplesize} \genericvariantcount = 0\right) \right] 
				\mathbf{1}\left( \sum_{\popidx=1}^{2} \sum_{\futureunitidx=1}^{\genericfuturesamplesize} 
				\variantcount{\popidx}{\genericsamplesize+\futureunitidx}{\variantidx} > 0 \right). 
	\label{eq:news}
\end{equation}

We next consider the number of new variants appearing $\occurrence{1}$ times in population 1 and $\occurrence{2}$ times in population 2.
In a single population, the variants appearing once are often called singletons, the variants appearing twice called doubletons, and so on. Analogously and collecting $\vecoccurrence := (\occurrence{1},\occurrence{2})$, we will denote the variants appearing $\vecoccurrence$ times across the two populations as \emph{$\vecoccurrence$-tons}.
We can write the number of $\vecoccurrence$-tons in the follow-up as
\begin{equation}
	\genericfreqnews = 
		\sum_{\variantidx=1}^{\infty} \left[ \prod_{\popidx=1}^{2} \mathbf{1}\left( \sum_{\unitidx=1}^{\genericsamplesize} \genericvariantcount = 0\right) \right]
		 \left[ \prod_{\popidx=1}^{2} \mathbf{1}\left( \sum_{\futureunitidx=1}^{\genericfuturesamplesize} \genericfuturevariantcount = \occurrence{\popidx} \right) \right]. 
	\label{eq:news_freq}
\end{equation}
The total number of new variants appearing in both populations in the follow-up can be found by summing $\genericfreqnews$ over integers $\occurrence{1},\occurrence{2}\ge1$. And
$\genericnews$ can similarly be recovered by summing over $\vecoccurrence$ such that $\occurrence{1} + \occurrence{2} \ge 1$.

Our aim in what follows is, then, to predict $\genericnews$ and $\genericfreqnews$, across $\vecfuturesamplesize$ and $\vecoccurrence$, given the pilot data
$\pilotdata$.

\section{A Bayesian framework for two populations}
\label{sec:bnp_two_pops}
% !TEX root = ../double_trouble_main.tex

Before choosing a particular model, we describe a general Bayesian framework for predicting 
follow-up variant numbers given pilot data in two populations.

First, we assume that variants in a particular sample appear independently of each other;
specifically, within any population $\popidx$ and sample $\unitidx$, we assume $\genericvariantcount$ are independent
across $\variantidx$. While ignoring linkage disequilibrium may seem at first to be too strong of an oversimplification, 
every existing work in a single population makes the same assumption \citep{Camerlenghi2021, Masoero2022,chakraborty2019using, zou2016quantifying,gravel2014predicting,ionita2009estimating}. And importantly, our empirical results below
show that our proposed method performs well despite this assumption.

Next we assume that samples within a single population are exchangeable; that is,
we assume a distribution over the samples and further that the probability of the data is unchanged 
when the samples are observed in a different order. By de Finetti's theorem \citep{de1929funzione}, the 
two preceding assumptions imply the existence of an unknown parameter $\genericvariantfreq \in [0,1]$
representing the frequency of variant $\variantidx$ in population $\popidx$,
such that $\genericvariantcount \sim \bernoullirv{\genericvariantfreq}$, i.i.d.\ across $\unitidx$
and independent across $\popidx$ and $\variantidx$.

We will find it convenient to collect the 
$\variantidx$-th variant frequencies across populations
in the vector $\genericvecvariantfreq = (\variantfreq{1}{\variantidx}, \variantfreq{2}{\variantidx})$. 
Then we can match the frequency vector $\genericvecvariantfreq$ with its variant label in the
following vector-valued 
measure: $\randommeas :=  \sum_{\variantidx=1}^{\infty} \genericvecvariantfreq \dirac{\genericvariantlabel}$.

To take a Bayesian approach, it remains to specify a prior over the frequencies.
In the one-population case, \citet{Masoero2022} generate the frequencies according
to a Poisson point process with a diffuse (i.e., non-atomic) rate measure, so we propose
an analogous approach in the two-population case. In particular, we treat the labels $\genericvariantlabel$
as values of convenience rather than meaningful in any way; without loss of generality, 
we can imagine them as generated i.i.d.\ uniform on $[0,1]$ in what follows, to ensure almost sure 
uniqueness.
Then we write $\randommeas \sim \PPP(\levync)$ to indicate that $\randommeas$
is generated by drawing $\{\genericvecvariantfreq\}_{\variantidx=1}^{\infty}$ from a Poisson
point process (PPP) with rate measure $\levync(\de\vecvariantfreq{})$ and drawing the $\genericvariantlabel$
i.i.d.\ uniform on $[0,1]$.

We call the resulting model the \emph{two-population model}. When we draw the $\genericvariantcount$ conditional on 
$\randommeas$ with the independent Bernoulli draws described above and collect them in the measures 
$\pilotdata$ as in \cref{eq:sample_meas}, we say that $\pilotdata$
is drawn according to a \emph{multiple-population Bernoulli process} ($\mBeP$) with measure parameter $\randommeas$ and count-vector parameter $\vecsamplesize$. Finally, then the two population model
can be written succinctly as:
\begin{align}
	\begin{split}
		\randommeas & \sim \PPP(\levync) \quad \textrm{ (Prior) } \\
		\pilotdata \mid \randommeas & \sim \mBeP(\randommeas, \vecsamplesize). \quad \textrm{ (Likelihood) }
		\label{eq:generative_model}
	\end{split}
\end{align}
%\footnote{\textcolor{red}{lom: I might be naive and/or missing something but something that just occurred to me upon re-reading: but should we be more explicit about our notation here? really what we end up having is a model in which $\genericmeasurevariantcount$ ``draws'' from the corrsponding $\theta_{p,k}$; what does it mean to have a Bernoulli process if $\randommeas$ is supported on $[0,1]^2 \times \Omega$?} \tb{hopefully now addressed!} \yunyi{I add explicitly that variant weights $x$'s are bernoulli independent}}
%While we find the $\mBeP$ notation convenient, it is equivalent to say that the $\genericmeasurevariantcount$ are drawn according to Bernoulli processes ($\BeP$s), independently across $\unitidx$ while conditional on $\randommeas$ and with different rate measures in each population specified by the relevant parts of $\randommeas$.
It remains to choose a rate measure $\levync(\de\vecvariantfreq{})$ for the two-population model.

\section{A natural extension from one population fails}
\label{sec:cant_have_it_all}
We here delineate desired model behaviors and then use these behaviors to rule out what
might at first seem like a natural choice of the rate measure $\levync(\de\btheta)$, namely a factorized rate measure.

As a first desirable behavior in our present model, we know that any real-life sample can, out of necessity, exhibit only finitely many variants.
\begin{desideratap}{A}
	Each sample exhibits finitely many variants. 
	I.e., $\forall \popidx, \unitidx$, we have $\sum_{\variantidx = 1}^{\infty} \genericvariantcount < \infty$.
	\label{req:finite_mean}
\end{desideratap}
Second, while in reality there exists a fixed, finite upper bound on the number of variants,
we are considering cases where the number of samples in the pilot and follow-up are small relative to that upper bound.
Therefore, it is conceptually and computationally convenient to avoid modeling the upper bound entirely;
instead, we assume that, if we sequence enough new samples, we will always discover new variants.
\begin{desideratap}{B}
	%No matter how many variants we have seen so far, there are always more variants to discover.
	%That is, if we were to keep sequencing samples, we would eventually encounter a new variant.
	No matter how many variants we have seen so far in any population, there are always more variants to discover in each population. That is, if we were to keep sequencing samples (no matter which population each sample is from), we would eventually encounter a new variant.
	\label{req:inf_mass}
\end{desideratap}

The model of \citet{Masoero2022} satisfies \cref{req:finite_mean,req:inf_mass} (for a single population; i.e., $p \in \{1\}$).
In particular, \citet{Masoero2022} generate frequencies according to a PPP with rate measure
$\levytbp(\de\variantfreqplain)= \mass{} \variantfreqplain^{-\discount{}-1}(1-\variantfreqplain)^{\conc{}+\discount{}-1}\de\variantfreqplain$
and hyperparameter values $\mass>0, \discount{} \in [0,1), \conc{}>-\discount{}$.
Given a particular variant label, variants appear in samples i.i.d.\ Bernoulli across samples, conditioned
on the frequency.
The resulting frequencies are said to be generated according to a \emph{three-parameter
beta process} (3BP) \citep{teh2009indian,Broderick2012,james2017bayesian}.
Since $\int \variantfreqplain \levytbp(\de\variantfreqplain)<\infty$, 
the number of variants in any sample is a.s.\ finite. And since $\int \levytbp(\de\variantfreqplain)=\infty$,
the number of non-zero frequencies from the PPP is countably infinite.

Given the success of this one-population rate measure,
a natural proposal for the two-population case is to choose a PPP rate measure 
that factorizes into two familiar one-population rate measures.
\begin{conditionp}{C} \label{req:product_measure}
	The rate measure for the frequency-generating PPP in two populations factorizes across populations:
	$\levync(\de\btheta) = \levyonepop(\de\variantfreqperpop{1}) \levyonepop(\de\variantfreqperpop{2})$.  
\end{conditionp}
%The positive expectation excluded the trivial case where one of them is a $\delta$ measure at 0, i.e. no variants in that population. For instance, one might choose $\levyonepop = \levytbp$.

Our next result shows that a two-population model with a factorized rate measure cannot simultaneously
satisfy both \cref{req:finite_mean,req:inf_mass}.

\begin{myTheorem} \label{prop:incomp}
	Take the two-population model of \cref{eq:generative_model}. If the Poisson point process
	rate measure in the prior factorizes across populations according to \cref{req:product_measure},
	the model cannot simultaneously generate samples with finitely many variants (\cref{req:finite_mean})
	and guarantee that there are always more variants to discover (\cref{req:inf_mass}).
\end{myTheorem}

See \cref{sec:product_measure} for the proof. The rough intuition is as follows. We observe that \cref{req:finite_mean}
is equivalent to $\forall \popidx \in \{1,2\}, \int_{[0,1]^{2}}\variantfreqperpop{\popidx}\levync(\de\vecvariantfreq{}) < \infty$.
To derive a contradiction, suppose that \cref{req:inf_mass} and \cref{req:product_measure} hold. By \cref{req:product_measure}, we can factorize
\[
	\int_{[0,1]^{2}}\variantfreqperpop{1}\levync(\de\vecvariantfreq{}) = \int_{[0,1]}\variantfreqperpop{1} \levync_{1}(\de\variantfreqperpop{1}) \int_{[0,1]}\levync_{2}(\de\variantfreqperpop{2}).
\]
We next show that \cref{req:inf_mass} and \cref{req:product_measure} together imply that at least one of the factors in the product must be infinite and the remaining factor must be
strictly positive. So the full product must be infinite, a contradiction with \cref{req:finite_mean}. See \cref{sec:product_measure} for
a fuller and more detailed treatment.

%The key intuition in the proof is that a variant at a particular
%locus will have the same label in either population, and thus variants across populations are inherently linked.
%As a result, any two-population model satisfying \cref{req:finite_mean,req:inf_mass} must generate strictly positive frequencies
%such that
%$\variantfreq{1}{\variantidx} \to 0$ and $\variantfreq{2}{\variantidx} \to 0$ as $\variantidx \to \infty$.
%It follows that the rate measures cannot act independently in the two populations.

\section{Our Model}
\label{sec:our_model} %\label{sec:dependent_prior}
To achieve the desiderata in \cref{sec:cant_have_it_all}, we propose a new rate measure. We show that the resulting two-population model (\cref{eq:generative_model}) satisfies \cref{req:finite_mean,req:inf_mass}. And we show how to make predictions using our model.

In what follows, we propose to use the two-population model with the following PPP rate measure on the variant frequencies:
\begin{equation}
	\ratemeasproposed(\de\vecvariantfreq{}) := \mass
			\frac{(\variantfreqperpop{1}+\variantfreqperpop{2}^{\rate{2}/\rate{1}})^{-\rate{1}}}{(\variantfreqperpop{1}+\variantfreqperpop{2})^{\corr{1}+\corr{2}}}
			\frac{\variantfreqperpop{1}^{\corr{1}-1}(1-\variantfreqperpop{1})^{\conc{1}-1}}{B(\corr{1}, \conc{1})}\frac{\variantfreqperpop{2}^{\corr{2}-1}(1-\variantfreqperpop{2})^{\conc{2}-1}}{B(\corr{2}, \conc{2})} \bm{1}_{[0,1]^2} (\vecvariantfreq{}) \de\btheta.
	\label{eq:proposed_prior}
\end{equation}
%\footnote{LM: I added $\levync$ - let's pick a better symbol tho}
Here $B(\cdot, \cdot)$ denotes the beta function. 
The measure $\levync$ is characterized by seven scalar hyperparameters, which we collect in the vector $\allhyper := (\mass, \rate{1}, \rate{2}, \corr{1}, \corr{2}, \conc{1}, \conc{2})$. Before discussing their roles next, we observe that our rate measure takes the form of the product rate measure $\levytbp(\de\variantfreqperpop{1}) \levytbp(\de\variantfreqperpop{2})$ multiplied by a term that does not factorize into components unique to 
frequencies in population 1 and 2.

The \emph{mass} parameter $\mass>0$ has the same role as in the (one-population) 3BP rate measure; namely, it linearly scales
the expected number of total variants for finite $\vecsamplesize$.

In the one-population case, the number of unique variants grows according to a power law as a function of the number of samples \citep{Broderick2012, Masoero2022}, and the single \emph{rate} parameter $\rate{}$ controls the power-law rate. In the two-population model with our new rate measure, growth of the number of variants has a more complex dependence on the number of samples since $\vecsamplesize$ can vary by different amounts in each population. But the rate parameters $\rate{1}, \rate{2} \in (0,1)$ can again be seen to control a notion of power-law growth in each population; we elaborate on the details in \cref{sec:powerlaw}.

Analogous to the one-population case, the \emph{concentration} parameters $\conc{1}, \conc{2} > 0$ modulate the frequencies of common variants; in particular, as $\conc{\popidx}$ grows large, population $\popidx$ is less likely to have variants whose frequencies are close to 1.

The \emph{correlation} parameters $\corr{1}, \corr{2} > 0$ do not have an analogue in the single-population case, since they control the relationship of the variant frequencies between the two populations. The factor $(\variantfreqperpop{1}+\variantfreqperpop{2}^{\rate{2}/\rate{1}})^{-\rate{1}} / (\variantfreqperpop{1}+\variantfreqperpop{2})^{\corr{1}+\corr{2}}$ distinguishes our proposed rate measure in \cref{eq:proposed_prior} from the factorized rate measure and induces correlation between the frequencies of the two populations. We illustrate how various settings of $\corr{1}, \corr{2}$ affect the rate measure in \cref{fig:log_rate_sigma}.

\begin{figure}
    \centering
    \includegraphics[width = \linewidth]{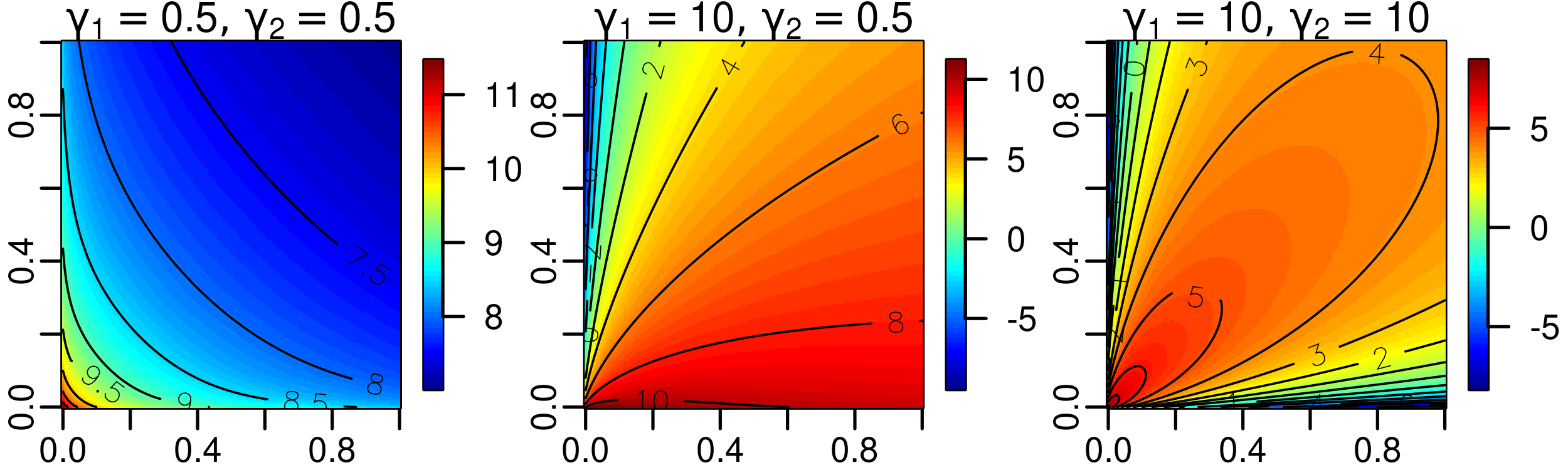}
    \caption{Visualization of the log density of our proposed rate measure in \cref{eq:proposed_prior} with different values of the $\corr{1},\corr{2}$ hyperparameters. In all plots, the other hyperparameters are $\rate{1}=\rate{2}=0.2, \conc{1}=\conc{2}=1$, and $\mass = 1$.}
	\label{fig:log_rate_sigma}
\end{figure}

We next verify that the hyperparameter ranges here are sufficient to ensure our desiderata.\footnote{An interesting line of investigation for future work would be to establish hyperparameter ranges that are both sufficient and necessary.}
\begin{myProposition}
If the hyperparameters satisfy $\mass, \conc{1}, \conc{2}, \corr{1}, \corr{2}>0$ and $\rate{1}, \rate{2} \in (0,1)$, the two-population model (\cref{eq:generative_model}) satisfies \cref{req:finite_mean,req:inf_mass} when paired with our rate measure $\ratemeasproposed$ from \cref{eq:proposed_prior}.
\end{myProposition}
See \cref{sec:proper_proof} for a proof. 

A common challenge in BNP is computation with the countably infinite set of parameters implied by \cref{req:inf_mass}.
Like typical models used in practice, though, our model enjoys conjugacy. Namely, analogous to the one-population model from \citet{Masoero2022}, our rate measure in \cref{eq:proposed_prior} is an exponential family conjugate prior for the Bernoulli process likelihood in the sense of \citet{Broderick2018}. We thereby avoid any infinite-dimensional integrals. See \cref{app:conjugacy} for a formal account and proof.

%\footnote{LM: we can use latex to work within e.g. matploltib --- https://stackoverflow.com/questions/46698921/latex-and-text-in-matplotlib-title; I would also explicity label both axes}

\section{Prediction} 
\label{sec:prediction_all}
With our model in hand, we now show how to predict the number of new variants in a follow-up study. We first consider the case where the hyperparameters are known (\cref{sec:prediction_fixed}). Then we show how to choose the hyperparameters in \cref{sec:hyperparameter_empirics}.

\subsection{Predicting variant numbers for fixed hyperparameters}
\label{sec:prediction_fixed}
First, suppose the model hyperparameters are fixed and known.

\begin{myProposition}
	%[Predicting $\genericfreqnews$ under near conjugate prior] 
	\label{prop:predict_kr_tons}
	Take a two-population model (\cref{eq:generative_model}) with our proposed PPP rate measure (\cref{eq:proposed_prior}).
	%
	%Assume the sampling model of \Cref{eq:generative_model}, where $\btheta_{\variantidx}$ are from a Poisson point process with {\Levy} measure given by \Cref{eq:proposed_prior} with known hyperparameters.
	Consider pilot data $\pilotdata$ for two populations, with pilot size $\vecsamplesize$. Consider a follow-up study with $\vecfuturesamplesize$ additional samples.
	Take any $\vecoccurrence$ such that, for each $\popidx \in \{1,2\}$, $\occurrence{\popidx} \le \futuresamplesize{\popidx}$.
	Then the number of new variants appearing exactly $\vecoccurrence$ times in the follow-up study satisfies
	\begin{align*}
		\genericfreqnews \mid \pilotdata \sim \Poisson\left(\genericfreqnewsparam\right),
	\end{align*}  
	where 
	\begin{align}
    		\begin{split}
		\genericfreqnewsparam := {} &\mass \binom{\futuresamplesize{1}}{\occurrence{1}} \binom{\futuresamplesize{2}}{\occurrence{2}}\frac{B(\corr{1}+\occurrence{1},\conc{1}+\samplesize{1}+\futuresamplesize{1}-\occurrence{1})B(\corr{2}+\occurrence{2},\conc{2}+\samplesize{2}+\futuresamplesize{2}-\occurrence{2})}{B(\corr{1},\conc{1})B(\corr{2},\conc{2})}\\
    				& {} \cdot \E_{Z,W}\left[\frac{(Z+W^{\rate{2}/\rate{1}})^{-\rate{1}}}{(Z+W)^{\corr{1}+\corr{2}}}\right].
    		\end{split} \label{eqn:kr_tons}
	\end{align}
$\E_{Z,W}$ denotes the expectation with respect to independent random variables $Z$ and $W$:
\begin{equation*}
    Z\sim \Beta(\corr{1}+\occurrence{1},\conc{1}+\samplesize{1}+\futuresamplesize{1}-\occurrence{1}), \quad \text{and} \quad
    W\sim \Beta(\corr{2}+\occurrence{2},\conc{2}+\samplesize{2}+\futuresamplesize{2}-\occurrence{2}).
\end{equation*}
\end{myProposition}
The proof appears in \cref{sec:kton_proof}. Like the one-population case \citep{Masoero2022},
the predictive distribution is Poisson. Unlike the one-population case,
the Poisson parameter here is not available in closed form (\cref{eqn:kr_tons}).
%One can see that the parameter $\genericfreqnewsparam$ is not expressed in closed form. 
%The fact that $\genericfreqnewsparam$ is expressed through an integral representation lacking a closed-form expression is a direct consequence of the cross-population coupling between frequencies introduced in \Cref{eq:proposed_prior}. 
In practice, we approximate the expectation in the prediction using the double exponential quadrature method \citep{takahasi1974double} provided in the Python package \texttt{mpmath} \citep{mpmath}.\footnote{We opted for quadrature since naive Monte Carlo struggles with the divisor in the integrand for $Z$ and $W$ near 0. We
expect more efficient choices are available with further work.}

As discussed in \cref{sec:setup}, the total number of new variants in the follow-up study can be recovered from the number
of new variants occurring $\vecoccurrence$ times. Relative to the separate proposal after \cref{eq:news_freq}, 
the following result provides a more computationally efficient way to 
recover the total number of new variants after computing the distributions in \cref{prop:predict_kr_tons}.
\begin{myProposition}
%[Predicting number of new variants] 
\label{coro:new_vars}
	Take the same assumptions as in \Cref{prop:predict_kr_tons}. Then the total number of new variants in a follow-up study satisfies
\begin{align} \label{eqn:faster_total_count}
     \genericnews \mid \pilotdata \sim \Pois\left(\sum_{\followupidx{1}=1}^{\futuresamplesize{1}}\newsparam{(\samplesize{1}+\followupidx{1}-1,\samplesize{2})}{(1,0),(1,0)}+\sum_{\followupidx{2}=1}^{\futuresamplesize{2}} \newsparam{(\samplesize{1}+\futuresamplesize{1},\samplesize{2}+\followupidx{2}-1)}{(0,1),(0,1)} \right),
\end{align} 
where $\genericfreqnewsparam$ is defined in \Cref{eqn:kr_tons}.
\end{myProposition}

%\footnote{old: Intuitively, this results come from a hypothetical sampling scheme that samples were taken one by one (i.e. a follow up of size $[0,1]$ or $[1,0]$) and predict number of new variants (necessarily 0-1 tons or 1-0 tons) in that individual not seen in previous samples (e.g. for individual $\followupidx{1}$ in population 1, the previous sample has size $(\samplesize{1}+\followupidx{1}-1,\samplesize{2})$).} 
See \Cref{sec:proof_total_number} for the proof. Intuitively, we first imagine taking a single new draw at a time from population 1 ($\vecfuturesamplesize = (1,0)$) and then a single new draw at a time from population 2 ($\vecfuturesamplesize = (0,1)$). Any variant that is new in the follow-up relative to the pilot must be new at exactly one step of this scheme.

The naive proposal in \cref{sec:setup} would require summing over all possible $\vecoccurrence$ such that $\occurrence{1} + \occurrence{2} \ge 1$. We know that, for each population, $\occurrence{\popidx} \le \futuresamplesize{\popidx}$. If we sum up to this bound, we would require $\futuresamplesize{1} * \futuresamplesize{2} - 1$ evaluations of \cref{eqn:kr_tons}. By contrast, \cref{eqn:faster_total_count} requires just $\futuresamplesize{1} + \futuresamplesize{2}$ total evaluations of \cref{eqn:kr_tons}. In practice, we expect that $\genericfreqnewsparam$ becomes negligible for larger $\occurrence{\popidx}$, so one might employ a truncation to speed up the naive scheme. But figuring out an appropriate truncation level would require some work and involve a tradeoff between computation and accuracy.

%each summand represents exactly one new draw from the respective population (i.e., $\vecfuturesamplesize = (0,1)$ or $\vecfuturesamplesize = (1,0)$). For such an extrapolation size, the number of new variants need to be necessarily singletons (i.e., variants that will appear exactly once in the complete sample).
%Although in practice we make predictions with posterior means, both~\Cref{prop:predict_kr_tons} and~\Cref{coro:new_vars} give the full posterior distributions that come with posterior uncertainty measures.\footnote{LM: not sure we need this remark?}

%\begin{proof}
%    The result directly follows from properties of independent Poisson random variables.
%\end{proof}

\subsection{Empirics for Hyperparameters}
\label{sec:hyperparameter_empirics}
In practice, the hyperparameter values are unknown, and we must learn them from the observed data.
We take an analogous empirical-Bayes approach to the one-population case of \citet{Masoero2022}.

In particular, we start by splitting the pilot data into a training set and a test set. Uniformly at random among the
pilot data, we choose $\vecsamplesize_\train =(\samplesize{1,\train}, \samplesize{2,\train})$ samples to form the training set.
Here $\samplesize{\popidx,\train}$ must satisfy $0 < \samplesize{\popidx,\train} < \samplesize{\popidx}$.
The remaining $\vecsamplesize - \vecsamplesize_\train$ samples form the test set. 
We now treat our training set as a miniature pilot and our test set as a miniature follow-up. 
For the vector of appearance counts $\vecoccurrence$, let $\newsparam{\vecsamplesize_\train}{\vecsamplesize - \vecsamplesize_\train, \vecoccurrence}$ 
be the Poisson parameter from \cref{prop:predict_kr_tons}. And let $\obsnews{\vecsamplesize_\train}{\vecsamplesize - \vecsamplesize_\train, \vecoccurrence}$
be the observed number of new variants appearing $\vecoccurrence$ times in the two populations; we use the lower
case character $\obsnewsNoargs$ to distinguish from the predictive random variable above.

Roughly, we will aim to choose hyperparameters that maximize the probability of observing 
$\obsnews{\vecsamplesize_\train}{\vecsamplesize - \vecsamplesize_\train, \vecoccurrence}$. 
More precisely, we consider $\occurrence{\popidx}$ in the miniature follow-up study up to a maximum count $v \in \NN$.
The log probability across the corresponding $\vecoccurrence$ terms can be written as
follows, where we emphasize the dependence on the hyperparameters $\allhyper := (\mass, \rate{1},\rate{2}, \corr{1},\corr{2}, \conc{1},\conc{2})$:
\begin{equation}
    \label{eq:likelihood_prior}
    \mathcal{L}(\allhyper)=
    	\sum_{\substack{\vecoccurrence: \occurrence{1},\occurrence{2} \in 0:v \\ \occurrence{1} + \occurrence{2} \ge 1}}
		\log \left[ \pois \left(
				\obsnews{\vecsamplesize_\train}{\vecsamplesize - \vecsamplesize_\train, \vecoccurrence}
				\mid \newsparam{\vecsamplesize_\train}{\vecsamplesize - \vecsamplesize_\train, \vecoccurrence}(\allhyper)
			\right) \right]	
%    	\sum_{\occurrence{1}=0}^v \sum_{\occurrence{2} = 0}^v  \log\left[
%			\pois\left( \obsnews{\vecsamplesize_\train}{\vecsamplesize - \vecsamplesize_\train, \vecoccurrence} \mid \newsparam{\vecsamplesize_\train}{\vecsamplesize - \vecsamplesize_\train, \vecoccurrence}\right)
%			\right]  1\left\{(\occurrence{1},\occurrence{2})\neq (0,0) \right\},
\end{equation}

We set $v$ to balance our goals of accurate prediction with computational cost. Each term in the sum requires
computing \cref{eqn:kr_tons}, and there are $(v+1)^2-1$ terms in the sum. We find empirically that setting $v=10$ gives 
results comparable to larger settings of $v$ and therefore use $v=10$ in our experiments; see \cref{sec:large_triming}
for our experimental comparison of $v$ levels.

We choose the training set size as $\samplesize{\popidx,\train}=\lfloor \samplesize{\popidx}/2\rfloor$.
We expect a larger training set size to more closely represent the actual pilot study. However, recall that we focus on 
the challenging case where the size of the pilot study is small. A larger training set size yields a smaller test size, which
can yield a very noisy estimate of test error if too small. Our choice is meant to balance between these two extremes.

In our experiments, we optimize $\mathcal{L}(\allhyper)$ across the hyperparameters $\allhyper$
using L-BFGS in \texttt{scipy} with default parameters
and gradients estimated numerically. We initialize at $\allhyper_{\text{init}} =
(\mass, \rate{1},\rate{2}, \corr{1},\corr{2}, \conc{1},\conc{2})_{\text{init}} = (10^3, 0.5,0.5,0.5,0.5,1,1)$.
Local optima and the accuracy of our gradient approximations are concerns. But our experiments
suggest that our procedure is able to return useful predictions.
The computational bottleneck is the many evaluations of the integral in \cref{eqn:kr_tons}. Currently the overall
optimization procedure can take several hours; reducing this cost is a promising area for future work.

\section{Power law asymptotics for the number of new variants}
\label{sec:powerlaw}
% !TEX root = ../double_trouble_main.tex

Past work suggests that, in a single population, we might expect the observed number of variants to grow as a power law of the number of samples \citep{Masoero2022}. More generally, existing BNP models typically either exhibit logarithmic growth or power-law growth, each asymptotic in the number of samples.\footnote{For instance, under a 3BP with strictly positive discount parameter and Bernoulli likelihood, the number of variants grows (asymptotically) as a power law of the number of samples \citep{Gnedin2007, Broderick2012}. Under a 3BP with zero discount, the number of variants grows logarithmically in the number of samples \citep{Broderick2012}. Likewise, a Dirichlet process with categorical likelihood yields logarithmic growth in the number of observed clusters, but a Pitman-Yor Process exhibits power-law growth \citep{Gnedin2007,teh2009indian}.} We here rule out logarithmic growth under our model (with fixed hyperparameters) and establish results that are suggestive of power-law growth.

Power-law behavior has widely been recognized as desirable in BNP models \citep{teh2009indian, lee2016bfry}. But past work establishing power laws in these models has focused on the single-population case \citep{Gnedin2007, teh2009indian, Broderick2012,lee2016bfry}. A number of new challenges arise when we consider two populations. A first challenge is specifying how samples are collected across the two populations. In one population, there is a single way to grow the number of samples; in two populations, there are many more possibilities.
We focus on two sampling schemes: 
\begin{itemize}
	\item the \emph{projection scheme}, where we sample from only one of the two populations
	\item the \emph{proportional scheme}, where we collect samples from both populations according to some fixed proportion $\propt > 0$; in particular, after $\propt$ is chosen, we have $\samplesize{2} =  \lceil\propt \samplesize{1}\rceil$.
\end{itemize}
In what follows, $\zerovec = (0,0)$ is a vector of zeros, and $a_\samplesizenosubscript \sim b_\samplesizenosubscript$ indicates $\lim_{\samplesizenosubscript\to\infty} a_\samplesizenosubscript/b_\samplesizenosubscript=1$.

\subsection{The projection scheme} \label{sec:projection_powerlaw}

First we establish power-law bounds for the expected rate of growth of the number of variants under the projection scheme. We present symmetric results for each population, so we will find it useful to introduce the notation $\notpopidx$ to indicate the population that is not $\popidx$; i.e., when $\popidx = 1$, we have $\notpopidx = 2$ and vice versa.
%%%
\begin{myProposition}[Power-law bounds under the projection scheme] 
\label{prop:power_law_proj}
	Fix hyperparameters $\mass, \conc{1}, \conc{2}, \corr{1}, \corr{2}>0$ and $\rate{1}, \rate{2} \in (0,1)$.
	Consider data drawn from the two-population model of \cref{eq:generative_model} with $\levync = \ratemeasproposed$ (\cref{eq:proposed_prior}).
	Then the number of new variants concentrates around its expectation: 
	\[
		\textrm{Fix } \samplesize{\notpopidx} = 0. \quad \news{\zerovec}{\vecsamplesize} \sim \E \news{\zerovec}{\vecsamplesize}\textrm{ a.s.\ as } \samplesize{\popidx} \to \infty.
	\]
	Moreover, there exist functions $\lowerboundproject, \upperboundproject$ that serve as respective upper and lower bounds on $\E \news{\zerovec}{\vecsamplesize}$,
	\begin{equation} 
		\label{eq:projection_bounds}
		\begin{aligned}
			&\forall \vecsamplesize, \quad \lowerboundproject(\vecsamplesize)
			\le \E \news{\zerovec}{\vecsamplesize}
			\le \upperboundproject (\vecsamplesize),
		\end{aligned}
	\end{equation}
	and that show the asymptotic behavior below.
	\paragraph{Upper bound} For either $\popidx \in \{1,2\}$ and for any $\smallpositive$ such that $0<\smallpositive< \min\{\corr{1},\corr{2}\}$, there exists a constant $\constantupper{\popidx} > 0$ not depending on $\vecsamplesize$ such that:
	\begin{equation}
		\label{eq:power-law_projection-upperbound}
		\textrm{Fix } \samplesize{\notpopidx} = 0. \quad  \upperboundproject (\vecsamplesize)
			\sim \constantupper{\popidx} \samplesize{\popidx}^{\rate{\popidx}+\smallpositive} \textrm{ as } \samplesize{\popidx} \to \infty.
	\end{equation}
	\paragraph{Lower bound} Here let $\popidx$ be the population such that $\rate{\popidx} \le \rate{\notpopidx}$.
	Define $\rate{\notpopidx}' := \rate{\notpopidx}-\corr{\popidx}(\rate{\notpopidx}/\rate{\popidx}-1)$. 
	Then there exist constants $\constantlower{\popidx}, \constantlowerwithotherrate{\popidx}, \constantlowerwithrateprime{\popidx} > 0$
	not depending on $\vecsamplesize$ such that:
	\begin{equation}
	\label{eq:power-law_projection-lowerbound}
		\textrm{Fix } \samplesize{\notpopidx} = 0. \quad  \lowerboundproject(\vecsamplesize) \sim \constantlower{\popidx} \samplesize{\popidx}^{\rate{\popidx}}
				\textrm{ as } \samplesize{\popidx} \to \infty.
	\end{equation}
	and
	\begin{equation}
	\label{eq:power-law_projection-max_lowerbound}
		\textrm{Fix } \samplesize{\popidx} = 0. \quad  \lowerboundproject(\vecsamplesize)
			\sim \max\{ \constantlowerwithotherrate{\popidx} \samplesize{\notpopidx}^{\rate{\popidx}}, \constantlowerwithrateprime{\popidx}\samplesize{\notpopidx}^{\rate{\notpopidx}'} \}
				\textrm{ as } \samplesize{\notpopidx} \to \infty.
	\end{equation}
\end{myProposition}

See \cref{sec:projection_scheme_proof} for the proof.
Note that, due to the maximum in \cref{eq:power-law_projection-max_lowerbound}, the lower bound will grow (asymptotically) at least as fast as 
$\samplesize{\notpopidx}^{\rate{\popidx}}$ as $\samplesize{\notpopidx} \to \infty$. It follows that the lower bound is (asymptotically) positive and diverging.

In the one-population 3BP, it is possible to establish almost sure power-law growth in the number of variants as a function of the number of samples \citep{Broderick2012, teh2009indian}. By contrast, here we are only able to provide power-law \emph{bounds} on the growth. Nonetheless our analysis does conclusively rule out logarithmic growth.

The gap between our upper and lower bounds depends on the population. If we take the population with lower rate parameter ($\rate{\popidx} \le \rate{\notpopidx}$), the gap between the powers in the upper and lower bounds (\cref{eq:power-law_projection-upperbound,eq:power-law_projection-lowerbound}) can be made arbitrarily close by taking $\smallpositive$ to 0. If we take the population with higher rate parameter, we generally expect a non-trivial gap between the powers in the upper and lower bounds (\cref{eq:power-law_projection-upperbound,eq:power-law_projection-max_lowerbound}). We highlight the special case where the rate parameters are equal next.

\begin{myCorollary}
    \label{coro:projection_power-law-rate-align}
    Take the same assumptions and constants as \cref{prop:power_law_proj}. Take $\rate{1}=\rate{2}=\rate{}$. Let $\popidx \in \{1,2\}$. Fix $\samplesize{\notpopidx} = 0$. Then
    	\[
		\constantlower{\popidx}\samplesize{\popidx}^{\rate{p}} \sim \lowerboundproject(\vecsamplesize) \le \E \news{\zerovec}{\vecsamplesize} \le \upperboundproject(\vecsamplesize) \sim \constantupper{\popidx} \samplesize{\popidx}^{\rate{p}+\smallpositive},
	\]
	where the asymptotic equivalences hold as $\samplesize{\popidx} \to \infty$.
\end{myCorollary}
\begin{proof}
   The result follows from \cref{eq:power-law_projection-upperbound,eq:power-law_projection-lowerbound}. \Cref{eq:power-law_projection-max_lowerbound} simultaneously applies -- with the roles of $\popidx$ and $\notpopidx$ switched -- and offers no contradiction.
\end{proof}

Despite the asymmetry of our PPP rate measure in \cref{eq:proposed_prior}, we do obtain symmetric roles for the rate
parameters across the two populations when $\rate{1}=\rate{2}$; this symmetry inspires our use of the name \emph{rate parameter} in analogy
to the one-population case.

We next simulate data from our model under the projection scheme, separately in each population (\cref{fig:power-law}); our simulations are compatible with power-law growth and with the bounds from our theory. For instance, when we grow population 1, we find that for larger $\samplesize{1}$ the number of variants grows roughly as a constant times $\samplesize{1}^{0.34}$. Since $\rate{1} = 0.3 < 0.6 = \rate{2}$, our bounds suggest growth between $\samplesize{1}^{0.3}$ and $\samplesize{1}^{0.3 + \smallpositive}$. We expect the small discrepancy is due to the finite sample. Similarly, when we grow population 2, we find that for larger $\samplesize{2}$, the number of variants grows roughly as a constant times $\samplesize{2}^{0.49}$. Since $\rate{2}' = 0.6 - 0.1(0.6/0.3-1) = 0.5$, our bounds suggest growth between $\samplesize{1}^{0.5}$ and $\samplesize{1}^{0.6 + \smallpositive}$. Again, we expect the small discrepancy is due to the finite sample size.

\begin{figure}[htp]
    \centering
    \includegraphics[width = 0.6\textwidth]{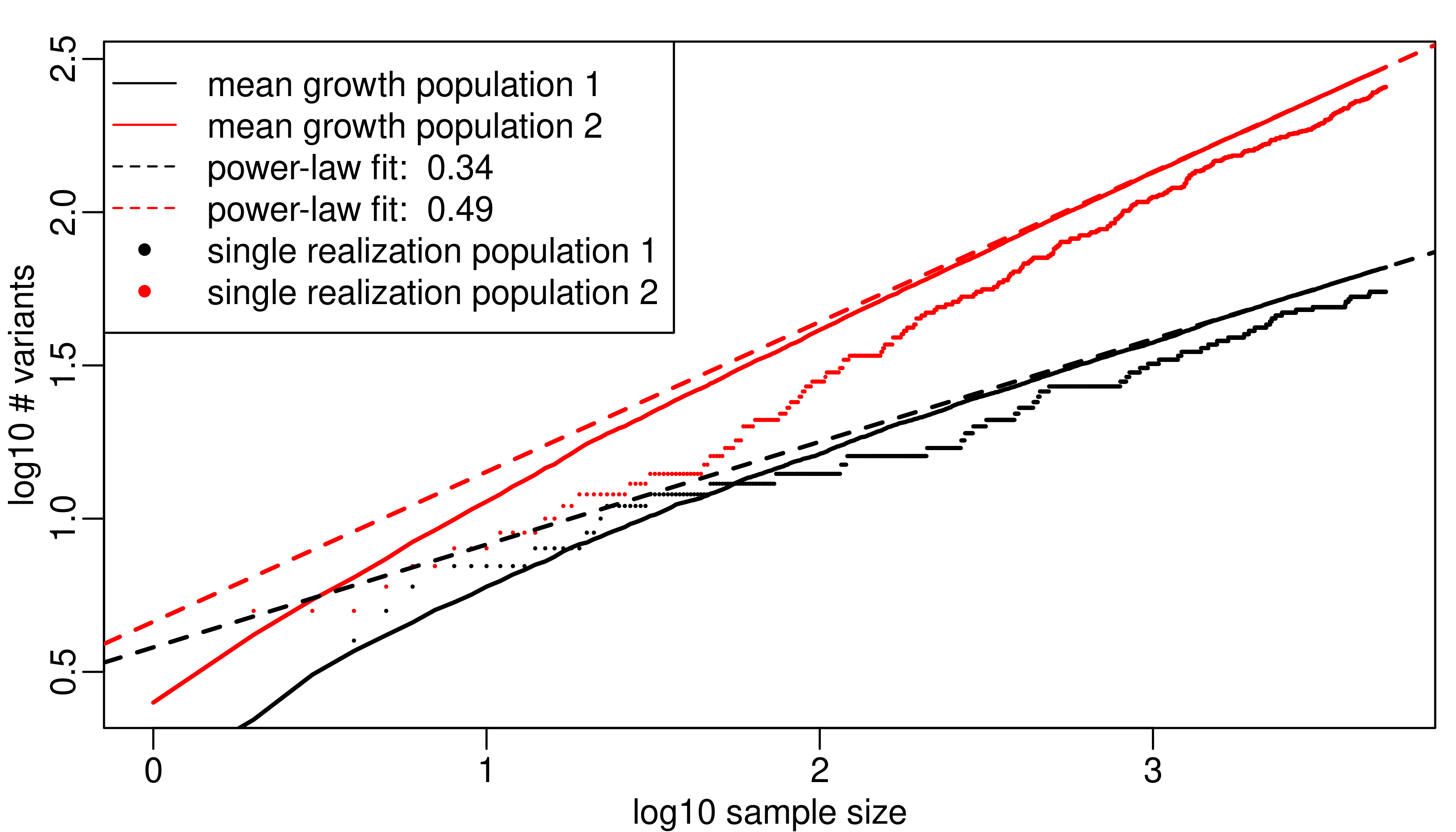}
    \caption{ \label{fig:power-law} Illustration of simulated data from our model with $\mass = 10, \rate{1}=0.3, \rate{2}=0.6, \conc{1}=\conc{2}=1, \corr{1}=\corr{2}=0.1$. Under the projection scheme, we take samples from only one population (population 1 in black and 2 in red). The observed number of variants for a given sample size is plotted as a point. To generate a solid curve, we take 100 independent realizations and plot the mean number of observed variants at a particular sample size (and interpolate between sample sizes). To plot the power-law fit, we make a linear least square fit of the log-log plot of the final third of the mean-growth curve and plot the resulting line. %note: log-log plot comes from OLS with both slope and offset parameter
	The fitted power for population 1 is 0.34 and for population 2 is 0.49.}
\end{figure}

%%%
\subsection{The proportional scheme}

We next establish power-law bounds for the expected rate of growth of the number of variants under the proportional scheme. Recall that under the projection scheme, the total number of samples was equal to the number of samples in one of the populations. Under the proportional scheme, the total number of samples will instead be $\samplesize{1} + \samplesize{2} = \lceil (1+\rho) \samplesize{1} \rceil$. We will control the asymptotics by taking $\samplesize{1} \rightarrow \infty$.
 %%%
\begin{myProposition}[Power-law bounds under the proportional scheme] 
    \label{coro:proportional}
	Take the model and hyperparameters as in \Cref{prop:power_law_proj}. Let $\popidx$ be such that $\rate{\popidx} \le \rate{\notpopidx}$, 
	%Assume $\rate{\notpopidx}' := \rate{\notpopidx}-\corr{\popidx}(\rate{\notpopidx}/\rate{\popidx}-1)\in (0,1)$. \textcolor{red}{TB: why is this assumption necessary here but not in the projection case? also, text needs to comment on the extent to which this assumption is reasonable. i see there was an old comment about the need to say what this assumption means, but it seems like it hasn't been addressed} /yunyi{we don't need this anymore, it's an old piece BEFORE another lower bound was found. Recall we have the maximum structure now, we have only this one once, and without this assumption we cannot say "power law" anymore. Thanks for catching it.}
	For any $\samplesize{1}$, let $\samplesize{2} = \lceil\propt \samplesize{1}\rceil$.
	Then the number of new variants concentrates around its expectation: 
	\[
		\news{\zerovec}{\vecsamplesize} \sim \E \news{\zerovec}{\vecsamplesize}\textrm{ a.s.\ as } \samplesize{1} \to \infty.
	\]
	 Moreover,
	there exist functions $\lowerboundproport, \upperboundproport$ that serve as respective upper and lower bounds on $\E \news{\zerovec}{\vecsamplesize}$,
    \begin{equation*}
        \forall \vecsamplesize,\quad \lowerboundproport(\vecsamplesize)
        		\le \E \news{\zerovec}{\vecsamplesize}
		\le \upperboundproport(\vecsamplesize),
    \end{equation*}
	and that exhibit the asymptotic behavior we describe next. 
	Define $\rate{\notpopidx}'$ as in \Cref{prop:power_law_proj}.
	For any $\smallpositive$ such that $0<\smallpositive<\min\{\corr{1},\corr{2}\}$, there exist constants $C,C'_{\smallpositive} > 0$ not depending on $\vecsamplesize$ such that 
    \begin{align*}
	\lowerboundproport(\vecsamplesize)
		&\sim C(\samplesize{1}+\samplesize{2})^{\max\{\rate{\popidx},\rate{\notpopidx}'\}}
			\textrm{ \ as } \samplesize{1} \rightarrow \infty \\
    	\textrm{and } \quad \upperboundproport(\vecsamplesize)
		&\sim C'_{\smallpositive}(\samplesize{1}+\samplesize{2})^{\max\{\rate{\popidx},\rate{\notpopidx}\}+\smallpositive} 
			\textrm{ \ as } \samplesize{1} \rightarrow \infty.
    \end{align*}
\end{myProposition}
See \cref{sec:proportion_scheme_proof} for a proof. Essentially we lower bound the total number of variants by the number seen in either population. 
We upper bound the total number of variants by the sum of variants seen in each population; this bound may be somewhat loose since some variants might be shared across the two populations, and hence double counted.  
%For the individual populations, we apply the bounds from \cref{prop:power_law_proj}. % LOM <<< not sure what we mean by this here? we have the projection scheme for that in place already. seems redundant.

We next present the special case where the two rate parameters are equal; in this case, the powers in the lower and upper bounds align up to an arbitrary $\smallpositive > 0$ difference.
%%%
\begin{myCorollary}
    \label{coro:proportional_power-law-rate-align}
    Take the same assumptions and constants as \cref{coro:proportional}. Take $\rate{1}=\rate{2}=\rate{}$. Then
    \begin{equation*}
        C(\samplesize{1}+\samplesize{2})^{\rate{}} \sim \lowerboundproport(\vecsamplesize)
        		\le \E \news{\zerovec}{\vecsamplesize}
		\le \upperboundproport(\vecsamplesize) \sim C'_{\smallpositive}(\samplesize{1}+\samplesize{2})^{\rate{}+\smallpositive} \textrm{\ as } \samplesize{1} \to \infty.
    \end{equation*} 
    %where $C'\ge 2^{-\rate{}}\min\{C_1,C_2\}$  and $C'\le 2\max\{C_1,C_2\}$ with $C_1$, $C_2$ being the constants in~\eqref{eq:power-law-pop1} and~\eqref{eq:power-law-pop2} respectively.
    %\textcolor{red}{TB: make sure you understand why these limits have to be stated as N1 goes to infinity and NOT as N1 and N2 go to infinity}\yunyi{N2 is a function of N1}
\end{myCorollary}

The result follows directly from \cref{coro:proportional} and noting that, when the rate parameters are equal, we have $\rate{\notpopidx}' = \rate{\notpopidx} = \rate{}$.

\section{Experiments}
\label{sec:experiments}
% !TEX root = ../double_trouble_main.tex

We next check the quality of our predictions experimentally. First we check that our method performs well on well-specified synthetic data, as well as some particular forms of misspecification. Then we see that our method provides high-quality predictions on a real human cancer genetics dataset and a real human genomic dataset. Throughout, we illustrate the advantages of modeling two populations together rather than treating the two populations separately or all from a single population.

\paragraph{Experimental setup.} In all of our experiments, we will have access to a collection of genetic samples for each of two populations (e.g., two types of cancer). Our validation procedure is reminiscent of cross-validation and analogous to the single population validation procedures of \citet{Masoero2022,  zou2016quantifying}. Within each population, we uniformly at random choose a partition of the data into 20 equally-sized blocks.\footnote{Since the dataset size may not be a multiple of 20, block sizes in practice may vary by 1. We will elide  small details of this form in the remainder of this section for clarity of exposition. Also, the population genetics data is substantially larger than our other datasets, so -- to avoid a very large pilot size -- we use 30-fold cross-validation in that case (\cref{sec:gnomad}).} In each of 20 folds, we reserve one block (unique to this fold) to be the pilot data for its corresponding population, so the final pilot sample has size $\vecsamplesize = (\samplesize{1}, \samplesize{2})$. The remaining data will serve as potential follow-up data. Each method has access only to the pilot data (across both populations), and we use the follow-up data (across both populations) to evaluate the accuracy of each method. None of the methods we consider in this paper use ordering information on the pilot data, but to generate all of our figures, we choose an ordering on the pilot data uniformly at random. Likewise, we choose an ordering on the follow-up data uniformly at random. In our figures, we present first the pilot data from population 1, then the pilot data from population 2, then the follow-up data from population 1, then the follow-up data from population 2. We can think of these figures as encapsulating a series of evaluations as follows. We train every method on the full pilot data in populations 1 and 2. Then we progressively test on follow-up sample sizes $\vecfuturesamplesize = (\futuresamplesize{1},\futuresamplesize{2})$ of the following forms: $(1,0), (2,0), \ldots, (19 \samplesize{1},0), (19 \samplesize{1},1), (19 \samplesize{1},2), \ldots, (19 \samplesize{1}, 19 \samplesize{2})$. In this way, we capture a wide range of ratios between follow-up sample sizes across the two populations: from all follow-up data in a single population to a ratio of follow-up sizes that matches the ratio in the pilot.
To summarize our results across all 20 folds, we will show an empirical mean across folds, both for each prediction method as well as for the ground truth. And we will show a shaded one-standard deviation interval for the prediction methods. The ground truth showed very small variation across folds so the interval appears almost invisible. %\tb{ultimately we will say something about ground truth here too once the results are in.} %\tb{Are the intervals for the ground truth very small? can we say something here?}

Recall that we are also interested in predicting the number of $\vecoccurrence$-tons in a follow-up sample of size $\vecfuturesamplesize$ (\cref{eq:news_freq}); that is, we are interested in predicting the number of new variants appearing exactly $\vecoccurrence$ times across the two populations. To evaluate our prediction in this case, we consider just the extreme where we use all possible data in the follow-up. For each value of $\vecoccurrence$ with $\occurrence{1}, \occurrence{2}$ each less than a threshold, we compute the \emph{relative residual}: namely, the signed difference between a method's prediction and the observed value in the reserved follow-up data, normalized by the observed value. We plot an empirical average of this quantity across the folds. We also report the actual observed values and absolute (un-normalized) residuals in \cref{app:additional_experiments}, again averaged across the folds.

\paragraph{Comparisons.} To the best of our knowledge, there are no alternative methods for the prediction task we are interested in here. However, there are a number of methods that predict the number of variants in a follow-up given pilot data in a single population~\citep{Masoero2022, chakraborty2019using,gravel2014predicting}. So we consider two natural competitors in our experiments. First, we define the \emph{independent} two-population extension of a single-population method as follows. Given $\vecsamplesize$ pilot samples and $\vecfuturesamplesize$ follow-up samples, we train the single-population method first on the $\samplesize{1}$ pilot samples from population 1 and predict the number of variants expected in $\futuresamplesize{1}$ follow-up samples from population 1. We separately train the single-population method on the $\samplesize{2}$ samples from population 2 and predict the number of variants expected in $\futuresamplesize{2}$ follow-up samples from population 2. We report the sum of the predicted number of variants across both populations as the final prediction. In general, we expect the independent method to overcount the number of variants when there are follow-up samples from both populations. The independent method essentially assumes each variant is unique to a single population. From that perspective, when $\vecoccurrence = (\occurrence{1},0)$, we can take the population-1 prediction for the number of $\occurrence{1}$-tons and use it as the independent-method prediction for the number of $\vecoccurrence$-tons; an analogous prediction holds when $\occurrence{1} = 0$ and $\occurrence{2} > 0$. Whenever both elements of $\vecoccurrence$ are strictly positive, the independent method predicts there are zero $\vecoccurrence$-tons. In the real data we analyze at the sizes we consider, this zero prediction is unrealistic (see~\Cref{fig:proposed_kton},~\ref{fig:i3BP_kton},~\ref{fig:d3BP_kton} for raw number of $\vecoccurrence$-tons in synthetic experiments and~\Cref{fig:brca_luad_kton},~\ref{fig:seu_bgr_kton} in real data experiments). 

Second, we define the \emph{dependent} two-population extension of a single-population method. In this case, we train the single-population method on all of the $\samplesize{1} +\samplesize{2}$ pilot samples, as though they are the pilot samples from a single population. To form a prediction on a follow-up sample of size $\vecfuturesamplesize$, we report the prediction on a single-population follow-up of size $\futuresamplesize{1} + \futuresamplesize{2}$. In general, we expect the dependent method to struggle when both (a) the two populations exhibit notably different variant growth rates and (b) the ratio of populations in the follow-up does not closely mirror the pilot. To predict the number of $\vecoccurrence$-tons, we predict the number of variants, conditional on the pilot, that we would expect to see $\occurrence{1}$ times out of $\futuresamplesize{1}$ follow-up samples and also $\occurrence{2}$ times out of $\futuresamplesize{2}$ distinct follow-up samples (all from the same assumed single population). For single-population methods that do not provide $\occurrence{}$-ton predictions (for scalar $\occurrence{}$), we will therefore not be able to provide $\vecoccurrence$-ton predictions in the dependent extension. However, we are able to provide $\vecoccurrence$-ton predictions for the special case of the single-population 3BP method due to \citet{Masoero2022}; we provide the prediction and its supporting theory in \cref{sec:d3BP}.

Since there are many single-population methods, in general we will select just one single-population method to use when reporting the independent and dependent results. In each case, we will try to use the single-population method that is best in a relevant sense, which we make precise in each case below.

\subsection{Synthetic experiments}
\label{sec:simulation}
% !TEX root = ../double_trouble_main.tex

We first confirm that our method is able to predict well on data simulated from our own model (\cref{eq:proposed_prior}), as a check on our methodology before moving to real data. To that end, we generate 500 data points in each of two populations; we describe the simulation procedure in more detail in~\cref{app:taking_sample}. So, in each of our 20 folds, the pilot data has size $\vecsamplesize = (25,25)$.

We simulate two different datasets from our model (\cref{eq:proposed_prior}); we choose the hyperparameters to reflect a range of growth rates, sample sizes, and correlations between the two populations in each case. For the first dataset, we choose hyperparameters $\allhyper := (\mass, \rate{1},\rate{2}, \corr{1},\corr{2}, \conc{1},\conc{2}) = (100,0.4,0.6,0.5,0.5,1,1)$ and $\allhyper = (1,0.8,0.7,0.5,0.5,1,1)$ for a qualitative match of the curve form in the GnomADfor a qualitative match of the curve form in GnomAD dataset (\cref{fig:pop_gen}) and cancer dataset (\cref{fig:cancer}) respectively. %Despite the correlation parameters being the same, we do observe difference in level of sharing due to interaction between rate parameters and correlation parameters. The intuition is that with high rate parameters, most new variants are singletons, i.e., variants only appear \textit{once} which cannot be shared variants. %The reason behind the low level of sharing when rate parameters are high is that in those cases variants are likely to be singletons (exist only once) that cannot be shared. 
%\tb{I'm not following the sentence before this comment. what is it saying?}
 %\tb{It seems like the correlation parameters are the same across the two datasets, but there seem to be effectively no shared variants across the two populations in dataset 2. Is it because of the interaction between correlation parameters and other hyperparameters? Or something else? We should explain this phenomenon.}

The results in \cref{fig:proposed} demonstrate that our method outperforms two potential baselines. Since here all of our data is simulated from our two-population BNP model, we use a single-population BNP model to form our baselines; namely, we compare to the independent and dependent two-population extensions of the 3BP approach due to \citet{Masoero2022}. We see in the upper left plot in \cref{fig:proposed} that our method (solid red line) much more closely matches the behavior of the ground truth (dashed black line) than either the independent (dotted green line) or dependent (dash-dot blue line) extensions. As expected, the independent extension dramatically overcounts the number of variants in the second follow-up population. Also as expected, the dependent extension performs well when the pilot and follow-up proportions are in agreement (at the largest sample size) but misestimates the number of variants elsewhere. By contrast, our method is able to capture the different behaviors of the two populations and account for any shared variants. In the lower left corner, the independent extension of the 3BP performs the best; its mean line is closest to the ground truth, and its prediction exhibit the smallest variation (shaded range) across folds. Importantly, we can infer from the center and left plots that there are extremely few shared variants in this simulation; see~\cref{sec:raw_responses_simulated} for more precise counts. That is, the assumption of the independent extension is very close to being correct. So, following Bayesian Occam's Razor \citep[ch.~28 p.~343]{mackaybook}, we expect better performance from the simpler model (the independent extension) when its more stringent assumptions are true. That being said, our proposed model still performs very well here; its mean line is also very close to the ground truth, in contrast to the dependent extension. We conclude that our method is able to handle prediction of the number of new variants under a variety of two-population behaviors -- while the independent extension is limited to the special cases where the two populations are known a priori to behave independently.

From the center and right plots in \cref{fig:proposed}, we also see that our method provides better estimates of the numbers of various $\vecoccurrence$-tons. Recall that the independent extension predicts zero for anything except private $\vecoccurrence$-tons, and we see that there are in fact non-trivial counts of shared $\vecoccurrence$-tons in the first data set. Given this limitation of the independent extension, we focus on plotting results for our method and the dependent extension. Darker colors indicate more erroneous predictions, so we conclude in both datasets that our method is providing more accurate predictions than the dependent 3BP extension. Note that the observed numbers of $\vecoccurrence$-tons vary dramatically across $\vecoccurrence$ values. So generally we expect more noise in the upper right of our plots of $\vecoccurrence$-tons -- and should therefore take care to not to treat every square in these plots equally. We can see the actual observation values for the two datasets in~\cref{fig:proposed_kton}; in particular, for the first dataset, the observed number of $\vecoccurrence$-tons varies from an average (across folds) of 1161.2 at $(0,1)$ to 0 at various locations, e.g. (8,3), (10,4), and (7,7) (see the gray boxes in the upper left panel of \cref{fig:proposed_kton}).  We see similar behavior in the real data, as we describe below.

\begin{figure}[htp]
    \centering
    \includegraphics[width = 0.9\textwidth]{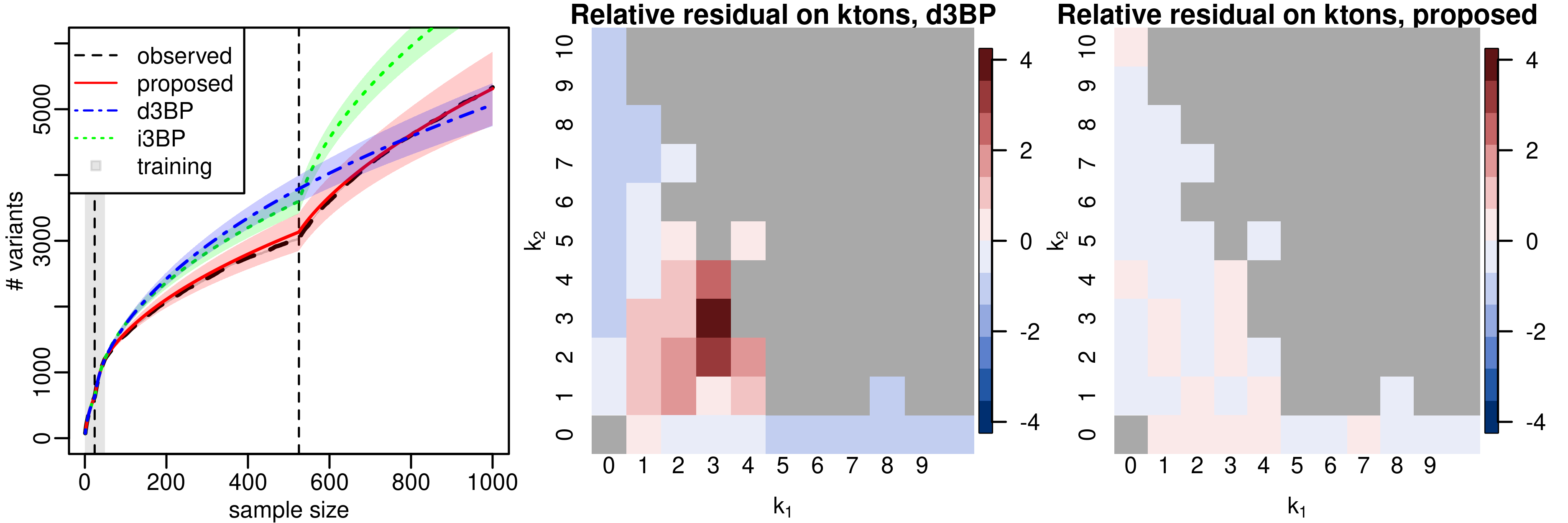}
    \includegraphics[width = 0.9\textwidth]{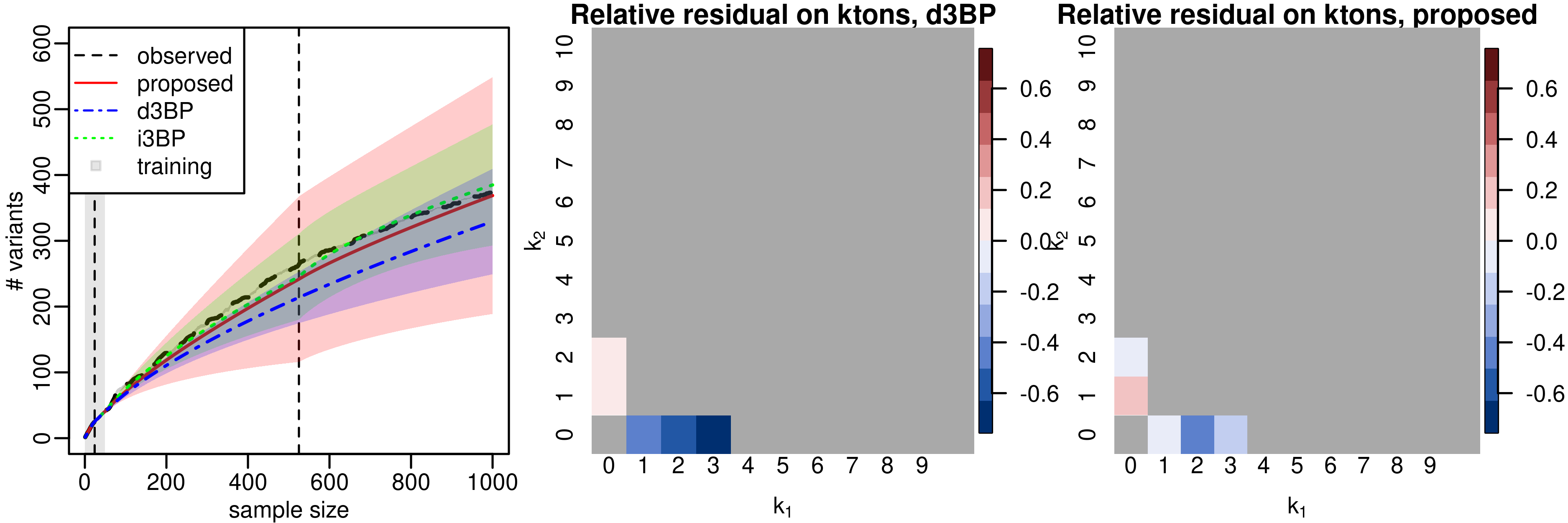}
    \caption{Predictions when data is sampled from our model (\cref{eq:proposed_prior}). Each row represents a different simulated dataset, as described in \cref{sec:simulation}. \emph{Left}: The dashed black line shows the observed number of variants as a function of the sample number. The samples appear in the following order: the pilot from population 1, the pilot from population 2, the follow-up from population 1, the follow-up from population 2. A vertical dashed gray line separates the two populations in the pilot and follow-up, respectively. The gray shading covers the pilot data. All curves agree exactly on the pilot data. In the follow-up region, we plot mean (lines) and 1 standard deviation intervals (shaded regions) from three methods: our proposed method (solid red), the dependent version of the single-population 3BP (dash-dot blue), and the independent version of the 3BP (dotted green). \emph{Center and Right}: For each $\vecoccurrence$ with component values up to 10, we plot the relative residual for predictions from our method (right) and the dependent 3BP (center). Note that the color scales are fixed across a row but vary across a column. A square is gray when the observed value is strictly less than 2 in at least one fold or $\vecoccurrence = (0,0)$, where by definition having 0 observed. %\tb{is it $<=$ or $<$? i know we talked about this before, but i'm not finding it quickly}\yunyi{strictly less is correct}
    }
    \label{fig:proposed}
\end{figure}

We provide further checks that our method is able to handle datasets generated from the (mis-specified) assumptions of the independent and dependent 3BP extensions in~\cref{sec:misspecified}.

\subsection{Real data}
\label{sec:real_data}
% !TEX root = ../double_trouble_main.tex

We next confirm that our method performs well on a real dataset in cancer genomics and semi-synthetic dataset in population genetics. In particular, we confirm that our method performs better than the independent and dependent extensions of one-population methods described above.

\paragraph{Comparisons.} For the real datasets, we find that different one-population methods exhibit substantially different performance for different types of data. So, for each type of data (cancer, population genetics), we test a variety of different one-population methods. We compare them on the one-population task and choose the best-performing method to serve as the comparison (via independent and dependent extensions) for our two-population method. For one-population methods, we consider the fourth-order jackknife \citep{burnham1978estimation, gravel2014predicting}, the Good-Toulmin estimator \citep{Orlitsky2016,chakraborty2019using}, linear programming \citep{zou2016quantifying}, the 3BP approach of \citet{Masoero2022} and the scaled process by~\citet{Camerlenghi2021} (See~\cref{sec:single_pop_methods} for this comparison). 
%\tb{I got rid of all the abbreviations since we don't need them in the main text. I presume they are needed in the supplement though, so you'll want to make sure they are defined before use there. (we can abbreviate the 4th order jackknife below.)}
%\tb{can we reference an appendix (or appendices) for these analyses here}

\subsubsection{Two cancer genomic datasets}

In cancer genomics, biologists are actively searching for (rare) shared genetic variants across cancer types that might lend insight into common causes of different cancers \citep{bojesen2013multiple, rashkin2020pan}. In our analysis that follows, we find that both our method and the independent extension perform well at predicting the number of new variants, while the dependent extension performs poorly. But unlike the independent extension, our method is able to predict the (non-trivial) number of shared variants across populations -- and provides a better prediction than the dependent extension.

In these experiments, we use the cancer genome atlas (TCGA) and MSK-IMPACT datasets \citep{cancer2008comprehensive, cancer2013cancer, cheng2015memorial}. 
We used the same processing as in \citet{chakraborty2019using} and \citet{Masoero2022}. %\tb{The original text said ``to compare these two datasets'', but it seems to me that we're not comparing the datasets in any meaningful way, so i don't follow that logic. is there any reason to focus on genes sequenced in both datasets?} 
%In what follows, we consider only genes sequenced in both datasets. 
To form our two populations, we focus on the two types of cancer with the largest numbers of samples; in each dataset, the resulting two types of cancer are lung cancer and breast cancer. In the MSK-IMPACT dataset, there are total of 749 lung cancer samples and 1532 breast cancer samples. In the TCGA dataset, there are total of 492 lung cancer samples and 811 breast cancer samples. So the pilot sizes are challenging in each case: 38, 57, 24, and 40, respectively.

For each population (lung cancer, breast cancer) within each dataset (TCGA, MSK-IMPACT), we compare performance of the one-population methods above; see \cref{sec:single_pop_methods}. For each population in each dataset, we find that the best-performing one-population method is the 3BP. So use the 3BP for the independent and dependent one-population extensions in this analysis.

\begin{figure}[htp]
    \centering
    \includegraphics[width = 0.9\textwidth]{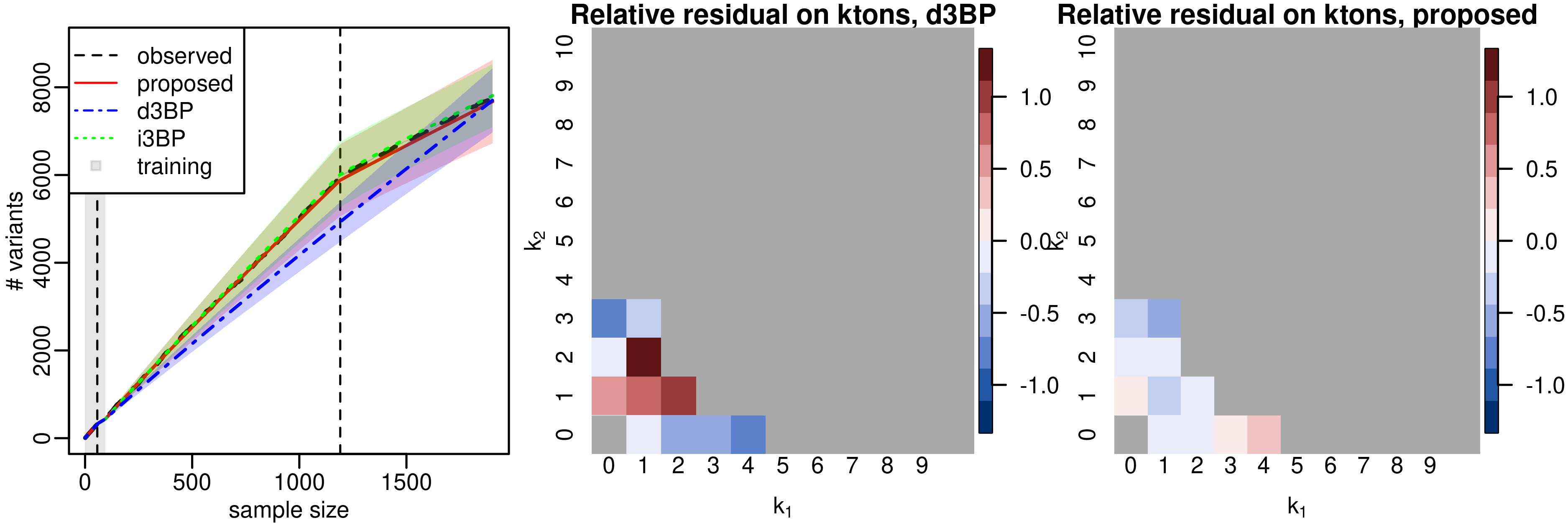}
    \includegraphics[width = 0.9\textwidth]{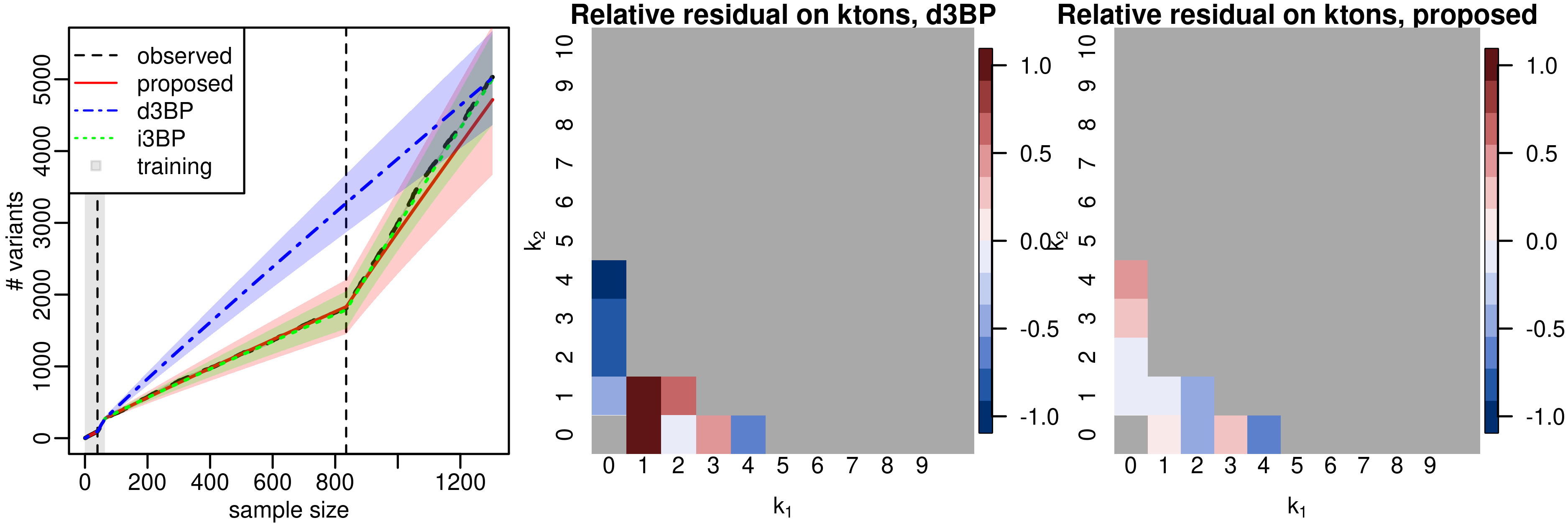}
    \caption{Predictions for the cancer genomics data. The rows correspond to datasets (upper: MSK-IMPACT, lower: TCGA). Population 1 is breast cancer, and population 2 is lung cancer. Other details of the plots are as in \cref{fig:proposed}.
    %\emph{Left}: The dashed black line shows the observed number of variants as a function of the sample number. The samples appear in the following order: the pilot from population 1 (breast cancer), the pilot from population 2 (lung cancer), the follow-up from population 1, the follow-up from population 2. A vertical dashed gray line separates the two populations in the pilot and follow-up, respectively. The gray shading covers the pilot data. All curves agree exactly on the pilot data. In the follow-up region, we plot mean (lines) and 1 standard deviation intervals (shaded regions) from three methods: our proposed method (solid red), the dependent version of the single-population 3BP (dash-dot blue), and the independent version of the 3BP (dotted green). \emph{Center and Right}: For each $\vecoccurrence$ with component values up to 10, we plot the relative residual for predictions from our method (right) and the dependent 3BP (center). Note that the color scales are fixed across a row but vary across a column. A square is gray when either $\vecoccurrence = (0,0)$ or the observed value is strictly less than 2.
    }
    \label{fig:cancer}
\end{figure}

The results in \cref{fig:cancer} demonstrate that our method outperforms the two one-population baselines. In particular, both our method and the independent extension agree closely with the ground truth number of new variants across the various follow-up sizes (left panels). Indeed, we expect a low level of sharing between cancer types since most variants in cancer are \textit{de novo} variants -- not subject to the natural selection that might make variants more likely to be shared across populations. As expected the dependent extension agrees with ground truth when the composition of the follow-up matches the pilot, but its predictions can be very far from the ground truth for other follow-up sizes. In the center and right plots, we see that there are still a number of shared variants with non-trivial counts in the ground truth; see also~\cref{fig:brca_luad_kton} for direct plots of the ground truth. The independent extension is unable to predict such counts. And while the dependent extension can predict a non-zero number in these cases, we see from the center and right plots in \cref{fig:cancer} that our method's predictions are more accurate. Overall then, our method is the only one able to provide both useful predictions of the numbers of variants in a follow-up as well as useful predictions of the number of shared variants across populations in the follow-up.

%In the datasets we observed similar predictions and performance in total variant prediction between our proposed method and i3BP. The similarity of the two models in this dataset is partly due to the very low level of sharing between cancer types since most variants in cancer are \textit{de novo} variants not subject to natural selection that might make a variants being shared across populations. Both growth curves are close to ground truth while our proposed method having slightly different variance across folds compare to i3bp (Fig.\ref{fig:bidc_la2}). However, our method was able to predict number of shared variants not seen in the pilot study while alternatives were not able to. We also observe the overestimation of k-tons along the diagonal by d3BP method, which made the assumption that all variants are shared with the same frequency, that allocated most variants along the diagonal causing overestimating in such region.

\subsubsection{A semi-synthetic population genetics dataset}
\label{sec:gnomad}

For population genetics, we use the GnomAD dataset \cite{karczewski2020mutational}. The dataset provides variant frequencies, but it does not provide exact variant data for each sample. Therefore, we take the semi-synthetic approach of \citet{Masoero2022}; in particular, for each individual, we randomly and independently sample the presence of each variant according to the reported frequency of that variant in the data. In general, prediction is more challenging (a) when the pilot is small in absolute terms since less data is available about the trend in variant growth and (b) when the follow-up is as large as possible since we have a priori more possibility for any observed trend to break down. To accomplish goal (b), we choose large populations, with several thousands of individuals: the Korean (1909 samples), Bulgarian (1335), and Southern European (5752) populations. With these large populations, we choose 30 folds in our evaluation to ensure the pilots are still somewhat small and similar to cancer dataset (64, 45, and 192, respectively), per goal (a).

For each population (Korean, Bulgarian, Southern European), we compare performance of the one-population methods at the start of \cref{sec:real_data}; see \cref{sec:single_pop_methods}. For each population in this case, we find that the fourth-order jackknife (4JK) is either the best performer or tied for best performance. So we use the 4JK for the independent and dependent one-population extensions in this analysis.

The results in \cref{fig:pop_gen} demonstrate that our method outperforms the two one-population methods at predicting the number of new variants in a follow-up. When the follow-up composition matches the pilot (the most favorable case for the dependent extension), our method is at least on par with the dependent one-population extension at predicting the number of shared variants in a follow-up. 
We see in the upper left plot that the dependent extension miscounts variants when the composition of the follow-up differs from the pilot. The independent extension predicts the first population well but seems to double count shared variants when the second population appears in the follow-up. In the lower left, the two populations are more similar, so the dependent extension performs reasonably; the independent extension, though, seems to double count shared variants in the second population in the follow-up. In both examples, our method tracks closely with ground truth. 

In the upper center plot, we see that the dependent one-population extension can be over an order of magnitude off in its predictions. In the upper right plot, we see that our method is much closer to ground truth across the $\vecoccurrence$ values here. Unlike in the cancer data case, the ground truth at all $\vecoccurrence$ values in this plot is substantially greater than zero; see \cref{fig:brca_luad_kton} for more information about the ground-truth counts. In particular, the signed residual plot in \cref{fig:brca_luad_kton} shows that the dependent extension is providing an overestimate along the diagonal. Indeed, we expect that the Korean and Bulgarian populations have little overlap and thus few shared variants -- whereas the dependent extension assumes they are the same population. Conversely, in the lower center and right plots, we see that our method is often farther from the ground truth than the dependent extension. We note that the magnitude of the error is smaller than in the upper plots, but it is still quite high (larger than 100\%). It may be that the complex history between the related Southern European and Bulgarian populations cannot be captured by the simple models we consider here.  We reiterate that the independent extension predicts zero for all $\vecoccurrence$ values that do not have a zero component and thus fails to predict all such values in these examples.

%We want a small pilot size since it is more challenging to predict when little data is available in absolute terms in the pilot, but we want a large follow-up to 

%We only have variant frequencies instead of individualized genotype in GnomAD dataset thus we took a synthetic approach following \citet{Masoero2022}. That is, for each individual, we randomly sample variants as independent Bernoulli trails with success probability being the frequency of that variant to form a synthetic genotype of each individual. Compare to the cancer genome datasets, the sample size of GnomAD is larger. To have a similar amount of training data as in the previous settings we form a 30-fold cross validation. We focused on two settings with several thousands individuals each population, namely Korean (pilot:64; follow-up:1845) and Bulgarian (pilot: 45; follow-up: 1,290), representing a combination of largely different population and Southern European (pilot: 192; follow-up: 5560), Bulgarian, representing similar individuals.

%During our single population test the best performing single population model in predicting total number of variants is fourth order jackknife and we show the comparison with i4jk and d4jk in this experiment. 

%Still, we predict number of total variants as well as number of shared variants and visualize the same way as in cancer datasets, i.e. plot the number of total predicted variants, averaged across folds, as a function of total data points. And visualize a heat map with relative residual ($(predicted-observed)/observed$) averaging over folds.  (Fig.~\ref{fig:kor_bgr}).

\begin{figure}[htp]
    \centering
    \includegraphics[width = 0.9\textwidth]{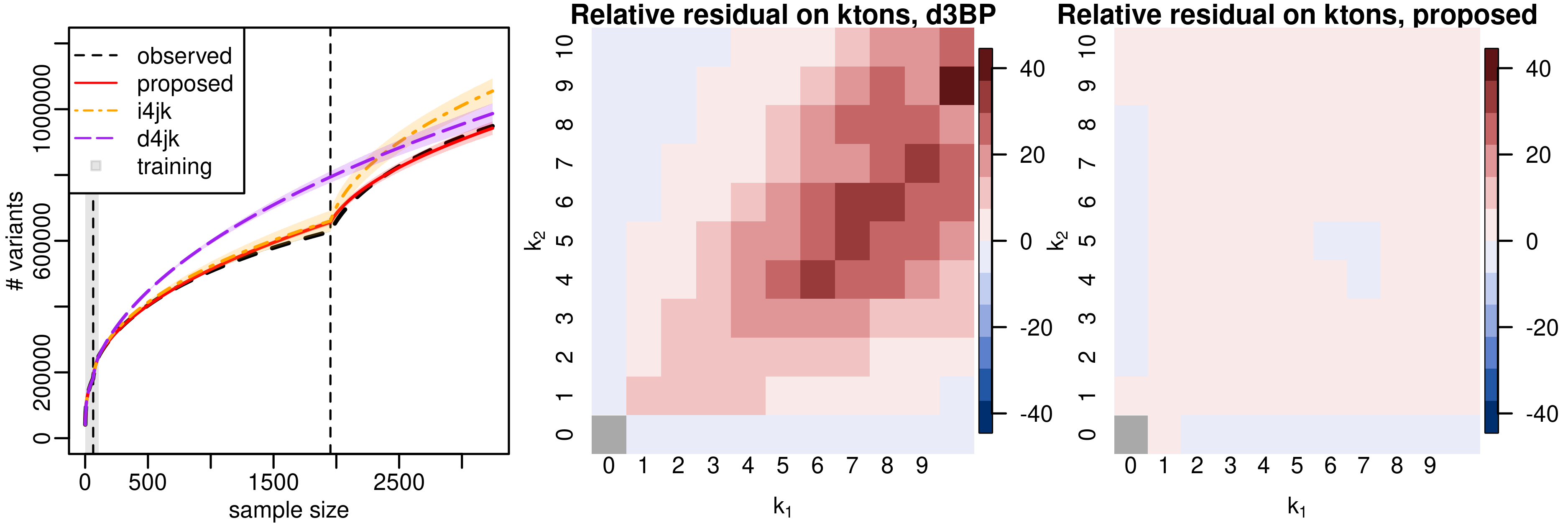}
    \includegraphics[width = 0.9\textwidth]{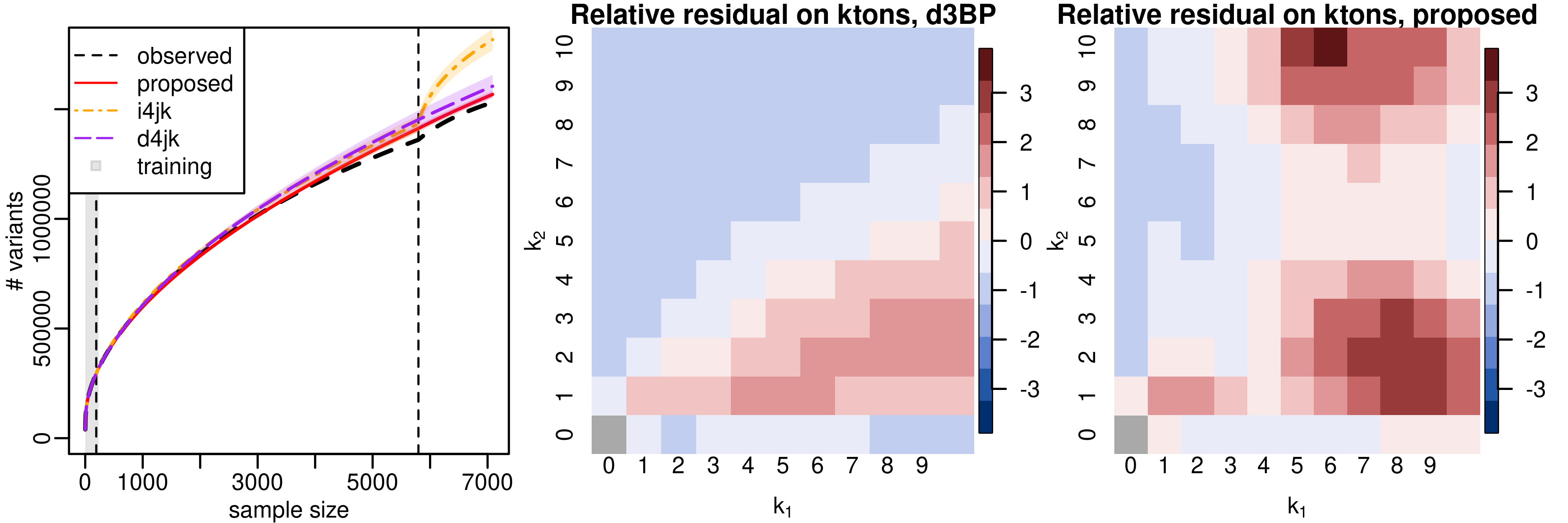}
    \caption{Predictions for the gnomAD data. The rows correspond to datasets (upper: Korean--Bulgarian, lower: Southern European--Bulgarian). In particular, in the top row, population 1 is Korean, and population 2 is Bulgarian. In the bottom row, population 1 is Southern European, and population 2 is Bulgarian. Other details of the plots are as in \cref{fig:proposed}. Other details of the plots are as in \cref{fig:proposed}. We highlight the very different color scales between the upper and lower rows.
    % \emph{Left}: The dashed black line shows the observed number of variants as a function of the sample number. The samples appear in the following order: the pilot from population 1 (upper:Korean lower: Southern Europeancancer), the pilot from population 2 (Bulgarian), the follow-up from population 1, the follow-up from population 2. A vertical dashed gray line separates the two populations in the pilot and follow-up, respectively. The gray shading covers the pilot data. All curves agree exactly on the pilot data. In the follow-up region, we plot mean (lines) and 1 standard deviation intervals (shaded regions) from three methods: our proposed method (solid red), the dependent version of the single-population fourth order jackknife (d4jk, dash-dot blue), and the independent version of the fourth order jackknife  (i4jk, dotted green). \emph{Center and Right}: For each $\vecoccurrence$ with component values up to 10, we plot the relative residual for predictions from our method (right) and the dependent 3BP (center). Note that the color scales are fixed across a row but vary across a column. A square is gray when either $\vecoccurrence = (0,0)$ or the observed value is strictly less than 2.
    }
    %\label{fig:kor_bgr}
   % \label{fig:seu_bgr}
   \label{fig:pop_gen}
\end{figure}

\section{Discussion}
\label{sec:discussion}
We proposed a novel approach to estimating the number of new genetic variants that will be observed in a follow-up sequencing study of two potentially heterogeneous populations. In this setting, a researcher has already conducted a pilot sequencing study and may be interested in the best allocation of resources to maximize the detection of new genetic variants in a larger follow-up study. Given the success of the Bayesian nonparametric methodology of \citet{Masoero2022} in the one-population case and the flexibility of the Bayesian framework, we chose to take a BNP approach for two populations. In particular, we decided to use a model where variant frequencies are generated according to a Poisson point process, parameterized by a rate measure.

Surprisingly, we proved that a product rate measure is incompatible with the requirements that any finite sample only obtains a finite number of variants while the number of variants will grow without bound as the sample sizes increase. This difficulty arises due to the fact that variants may be shared between groups; thus we proposed a novel conjugate measure that satisfies these desiderata. Given this incompatibility we used a Poisson point process with a novel non-product rate measure. Although the form of our new approach prohibits analytical calculations, we showed that numerical integrals still facilitate robust and practical evaluation of quantities of interest.

We provided upper and lower bounds on the growth in the number of new variants in our model, such that the bounds exhibit power laws. 
We checked the bounds empirically and also noted that empirically the variant-number growth from our model seems to exhibit a power law itself.
Nonetheless, explicit proofs that this growth exhibits power law behavior will provide deeper understanding of this process. %A potential approach for Type-I power law may be to bound with projection on each axis, but other types of power laws may require new theoretical developments (e.g. \citet{Gnedin2007}). 

Using simulations, we showed that our approach works to predict the number of novel variants from both real and simulated data. Since we are not aware of other two-population methods, we compared to simple baselines based on one-population methods: (1) treating the two populations as separate, which we call the independent extension and (2) treating all the data as from a single population, which we call the dependent extension. As expected, we find that the independent extension tends to overestimate the number of new variants when there are any shared variants, and it is not able to predict shared variants. The dependent extension is accurate when the populations appear in the same proportions in the pilot and the follow-up, but otherwise can provide poor estimates. Our method, by contrast, can correctly account for shared variants and differing population proportions between a pilot and follow-up. 

Several opportunities remain for future work. First, we rely on numerical integration for our predictions. Even in the two-population case, these computations can be somewhat expensive. We expect further computational (and statistical) challenges when considering higher numbers of populations. We suspect that alternative approaches may prove fruitful. In addition, we explored only one possible option for generalizing the beta process to higher dimensions. Although our experiments validated our approach, even better performance might be possible with an alternative model -- especially one carefully aligned with biological prior knowledge.

\FloatBarrier

\section*{Acknowledgments}
YS and TB were supported in part by the DARPA I2O LwLL program, an NSF Career Award, and an ONR Early Career Grant. JGS was supported by NIH grant R35GM137758 to Michael D. Edge. The authors acknowledge the MIT SuperCloud and Lincoln Laboratory Supercomputing Center for providing HPC resources that have contributed to the research results reported within this paper. 

\bibliography{refs}

\begin{thebibliography}{}
\bibitem[Freedman, 1973]{Freedman1973}
Freedman, D. (1973).
\newblock Another note on the {Borel-Cantelli} lemma and the strong law , with
  the {Poisson} approximation as a by-product.
\newblock {\em The Annals of Probability}, 1(6):910--925.
\end{thebibliography}

\renewcommand{\theequation}{S\arabic{equation}}
\renewcommand{\thesection}{S\arabic{section}}  
\renewcommand{\thefigure}{S\arabic{figure}}  
\renewcommand{\thetable}{S\arabic{table}} 
\renewcommand{\themyLemma}{S\arabic{myLemma}} 
\renewcommand{\themyRemark}{S\arabic{myRemark}} 
\renewcommand{\themyTheorem}{S\arabic{myTheorem}} 
\renewcommand{\themyProposition}{S\arabic{myProposition}}
\setcounter{equation}{0}
\setcounter{section}{0}
\setcounter{subsection}{0}
\setcounter{subsubsection}{0}
\setcounter{myLemma}{0}
\setcounter{myTheorem}{0}
\setcounter{myProposition}{0}
\setcounter{myRemark}{0}

%\section*{Appendix}
%\section{Appendix}
\newpage
\FloatBarrier
\appendix

\section{Proofs of results in \Cref{sec:cant_have_it_all}}
\label{sec:product_measure}

In this section, we provide additional details and proofs for the results presented in \Cref{sec:cant_have_it_all} for a generic finite number of populations. We state and prove a slight generalization of \cref{eq:generative_model}, where the number of populations can be two or greater: $\numpops > 1$. The two-population case of \cref{eq:generative_model} follows as a corollary.

Before stating our result, we extend our notation to encompass any strictly positive integer number of populations $\numpops > 0$.
We let $\samplesize{\popidx}$ denote the number of observed pilot samples in population $p$ with $1 \le p \le P$.
We collect all the pilot sizes in the vector $\vecsamplesize := (\samplesize{1}, \ldots,\samplesize{\numpops})$.
We collect the observations from the $\unitidx$-th sample in the $\popidx$-th population in $\genericmeasurevariantcount$ as in \cref{eq:observation_measure_inf}.
We write $\measurevariantcount{\popidx}{1:\samplesize{\popidx}}$ for the pilot samples in population $\popidx$.
And we write $X_{1:\numpops, \bm{1}:\vecsamplesize}$ for the full collection of pilot samples across populations.

We collect the variant frequencies across populations as $\genericvecvariantfreq = (\variantfreq{1}{\variantidx}, \ldots, \variantfreq{\numpops}{\variantidx})$.
The vector-valued measure $\randommeas :=  \sum_{\variantidx=1}^{\infty} \genericvecvariantfreq \dirac{\genericvariantlabel}$ matches variant frequency
vectors with their labels. We write $\randommeas \sim \PPP(\levync)$ to indicate that $\randommeas$
is generated by drawing $\{\genericvecvariantfreq\}_{\variantidx=1}^{\infty}$ from a Poisson
point process (PPP) with rate measure $\levync(\de\vecvariantfreq{})$ and drawing the $\genericvariantlabel$
i.i.d.\ uniform on $[0,1]$. We model $\genericvariantcount \sim \bernoullirv{\genericvariantfreq}$, i.i.d.\ across $\unitidx$
and independent across $\popidx$ and $\variantidx$; when the $\genericvariantcount$ are drawn this way, we say that
$\pilotdata$
is drawn according to a \emph{multiple-population Bernoulli process} ($\mBeP$) with measure parameter $\randommeas$ and count-vector parameter $\vecsamplesize$.
The $\numpops$-population generalization of \cref{eq:generative_model} is then
\begin{align}
	\begin{split}
		\randommeas & \sim \PPP(\nu) \quad \textrm{ (Prior) } \\
		X_{1:\numpops, \bm{1}:\vecsamplesize} \mid \randommeas & \sim \mBeP(\randommeas, \vecsamplesize). \quad \textrm{ (Likelihood) }
		\label{eq:appx_generative_model}
	\end{split}
\end{align}
In particular, the notation and development in the main text is the $\numpops=2$ special case of what is written above.
Observe that \cref{req:finite_mean} and \cref{req:inf_mass} do not specifically reference $\numpops=2$ and can be
applied to the general $\numpops > 0$ case as well.

To prove our main result, we will find it useful to show that each of \cref{req:finite_mean} and \cref{req:inf_mass} is equivalent to
an integral constraint on the L\'{e}vy  rate measure $\levync$ of the prior distribution $\randommeas$. To that end, we state and prove the following lemma.
\begin{myLemma}[Integral requirements]
	Take a number of populations $\numpops$ with $1 < \numpops < \infty$.
	Take data generated according to the model in \Cref{eq:appx_generative_model}. Then \Cref{req:finite_mean} holds if and only if the rate measure $\levync$ satisfies  
	\begin{equation}
		\forall \popidx \in \{1,\ldots,\numpops\}, \quad \int_{[0,1]^\numpops} \variantfreqperpop{\popidx} \levync(\de\vecvariantfreq{})<\infty.  \label{eq:finite_first_moment}
	\end{equation} 
	And \Cref{req:inf_mass} holds if and only if the rate measure $\levync$ satisfies 
	\begin{equation}
		\forall \popidx \in \{1,\ldots,\numpops\}, \quad \int_{(0,1]}  \levync_{-\popidx} (\de \variantfreqperpop{\popidx})=\infty. \label{eq:infinite_mass}
	\end{equation}
		Here $\levync_{-\popidx} (\de \variantfreqperpop{\popidx}) := \int_{[0,1]^{\numpops-1}} \levync(\de\vecvariantfreq{-\popidx})$, where $\vecvariantfreq{-\popidx} = [\variantfreqperpop{1},\ldots,\variantfreqperpop{\popidx-1}, \variantfreqperpop{\popidx+1}, \ldots,\theta_P]^\top$ denotes a vector containing all population indices other than $\popidx$.
    \label{lemma:integral_requirements}
\end{myLemma}
\begin{proof}
%	Our proof works by examining one population at the time, and ensuring each population exhibits finitely many variants. This is accomplished via marginalization and application of one population results.
	By construction, a point drawn from the Poisson point process with rate measure $\levync$ is of the form 
	$\vecvariantfreq{} = [\variantfreqperpop{1}, \ldots,\theta_P]^\top \in [0,1]^\numpops$. 
	By standard properties of multivariate Poisson processes, the $\popidx$-th coordinate of the multivariate Poisson process is a univariate Poisson point process with rate measure $\levync_{-\popidx} (\de \variantfreqperpop{\popidx}) := \int \levync(\de\vecvariantfreq{-\popidx})$ \citep[Mapping theorem,][ch. 2.3]{kingman1992poisson}.

	In our setting, then, an individual $\genericmeasurevariantcount$ in population $\popidx$ is drawn from a Bernoulli process, where the parameter encoding the variants with their frequency follows a Poisson point process with rate measure $\levync_{-\popidx}$. Following the same argument in \citet[][Section 2.3]{Broderick2018}, $\genericmeasurevariantcount$ shows finitely many variants (\cref{req:finite_mean}) with probability one if and only if the underlying rate measure has finite expectation on $[0,1]$. In our case, population $\popidx$'s frequency is from the marginal rate; thus individuals exhibit finitely many variants if and only if $\int\theta_p \levync_{-\popidx}(\de \variantfreqperpop{\popidx}) < \infty$.
	By substituting the definition of $\levync_{-\popidx}(\de \variantfreqperpop{\popidx})$, $\int\theta_p \levync_{-\popidx}(\de \variantfreqperpop{\popidx}) < \infty$ can be rewritten as \cref{eq:finite_first_moment}.
	 
%	By \citet[][Section 2.3]{Broderick2018}, individuals from population $\popidx$ only exhibit finitely many variants when their frequencies are from a Poisson point process if and only if the rate measure have finite expectations on $[0,1]$. In our case, population $\popidx$'s frequency is from the marginal rate, thus individuals exhibit finitely many variants if and only if $\int\theta_p \levync_{-\popidx}(\de \variantfreqperpop{\popidx}) < \infty$.
%	By substituting the definition of $\levync_{-\popidx}(\de \variantfreqperpop{\popidx})$, this can be rewritten as~\cref{eq:finite_first_moment}.
	%\[
%		\int \dots \int \theta_p \levync(\vecvariantfreq{}) < \infty,
%	\]

	\Cref{req:inf_mass} is satisfied if and only if for each population there are infinitely many variants that can be discovered. Again by \citet[][Section 2.3]{Broderick2018} population $\popidx$ has infinitely many variants to be seen if and only if $\levync_{-\popidx} (\de \variantfreqperpop{\popidx})$ has infinite mass on $(0,1]$, which is the requirement stated in \cref{eq:infinite_mass}.  
\end{proof}

\begin{myRemark}[On \cref{req:inf_mass}]
	\label{remark:inf_mass}
	A traditional requirement of one-population Bayesian nonparametric feature models is that the rate measure has infinite mass on $(0,1]$.
	\Cref{req:inf_mass} provides an extension of this requirement to multiple populations. If we 
	think of \cref{req:inf_mass} as an extension from a single dimensional $\variantfreq{}{}$ to a multidimensional $ \vecvariantfreq{}$, 
	there a number of other possible extensions that might a priori seem reasonable.
	For instance, we might instead require infinite mass on $(0,1]^\numpops$ (i.e., all vectors in the unit cube with strictly positive entries) or
	infinite mass on $[0,1]^\numpops\setminus{\bm{0}}$ where $\bm{0}$ is a $\numpops$-long 0-vector (i.e., the unit cube without the origin).
	
	We next observe that these alternative extensions are not equivalent to \cref{req:inf_mass}.
	To help us construct our counterexamples, we will use $\numpops$ (unidimensional) measures $\mu_{1}(\de\variantfreqperpop{1}),\dots ,\mu_{\numpops}(\de\variantfreqperpop{\numpops})$. We assume that for $\popidx=1,\dots,\numpops$, we have $\int_{(0,1]} \mu_{\popidx}(\de\variantfreqperpop{\popidx})=\infty$ and $\int_{(0,1]}\variantfreqperpop{\popidx} \mu_{\popidx}(\de\variantfreqperpop{\popidx})<\infty$. We denote the delta measure at $0$ as $\delta_0(\cdot)$. 
	
	First, we construct a measure that has 0 mass on $(0,1]^\numpops$ but satisfies \cref{req:inf_mass}. To construct our measure, we put mass only on the axes; namely, we consider the measure $\sum_{\popidx=1}^\numpops\mu_{\popidx}(\de\variantfreqperpop{\popidx})\prod_{k: k\ne \popidx}\delta_{0}(\de\variantfreqperpop{k})$. This measure can be interpreted as $\numpops$ populations that do not share variants; each population separately draws its variant frequencies from $\mu_{\popidx}$. Since by assumption $\mu_{\popidx}$ satisfies the single-population requirements, the full measure satisfies \cref{req:inf_mass}. But the measure has 0 mass on $(0,1]^\numpops$ because its mass lies entirely along the axes.

	Next, we construct a measure that has infinite mass on $[0,1]^\numpops\setminus{\bm{0}}$ but does not satisfy \cref{req:inf_mass}. In this case, we put mass on exactly one axis. Then all but one population exhibit zero variants. In particular, we consider the measure $\mu_{1}(\de\variantfreqperpop{1})\prod_{k=2}^\numpops\delta_0(\de\variantfreqperpop{k})$. This measure has infinite mass on $[0,1]^\numpops\setminus{\bm{0}}$. But when we sample from, e.g., population 2, we will not see any variants since all variants have frequency 0. Thus this measure does not satisfy \cref{req:inf_mass}. 
\end{myRemark}

We now extend the statement of \Cref{prop:incomp} of the main text to the general setting for an arbitrary number of populations $\numpops$. 
To do so, we first restate~\Cref{req:product_measure}, now for a generic number of populations.
\begin{conditionp}{C'} \label{req:appx_product_measure}
	The rate measure for the frequency-generating PPP in $\numpops$ populations factorizes across populations:
	$\levync(\de\btheta) = \prod_{\popidx=1}^{\numpops} \levync_{\popidx}(\de\variantfreqperpop{\popidx})$.  
\end{conditionp}
Notice that \Cref{req:product_measure} is a special case of the condition above, when $\numpops = 2$.

We can now state the result which generalizes \Cref{prop:incomp}.

\begin{myTheorem} \label{prop:incomp_general}
	Take the general $\numpops$ population model of \cref{eq:appx_generative_model}. If the Poisson point process
	rate measure in the prior factorizes across populations according to \cref{req:appx_product_measure},
	the model cannot simultaneously generate samples with finitely many variants (\cref{req:finite_mean})
	and guarantee that there are always more variants to discover (\cref{req:inf_mass}).
\end{myTheorem}

\begin{proof}
We first show that if \Cref{req:inf_mass} and \Cref{req:appx_product_measure} hold, then \Cref{req:finite_mean} cannot hold:
\[
	\text{If }\Cref{req:inf_mass} \text{ and } \Cref{req:appx_product_measure} \text{ hold, } \Cref{req:finite_mean} \text{ does not hold}.
\]
By \Cref{req:appx_product_measure} we have that $\levync$ is a ``product measure'', i.e. $\levync(\de\vecvariantfreq{})=\prod_{\popidx} \levync_{\popidx}(\de\variantfreqperpop{\popidx})$ and in turn we have $\levync_{-\popidx} (\de \variantfreqperpop{\popidx}) = \levync_{\popidx}(\de\variantfreqperpop{\popidx}) \prod_{k: k\ne \popidx} \int_{[0,1]}\levync_{k}(\de\variantfreqperpop{k})$. By \Cref{req:inf_mass} and its equivalent~\cref{eq:infinite_mass}, for all $\popidx$, we have $\int_{(0,1]} \levync_{-\popidx} (\de \variantfreqperpop{\popidx})=\infty$ thus either:
\begin{itemize}
	\item[(i)] $\int_{(0,1]}\levync_{\popidx}(\de\variantfreqperpop{\popidx})=\infty$ and all other $k\ne \popidx$, $\int_{[0,1]}\levync_{k}(\de\variantfreqperpop{k})>0$, or 
	\item[(ii)] $0<\int_{(0,1]}\levync_{\popidx}(\de\variantfreqperpop{\popidx})<\infty$ and there exists one $k\ne \popidx$ for which $\int_{[0,1]}\levync_{k}(\de\variantfreqperpop{k})=\infty$.
\end{itemize}
We now show that in both cases (i), (ii),  \Cref{req:finite_mean} does not hold.

In case (i), for $k\ne \popidx$:
\[
	\int_{[0,1]^{\numpops}}\variantfreqperpop{k}\levync(\de\vecvariantfreq{}) = \int_{[0,1]}\variantfreqperpop{k} \levync_{k}(\de\variantfreqperpop{k}) \int_{[0,1]} \levync_{\popidx}(\de\variantfreqperpop{\popidx})\prod_{i: i\ne k, \popidx} \int_{[0,1]}\levync_{i}(\de\variantfreqperpop{i})=\infty.
\]
Since we have $\prod_{i: i\ne k, \popidx} \int_{[0,1]}\levync_{i}(\de\variantfreqperpop{i})>0$ and $\int_{[0,1]}\levync_{\popidx}(\de\variantfreqperpop{\popidx})=\infty$, it remains to check $\int_{[0,1]} \variantfreqperpop{k} \levync_{k}(\de\variantfreqperpop{k})>0$. 
Recall that we assumed~\cref{eq:infinite_mass} is satisfied, we must have the ``projected'' measure $\levync_{k}(\de\variantfreqperpop{k}) \int_{[0,1]} \levync_{\popidx}(\de\variantfreqperpop{\popidx})\prod_{i: i\ne k, \popidx} \int_{[0,1]}\levync_{i}(\de\variantfreqperpop{i})$ to have mass on $(0,1]$, which implies $\levync_{k}(\de\variantfreqperpop{k})$ must have mass on $(0,1]$. Then we must also have $\int_{[0,1]}\variantfreqperpop{k} \levync_{k}(\de\variantfreqperpop{k})>0$, hence implying that \Cref{req:finite_mean} does not hold.

In case (ii), by direct factorization of the rate measure it holds that
\[
	\int_{[0,1]}\variantfreqperpop{\popidx}\levync(\de\vecvariantfreq{}) = \int_{(0,1]}\variantfreqperpop{\popidx} \levync_{\popidx}(\de\variantfreqperpop{\popidx}) \prod_{k: k\ne \popidx} \int_{[0,1]}\levync_{k}(\de\variantfreqperpop{k})=\infty.
 \] 
Hence, also in this case, \Cref{req:finite_mean} does not hold. We have then showed that if \Cref{req:inf_mass} and \Cref{req:appx_product_measure} hold, then \Cref{req:finite_mean} cannot hold.

%it must be the case that for any population $i \ne \popidx, k$,
%\begin{align*}
%    \int \variantfreqperpop{i} \levync(\de\vecvariantfreq{})&\ge \prod_{k\ne i,\popidx}\int_{[0,1]} \levync_k(\de\variantfreqperpop{k}) \times \int_{(0,1]} \levync_\popidx(\de \variantfreqperpop{\popidx})\times\int \variantfreqperpop{i}\levync_{i}(\de\variantfreqperpop{i}) = \infty.
%\end{align*}
%The first product in the right hand side is strictly positive, the second integral is infinite, and the third term is strictly positive by the assumption stated in  \Cref{req:product_measure}, proving that if \Cref{req:inf_mass} and \Cref{req:product_measure} hold, then \Cref{req:finite_mean} cannot hold.

%The first factor is positive and second is $\infty$, while the last one is non-negative, thus the expectation can only be either 0 or $\infty$. Further by assumption we assume all coordinates has positive mass in \Cref{req:product_measure}, the product has to be $\infty$ so we cannot have \Cref{req:finite_mean} to hold. 

We now show that if \Cref{req:appx_product_measure} and \Cref{req:finite_mean} hold, then \Cref{req:inf_mass} cannot hold. 
\[
	\text{If }\Cref{req:finite_mean} \text{ and } \Cref{req:appx_product_measure} \text{ hold, } \Cref{req:inf_mass} \text{ does not hold}.
\]
Under \Cref{req:finite_mean} and its equivalent~\cref{eq:finite_first_moment} and \Cref{req:appx_product_measure}, for a given $\popidx$, 
\begin{align*}
    \int_{[0,1]^\numpops} \variantfreqperpop{\popidx} \levync(\de\vecvariantfreq{})&=\prod_{k\ne \popidx}\int_{[0,1]} \levync_k(\de\variantfreqperpop{k}) \times\int_{[0,1]} \variantfreqperpop{\popidx}\levync_{\popidx}(\de\variantfreqperpop{\popidx})<\infty.
\end{align*}
For the condition above to be verified, one of the following two cases must hold:
\begin{itemize}
	\item[(i)] $\int_{[0,1]} \variantfreqperpop{\popidx}\levync_{\popidx}(\de\variantfreqperpop{\popidx})=0$, or
	\item[(ii)] $\int_{[0,1]} \variantfreqperpop{\popidx}\levync_{\popidx}(\de\variantfreqperpop{\popidx})>0$. 
\end{itemize}

We first show that when (i) holds, \Cref{req:inf_mass} can't be verified. Indeed, $\int_{[0,1]} \variantfreqperpop{\popidx}\levync_{\popidx}(\de\variantfreqperpop{\popidx})=0$ only holds when $\levync_{\popidx}(\de\variantfreqperpop{\popidx})$ is a $\delta$ measure at 0. And in this case, the marginal $\levync_{-\popidx}$ has no support on $(0,1]$, thus~\cref{eq:infinite_mass} cannot hold for coordinate $\popidx$. 

For case (ii), we have for all coordinates $k\ne \popidx $, $\int_{[0,1]} \levync_k(d\variantfreqperpop{k})<\infty$. We can repeat the argument with a different coordinate and conclude also $\int_{[0,1]} \levync_\popidx(d\variantfreqperpop{\popidx})<\infty$, thus the total mass is finite, $\int \levync(\de\vecvariantfreq{})<\infty$, and \Cref{req:inf_mass} cannot hold.  

%Then we can take a different coordinate $k\ne \popidx$ and repeat the argument, which implies for all $i\ne k$ $\int \levync_i(d\variantfreqperpop{i})<\infty$, i.e. for all coordinate $\popidx$ we must have $\int \levync_\popidx(d\variantfreqperpop{\popidx})<\infty$. That is, we cannot have \Cref{req:inf_mass} to hold.

Last, suppose both  \Cref{req:finite_mean} and \Cref{req:inf_mass}. Using the two results above, by contradiction it cannot be that \Cref{req:appx_product_measure} holds.
\end{proof}

\begin{proof}[Proof of \Cref{prop:incomp}]
	\Cref{prop:incomp} is a special case of \Cref{prop:incomp_general} when $\numpops = 2$. 
\end{proof}

\section{Multi-population BNP models satisfying Desiderata \ref{req:finite_mean} and \ref{req:inf_mass}}
\label{app:unified_view}
% !TEX root = ../double_trouble_supp.tex

In this section, in light of the incompatibility result stated in \Cref{prop:incomp} and proved in \Cref{sec:product_measure}, we provide details on BNP models that instead satisfiy \Cref{req:finite_mean} and~\Cref{req:inf_mass}.

\subsection{Proposed prior satisfies \Cref{req:finite_mean} and~\Cref{req:inf_mass}}
\label{sec:proper_proof}

%We choose to keep \Cref{req:finite_mean} and~\Cref{req:inf_mass} since \Cref{req:inf_mass} can guarantee our proposed model will have infinitely many potential variants thus there is no need to model upper bound while \Cref{req:finite_mean} can guarantee our model will not make predictions that infinitely many variants showed up in finite samples. 
In \Cref{eq:proposed_prior}, we proposed a prior distribution that is \emph{not} a product measure factorizing across multiple populations, i.e.\ it does not satisfy \Cref{req:product_measure}.
We now prove in ~\cref{prop:near_conjugate_proper} that our proposed prior of \cref{eq:proposed_prior}  satisfies \Cref{req:finite_mean} and~\Cref{req:inf_mass}. %\footnote{\lom{again I think we want to be more verbose here: 1) what is the prior? cref it to the main text. 2) what do we mean by proper? 3) why does this matter?}}
\begin{myProposition}
    \label{prop:near_conjugate_proper}
    Consider data generated from \Cref{eq:generative_model}, where we replace the generic rate measure $\levync$ with the following rate measure defined on the unit square $[0,1]^2$%\footnote{\lom{should we add the support? i.e. $1(\theta_1 \in [0,1])1(\theta_2 \in [0,1])$? Also do we want to give it a special symbol?}} 
    \begin{equation}
    	\ratemeasproposed(\de\vecvariantfreq{})=\frac{\mass}{B(\corr{1},\conc{1})B(\corr{2},\conc{2})}\frac{(\variantfreqperpop{1}+\variantfreqperpop{2}^{\rate{2}/\rate{1}})^{-\rate{1}}}{(\variantfreqperpop{1}+\variantfreqperpop{2})^{\corr{1}+\corr{2}}}\variantfreqperpop{1}^{\corr{1}-1}\variantfreqperpop{2}^{\corr{2}-1}(1-\variantfreqperpop{1})^{\conc{1}-1}(1-\variantfreqperpop{2})^{\conc{2}-1} \de \variantfreqperpop{1} \de \variantfreqperpop{2}. \label{eq:proposed_levy}
    \end{equation}
    When $\mass,\conc{1},\conc{2}, \corr{1},\corr{2}>0$ and $\rate{1},\rate{2}\in (0,1)$, desiderata \ref{req:finite_mean} and \ref{req:inf_mass} hold, i.e., by virtue of \Cref{lemma:integral_requirements},
    \begin{itemize}
    	\item[(i)] $\iint_{[0,1]^2}\variantfreqperpop{\popidx}\ratemeasproposed(\de\btheta)<\infty$ for $\popidx=1,2$, and
    	\item[(ii)] $\iint_{(0,1]^2}\ratemeasproposed(\de \vecvariantfreq{})=\infty$.
	\end{itemize} 
    Note that in~\cref{remark:inf_mass} we discussed that (ii) is not equivalent to~\cref{eq:infinite_mass}, but it is sufficient for~\cref{eq:infinite_mass}.\end{myProposition}

\begin{proof} \label{proof:proper}
    To show that the desiderata \ref{req:finite_mean} and \ref{req:inf_mass} for our data hold, we equivalently prove the integral requirements provided in \Cref{lemma:integral_requirements} are verified when using $\ratemeasproposed$, with hyperparameters $\mass,\conc{1},\conc{2}, \corr{1},\corr{2}>0$ and $\rate{1},\rate{2}\in (0,1)$.
    The crux of the proof relies on the analysis of the behavior of the integral requirements in a neighborhood of the origin  $A:=[0,\epsilon]^2$ and $D:=(0,\epsilon]^2$ for an arbitrarily small scalar $\epsilon > 0$. 
    In what follows, we use the notation $I_A(x) = \mathbbm{1}(x \in A)$, $I_{D}(x)=\mathbbm{1}(x \in D)$ and
    \[
    	\phi:=\frac{\mass}{B(\corr{1},\conc{1})B(\corr{2},\conc{2})}<\infty.
    \]
    We start with the integral requirement (i).
    \paragraph{Proof of the integral requirement (i) (\Cref{eq:finite_first_moment} in \Cref{lemma:integral_requirements}).}
    Notice that for an arbitrary $\popidx \in \{1,2\}$, we can rewrite \Cref{eq:finite_first_moment} as follows:
    	\begin{equation*}
    		\iint_{[0,1]^2}\variantfreqperpop{\popidx}\ratemeasproposed(\de \vecvariantfreq{} )= \int_{A^C}\variantfreqperpop{\popidx}\ratemeasproposed(\de\vecvariantfreq{}) + \int_A \variantfreqperpop{\popidx}\ratemeasproposed(\de\vecvariantfreq{}).
	\end{equation*}
	We will now show that for our proposed $\ratemeasproposed$ both integrals above are finite.
%        	
%        &= \frac{\mass}{B(\corr{1},\conc{1})B(\corr{2},\conc{2})}\int_A \frac{(\variantfreqperpop{1}+\variantfreqperpop{2}^{\rate{2}/\rate{1}})^{-\rate{1}}}{(\variantfreqperpop{1}+\variantfreqperpop{2})^{\corr{1}+\corr{2}}}\variantfreqperpop{1}^{\corr{1}-1}\variantfreqperpop{2}^{\corr{2}-1}(1-\variantfreqperpop{1})^{\conc{1}-1}(1-\variantfreqperpop{2})^{\conc{2}-1}  \variantfreqperpop{\popidx}\de\vecvariantfreq{}\\
%        & +  \frac{\mass}{B(\corr{1},\conc{1})B(\corr{2},\conc{2})} \int I_{A^C}(\vecvariantfreq{})\frac{(\variantfreqperpop{1}+\variantfreqperpop{2}^{\rate{2}/\rate{1}})^{-\rate{1}}}{(\variantfreqperpop{1}+\variantfreqperpop{2})^{\corr{1}+\corr{2}}}\variantfreqperpop{1}^{\corr{1}-1}\variantfreqperpop{2}^{\corr{2}-1}(1-\variantfreqperpop{1})^{\conc{1}-1}(1-\variantfreqperpop{2})^{\conc{2}-1}  \variantfreqperpop{\popidx} \de\vecvariantfreq{} 
%    \end{align*} 
    First, consider the integral over $A^C$ above. For $\popidx\in\{1,2\}$,
    	\begin{equation}
    		\int_{A^C}\variantfreqperpop{\popidx}\ratemeasproposed(\de\vecvariantfreq{}) =  \phi \int I_{A^C}(\vecvariantfreq{})\frac{(\variantfreqperpop{1}+\variantfreqperpop{2}^{\rate{2}/\rate{1}})^{-\rate{1}}}{(\variantfreqperpop{1}+\variantfreqperpop{2})^{\corr{1}+\corr{2}}}\variantfreqperpop{\popidx}\variantfreqperpop{1}^{\corr{1}-1}\variantfreqperpop{2}^{\corr{2}-1}(1-\variantfreqperpop{1})^{\conc{1}-1}(1-\variantfreqperpop{2})^{\conc{2}-1}  \de\vecvariantfreq{}. \label{eq:int_AC}
    	\end{equation}
    Notice that  we can bound the first part of the integrand as follows:
   	\begin{equation}
   		I_{A^C}( \vecvariantfreq{})\frac{(\variantfreqperpop{1}+\variantfreqperpop{2}^{\rate{2}/\rate{1}})^{-\rate{1}}}{(\variantfreqperpop{1}+\variantfreqperpop{2})^{\corr{1}+\corr{2}}}
		\le \min\{\epsilon,\epsilon^{\rate{2}/\rate{1}}\}^{-\rate{1}-\corr{1}-\corr{2}}<\infty. \label{eq:bound_betas}
	\end{equation}
	Indeed, if $\vecvariantfreq{}\notin A^C$, then the inequality holds trivially. If instead $\vecvariantfreq{}\in A^C$, then $ \variantfreqperpop{1} + \variantfreqperpop{2} >\epsilon$. In turn, 
	\[
		(\variantfreqperpop{1}+\variantfreqperpop{2})^{-\corr{1}-\corr{2}}< \epsilon^{-\corr{1}-\corr{2}} \le \min\{\epsilon, \epsilon^{\rate{2}/\rate{1}}  \}^{-\corr{1}-\corr{2}},
	\] 
	since $\epsilon < 1$. 
	Similarly, 
	$
		(\variantfreqperpop{1} + \variantfreqperpop{2}^{\rate{2}/\rate{1}})> \min \{ \epsilon, \epsilon^{\rate{2}/\rate{1}} \}
	$
	thus 
	\[
		(\variantfreqperpop{1}+\variantfreqperpop{2}^{\rate{2}/\rate{1}})^{-\rate{1}} < \min \{ \epsilon, \epsilon^{\rate{2}/\rate{1}} \}^{-\rate{1}}.
	\]
	 The two bounds combined yield \Cref{eq:bound_betas}. Recognizing that the remaining part of the integrand in \Cref{eq:int_AC} is the kernel of a beta random variable for either value of $\popidx$, it holds
	 \begin{align}
%	 \begin{split}
	 	\int_{A^C}  \variantfreqperpop{\popidx}\ratemeasproposed(\de\vecvariantfreq{}) &= \phi \int I_{A^C}(\btheta)\frac{(\variantfreqperpop{1}+\variantfreqperpop{2}^{\rate{2}/\rate{1}})^{-\rate{1}}}{(\variantfreqperpop{1}+\variantfreqperpop{2})^{\corr{1}+\corr{2}}}\variantfreqperpop{1}^{(\corr{1}+1)-1}\variantfreqperpop{2}^{\corr{2}-1}(1-\variantfreqperpop{1})^{\conc{1}-1}(1-\variantfreqperpop{2})^{\conc{2}-1}\de\vecvariantfreq{}  \nonumber \\
		&<\phi \min\{\epsilon,\epsilon^{\rate{2}/\rate{1}}\}^{-\rate{1}-\corr{1}-\corr{2}} \int_{A^C}  \variantfreqperpop{\popidx} \variantfreqperpop{1}^{\corr{1}-1} \variantfreqperpop{2}^{(\corr{2}-1}(1-\variantfreqperpop{1})^{\conc{1}-1}(1-\variantfreqperpop{2})^{\conc{2}-1}\de\vecvariantfreq{} \nonumber \\
		&<\infty. \label{eq:req_p0}
%	\end{split}
   	 \end{align}
    Thus, it remains for us to check that also on the set $A$, around the origin, the integral is finite, i.e.\ $\int_A \variantfreqperpop{\popidx}\ratemeasproposed(\de\vecvariantfreq{})<\infty$.
First, notice that on the set $A$, we can obtain a pointwise upper bound for the L\'{e}vy rate $\ratemeasproposed$ of \Cref{eq:proposed_prior} as follows:
    \[
    	\ratemeasproposed(\de \btheta) I_A(\btheta) \le  \eta \frac{(\variantfreqperpop{1}+\variantfreqperpop{2}^{\rate{2}/\rate{1}})^{-\rate{1}}}{(\variantfreqperpop{1}+\variantfreqperpop{2})^{\corr{1}+\corr{2}}}\variantfreqperpop{1}^{\corr{1}-1}\variantfreqperpop{2}^{\corr{2}-1} I_A(\btheta) \de \variantfreqperpop{1} \de \variantfreqperpop{2},
    \]
    with 
    \[
    	\eta := \phi \max\{1, (1-\epsilon)^{\conc{1}-1}\} \max\{ 1, (1-\epsilon)^{\conc{2}-1}\}.
    \]
    We here consider two separate cases. 
    
    First, for $p=1$,
    \begin{align*}
    \int_A \variantfreqperpop{1}\ratemeasproposed(\de\vecvariantfreq{}) & < \eta \int_A \frac{(\variantfreqperpop{1}+\variantfreqperpop{2}^{\rate{2}/\rate{1}})^{-\rate{1}}}{(\variantfreqperpop{1}+\variantfreqperpop{2})^{\corr{1}+\corr{2}}}\variantfreqperpop{1}^{(1+\corr{1})-1}\variantfreqperpop{2}^{\corr{2}-1}\de\btheta.
        \end{align*} 
    Now notice that $(\variantfreqperpop{1}+\variantfreqperpop{2}^{\rate{2}/\rate{1}})^{\rate{1}}>\variantfreqperpop{1}^{\rate{1}}$, and for $0<\delta<\corr{2}$
    \begin{align*}
    	(\variantfreqperpop{1}+\variantfreqperpop{2})^{\corr{1}+\corr{2}} &= (\variantfreqperpop{1}+\variantfreqperpop{2})^{\corr{1}+\corr{2}+\delta-\delta} 
												=(\variantfreqperpop{1}+\variantfreqperpop{2})^{\corr{1}+\delta} (\variantfreqperpop{1}+\variantfreqperpop{2})^{\corr{2}-\delta}
												>(\variantfreqperpop{1})^{\corr{1} + \delta } (\variantfreqperpop{2})^{\corr{2} - \delta},
	\end{align*}
	so we can further obtain the upper bound
    \begin{align}
     \int_A \variantfreqperpop{1}\ratemeasproposed(\de\vecvariantfreq{}) & < \eta \int_A (\variantfreqperpop{1})^{-\rate{1}} \variantfreqperpop{1}^{-\corr{1}-\delta} \variantfreqperpop{2}^{-\corr{2}+\delta} \variantfreqperpop{1}^{(1+\corr{1})-1}\variantfreqperpop{2}^{\corr{2}-1}\de\btheta \nonumber \\
%    &= \eta \int_A \variantfreqperpop{1}^{(1-\rate{1}-\delta) -1}\variantfreqperpop{2}^{\delta-1}\de\btheta \\
   &= \eta \int_0^1 \variantfreqperpop{1}^{(1-\rate{1}-\delta) -1} \de \variantfreqperpop{1} \int_0^1 \variantfreqperpop{2}^{\delta-1}\de\variantfreqperpop{2} < \infty, \label{eq:req_p1}
%        &=\eta\int_A \frac{(\variantfreqperpop{1}+\variantfreqperpop{2}^{\rate{2}/\rate{1}})^{-\rate{1}}}{(\variantfreqperpop{1}+\variantfreqperpop{2})^{\corr{1}+\corr{2}+\delta-\delta}}\variantfreqperpop{1}^{\corr{1}}\variantfreqperpop{2}^{\corr{2}-1}\de\btheta + C\\
%        &=\eta\int_A (\variantfreqperpop{1}+\variantfreqperpop{2}^{\rate{2}/\rate{1}})^{-\rate{1}}(\variantfreqperpop{1}+\variantfreqperpop{2})^{-\corr{1}-\delta} \variantfreqperpop{1}^{\corr{1}} (\variantfreqperpop{1}+\variantfreqperpop{2})^{-\corr{2}+\delta} \variantfreqperpop{2}^{\corr{2}-1}\de\btheta + C\\
%        &\le \eta \int_A\variantfreqperpop{1}^{-\rate{1}} \variantfreqperpop{1}^{-\corr{1}-\delta} \variantfreqperpop{1}^{\corr{1}}\variantfreqperpop{2}^{-\corr{2}+\delta}\variantfreqperpop{2}^{\corr{2}-1} \de\variantfreqperpop{1} \de \variantfreqperpop{2}+ C\\
%        %&\le \int_A(\variantfreqperpop{1}+\variantfreqperpop{2})^{-\corr{1}-1+\delta/2}\variantfreqperpop{1}^{\corr{1}}(\variantfreqperpop{1}+\variantfreqperpop{2})^{-\corr{2}+\delta/2}\variantfreqperpop{2}^{\corr{2}-1}\\
%        &= \eta \int_A\variantfreqperpop{1}^{-\corr{1}-\rate{1}-\delta}\variantfreqperpop{1}^{\corr{1}}\variantfreqperpop{2}^{-\corr{2}+\delta}\variantfreqperpop{2}^{\corr{2}-1} \de\variantfreqperpop{1} \de \variantfreqperpop{2}+ C\\
%        & \le \eta \int_0^1 \variantfreqperpop{1}^{(1-\rate{1}-\delta)-1}\de\variantfreqperpop{1} \int_0^1 \variantfreqperpop{2}^{\delta-1}\de\variantfreqperpop{2} + C\\
%        &<\infty
    \end{align}    
	where the last inequality holds as long as $1-\rate{1}-\delta > 0$, i.e.\ $\delta \in (0,1-\rate{1})$.

    For $\popidx = 2$, 
     \begin{align*}
    \int_A \variantfreqperpop{2}\ratemeasproposed(\de\vecvariantfreq{}) & < \eta \int_A \frac{(\variantfreqperpop{1}+\variantfreqperpop{2}^{\rate{2}/\rate{1}})^{-\rate{1}}}{(\variantfreqperpop{1}+\variantfreqperpop{2})^{\corr{1}+\corr{2}}}\variantfreqperpop{1}^{\corr{1}-1}\variantfreqperpop{2}^{(1+\corr{2})-1}\de\btheta.
        \end{align*} 
	We use symmetric bounds as the derivation above for $\popidx=1$ now for the case $\popidx=2$, i.e.\ $(\variantfreqperpop{1}+\variantfreqperpop{2}^{\rate{2}/\rate{1}})^{\rate{1}}>(\variantfreqperpop{2})^{\rate{1}}$, and, for $0<\delta < \corr{1}$, $(\variantfreqperpop{1}+\variantfreqperpop{2})^{\corr{1}+\corr{2}}>(\variantfreqperpop{1})^{\corr{1} - \delta } (\variantfreqperpop{2})^{\corr{2} + \delta}$, and plugging them in we get
     \begin{align}
    \int_A \variantfreqperpop{2}\ratemeasproposed(\de\vecvariantfreq{}) & < \eta \int_A (\variantfreqperpop{1})^{\delta-1}\variantfreqperpop{2}^{(1-\rate{1}-\delta)-1}\de\btheta
    =\eta \int_0^1 (\variantfreqperpop{1})^{\delta-1} \de \variantfreqperpop{1} \int_0^1 \variantfreqperpop{2}^{(1-\rate{1}-\delta)-1}\de \variantfreqperpop{2} < \infty,  \label{eq:req_p2}
        \end{align} 	
        where the last inequality holds as long as $1-\rate{1}-\delta > 0$, i.e.\ $\delta \in (0,1-\rate{1})$. Thus, \Cref{eq:req_p0,eq:req_p1,eq:req_p2} combined yield requirement (i) (\Cref{eq:finite_first_moment}).
        
        This concludes the proof of integral requirement (i) (see \Cref{eq:finite_first_moment} in \Cref{lemma:integral_requirements}).

%    Next prove requirement (ii) that implies \Cref{eq:infinite_mass}.
    \paragraph{Proof of the integral requirement (ii)  which implies \Cref{eq:infinite_mass}.}
    Recall $D=(0,\epsilon]^2$. We here recur to a decomposition similar to the first part of the proof, and rewrite the integral as
	\[
		\iint_{(0,1]^2} \ratemeasproposed(\de \vecvariantfreq{} )= \int_{(0,1]^2\setminus D} \ratemeasproposed(\de\vecvariantfreq{}) + \int_D \ratemeasproposed(\de\vecvariantfreq{}) >  \int_D \ratemeasproposed(\de\vecvariantfreq{}).
	\]

    We now show that in any neighborhood $D$ of the origin, the integral above diverges.  First, notice that on the set $D$, we can pointwise (lower) bound the L\'{e}vy rate of our proposed prior (\Cref{eq:proposed_prior}) as follows:
    \[
    	\ratemeasproposed(\de \btheta) I_{D}(\btheta) \ge  \kappa \frac{(\variantfreqperpop{1}+\variantfreqperpop{2}^{\rate{2}/\rate{1}})^{-\rate{1}}}{(\variantfreqperpop{1}+\variantfreqperpop{2})^{\corr{1}+\corr{2}}}\variantfreqperpop{1}^{\corr{1}-1}\variantfreqperpop{2}^{\corr{2}-1} I_D(\btheta) \de \variantfreqperpop{1} \de \variantfreqperpop{2},
    \]
    with 
    \[
    	\kappa := \frac{\mass}{B(\corr{1},\conc{1})B(\corr{2},\conc{2})} \min\{1, (1-\epsilon)^{\conc{1}-1}\} \min\{ 1, (1-\epsilon)^{\conc{2}-1}\},
    \]
    since on $D$, $\variantfreqperpop{p} < 1$. Hence,
    \begin{align*}
       \int_D \ratemeasproposed(\de\vecvariantfreq{}) \ge \int_{D \cap \{\variantfreqperpop{1} \ge \variantfreqperpop{2} \}} \nu (\de \btheta)  
        \ge \kappa \int_0^{\epsilon} \int_0^{\variantfreqperpop{1}} \frac{(\variantfreqperpop{1}+\variantfreqperpop{2}^{\rate{2}/\rate{1}})^{-\rate{1}}}{(\variantfreqperpop{1}+\variantfreqperpop{2})^{\corr{1}+\corr{2}}}\variantfreqperpop{1}^{\corr{1}-1}\variantfreqperpop{2}^{\corr{2}-1} \de \variantfreqperpop{1} \de \variantfreqperpop{2}.
    \end{align*}

    In the case in which $\rate{1} \le \rate{2}$  when $0\le\variantfreqperpop{2} \le \variantfreqperpop{1}\le1$ it is also the case that $\variantfreqperpop{2}^{\rate{2}/\rate{1}}\le \variantfreqperpop{1} $ (since $\rate{2}/\rate{1}\ge 1$). In turn,
    	$\variantfreqperpop{1}+\variantfreqperpop{2}^{\rate{2}/\rate{1}} \le 2 \variantfreqperpop{1}$ which implies $\left( \variantfreqperpop{1}+\variantfreqperpop{2}^{\rate{2}/\rate{1}}  \right)^{-\rate{1}} \ge \left( 2 \variantfreqperpop{1} \right)^{-\rate{1}}$. Similarly, $\variantfreqperpop{1} + \variantfreqperpop{2} \le 2 \variantfreqperpop{1}$ which implies $(\variantfreqperpop{1} + \variantfreqperpop{2})^{-\corr{1}-\corr{2}} \ge (2\variantfreqperpop{1})^{-\corr{1}-\corr{2}}$. The two observations together allow us to further bound the integral above:
	\begin{align*}
        \int_{(0,1]^2} \ratemeasproposed(\de\vecvariantfreq{}) \ge \int_D \ratemeasproposed(\de\vecvariantfreq{})&\ge \kappa \int_0^{\epsilon}\int_0^{\variantfreqperpop{1}} (2\variantfreqperpop{1})^{-\rate{1}-\corr{1}-\corr{2}}\variantfreqperpop{1}^{\corr{1}-1}\variantfreqperpop{2}^{\corr{2}-1} \de \variantfreqperpop{1} \de \variantfreqperpop{2} \\
        &= \kappa 2^{-\rate{1}-\corr{1}-\corr{2}} \int_0^{\epsilon} \variantfreqperpop{1}^{-\rate{1}-\corr{2}-1} \left[ \int_0^{\variantfreqperpop{1}} \variantfreqperpop{2}^{\corr{2}-1}\de \variantfreqperpop{2} \right] \de \variantfreqperpop{1} \\
        &=  \frac{\kappa 2^{-\rate{1}-\corr{1}-\corr{2}} }{\corr{2}} \int_0^{\epsilon}\variantfreqperpop{1}^{-\rate{1}-1} \de \variantfreqperpop{1} \\
        &= + \infty,
\end{align*}
where the last equality holds whenever $\rate{1} \ge 0$.

In the case in which $\rate{1} > \rate{2}$, when $0 \le \variantfreqperpop{1}\le\variantfreqperpop{2}\le1$ it is also the case that $\variantfreqperpop{2}^{\rate{2}/\rate{1}}\ge \variantfreqperpop{1} $ (since $\rate{2}/\rate{1}\le 1$). In turn,
$\variantfreqperpop{1}+\variantfreqperpop{2}^{\rate{2}/\rate{1}} \le 2 \variantfreqperpop{2}^{\rate{2}/\rate{1}}$ which implies $\left( \variantfreqperpop{1}+\variantfreqperpop{2}^{\rate{2}/\rate{1}}  \right)^{-\rate{1}} \ge \left( 2 \variantfreqperpop{2}^{\rate{2}/\rate{1}} \right)^{-\rate{1}}$. Similarly, $\variantfreqperpop{1} + \variantfreqperpop{2} \le 2 \variantfreqperpop{2}$ which implies $(\variantfreqperpop{1} + \variantfreqperpop{2})^{-\corr{1}-\corr{2}} \ge (2\variantfreqperpop{2})^{-\corr{1}-\corr{2}}$. 

The two observations together allow us to further bound the integral above:
	\begin{align*}
        \int_{(0,1]^2} \ratemeasproposed(\de\vecvariantfreq{}) \ge \int_D \ratemeasproposed(\de\vecvariantfreq{}) &\ge \kappa \int_0^{\epsilon}\int_0^{\variantfreqperpop{2}} (2\variantfreqperpop{2}^{\rate{2}/\rate{1}})^{-\rate{1}}(2\variantfreqperpop{2})^{-\corr{1}-\corr{2}}\variantfreqperpop{2}^{\corr{2}-1}\variantfreqperpop{1}^{\corr{1}-1} \de \variantfreqperpop{1} \de \variantfreqperpop{2} \\
        &= \kappa 2^{-\rate{1}-\corr{1}-\corr{2}} \int_0^{\epsilon} \variantfreqperpop{2}^{-\rate{2}-\corr{1}-1} \left[ \int_0^{\variantfreqperpop{2}} \variantfreqperpop{1}^{\corr{1}-1}\de \variantfreqperpop{1} \right] \de \variantfreqperpop{2} \\
        &=  \frac{\kappa 2^{-\rate{1}-\corr{1}-\corr{2}} }{\corr{1}} \int_0^{\epsilon}\variantfreqperpop{2}^{-\rate{2}-1} \de \variantfreqperpop{2} \\
        &= + \infty,
\end{align*}
where the last equality holds whenever $\rate{2} \ge 0$.

This concludes the proof of integral requirement (ii), in turn yielding the thesis.
\end{proof}

\subsection{Another model defines proper rate measure satisfies \Cref{req:finite_mean} and~\Cref{req:inf_mass}}

We emphasize that the prior distribution we proposed in \Cref{eq:proposed_prior} is not the \emph{only} prior that satisfies \Cref{req:finite_mean} and~\Cref{req:inf_mass}. Next, we provide another example of a prior satisfying such requirements.

\begin{myExample}[Hierarchical-3BP \citep{Masoero2021}]
    A hierarchical Bayesian nonparametric multi-population model satisfying \Cref{req:finite_mean} and~\Cref{req:inf_mass} was proposed by \cite{Masoero2021}.
     The model consists of two layers: first, a latent common ``super-population'' of frequencies pairs drawn from a Poisson Point Process (PPP) with L\'{e}vy measure given by the three parameter beta process (3BP):
	\[
    		\nu_{\mathrm{3BP}}(\de\theta) = \mass \frac{\Gamma(1+\conc{})}{\Gamma(1-\rate{})\Gamma(\conc{}+\rate{})} \theta^{-1-\rate{}}(1-\theta)^{\conc{}+\rate{}-1}1(\theta \in [0,1]) \de \theta.
	\] 
	Then, conditionally on the latent common frequencies, population specific frequencies are drawn using a fixed parametric form: 
    \begin{align*}
        \{\theta_{0,1},\theta_{0,2},\ldots,\}&\sim \mathrm{PPP}(\nu_{\mathrm{3BP}}(\de \theta))\\
        \theta_{1,k} \mid \theta_{0,k} \sim \Beta(\theta_1 \mid a_1\theta_{0,k} , b_1(1-\theta_{0,k}))\quad&\text{and} \quad
        \theta_{2,k} \mid \theta_{0,k}  \sim \Beta(\theta_2 \mid a_2\theta_{0,k} , b_2(1-\theta_{0,k})).
    \end{align*}

    Using the theory of marked Poisson processes, the population-specific rates are themselves Poisson point processes with rate measure given by
    \begin{align*}
        \mu(\de\btheta)&=\int_{0}^1 \prod_{j=1}^2 \Beta(\theta_j|a_j s, b_j(1- s))\nu_{\mathrm{3BP}}(\de s) \de \theta_1 \de \theta_2.
    \end{align*}
        The integral requirements (i), (ii) of \cref{eq:finite_first_moment,eq:infinite_mass} can be shown analytically.
    For (i), we here show finiteness of the integral in the first population. The argument for the second population is symmetric and omitted.
    \begin{align*}
        \iint \variantfreqperpop{1}\mu(\de\btheta)&=\iint\variantfreqperpop{1}\int_{s=0}^{s=1} \prod_{j=1}^2 \Beta(\theta_j|a_j s, b_j(1- s))\nu_{\mathrm{3BP}}(\de s)\de\btheta\\
        &=\int_{s=0}^{s=1} \iint \left\{ \variantfreqperpop{1}\prod_{j=1}^2 \Beta(\theta_j|a_j s, b_j(1-s))\de\btheta \right\}\nu_{\mathrm{3BP}}(\de s)\\
        &=\int_{0}^1\int\variantfreqperpop{1} \Beta(\variantfreqperpop{1}|a_1 s, b_1(1-s))\de\variantfreqperpop{1}\nu_{\mathrm{3BP}}(\de s)\\
        &=\int_0^1 \frac{a_1 s}{b_1+(a_1-b_1)s} \nu_{\mathrm{3BP}}(\de s) \\
        & < \frac{\mass}{\Gamma(1-\rate{})\Gamma(\conc{}+\rate{})\min\{a_1, b_1\}} \int s^{(1-\rate{})-1}(1-s)^{\conc{}+\rate{}-1} \de s < \infty.
    \end{align*}
    For (ii),
    \begin{align*}
        \iint \mu(\de\btheta)&=\iint \int_{s=0}^{s=1} \prod_{j=1}^2 \Beta(\theta_j|a_j s, b_j(1- s))\nu_{\mathrm{3BP}}(\de s)\de\btheta\\
        &=\int_{s=0}^{s=1} \iint \left\{ \prod_{j=1}^2 \Beta(\theta_j|a_j s, b_j(1-s))\de\btheta \right\}\nu_{\mathrm{3BP}}(\de s)\\
        &=\int_{0}^1 \nu_{\mathrm{3BP}}(\de s) = \infty.
    \end{align*}    
%    Indeed, the integral requirement of \cref{eq:infinite_mass} can be shown by first applying Tonelli theorem and then observing that the integral on $\theta_{j}$'s are factorized and the overall integral is the integral of 3BP mass.\footnote{\lom{let's spell this out in math.}}

\end{myExample}

\section{Conjugacy}
\label{app:conjugacy}
% !TEX root = ../double_trouble_supp.tex

We now show that our proposed prior introduced in \Cref{eq:proposed_prior} when paired with the two-population Bernoulli process model yields a posterior that is conjugate.
%\textcolor{red}{lom: proposal: should we have a theorem/proposition statement about the exact form of the posterior? It is weird to have a proof without a statement.}
\begin{myLemma}
	The multipopulation Bayesian nonparametric hierarchical model obtained by pairing the variants prior rate measure $\ratemeasproposed$ (of \Cref{eq:proposed_prior}) with independent Bernoulli process likelihoods within each subpopulation (i.e., $X_{\popidx, \unitidx, \variantidx} \mid \variantfreq{\popidx}{\variantidx}\sim \BER(\variantfreq{\popidx}{\variantidx})$, independently across $\popidx$ and $\variantidx$, and i.i.d.\ across $\unitidx$ for the same $\popidx, \variantidx$), yields a conjugate posterior distribution.%\textcolor{red}{lom: I think this would be clearer if we stated the model explicitly; in particular, $X_{p,n}$ follows the rates in population $p$, we should be clear about this?}
\end{myLemma}
\begin{proof}
    Recall that our proposed prior for variants frequencies $\ratemeasproposed$ was defined in \Cref{eq:proposed_levy}.
%    \begin{align*}
%        \ratemeasproposed(\de\vecvariantfreq{}) := \frac{\mass}{B(\corr{1}, \conc{1})B(\corr{2}, \conc{2})}
%			\frac{(\variantfreqperpop{1}+\variantfreqperpop{2}^{\rate{2}/\rate{1}})^{-\rate{1}}}{(\variantfreqperpop{1}+\variantfreqperpop{2})^{\corr{1}+\corr{2}}}
%			\variantfreqperpop{1}^{\corr{1}-1}\variantfreqperpop{2}^{\corr{2}-1}(1-\variantfreqperpop{1})^{\conc{1}-1}(1-\variantfreqperpop{2})^{\conc{2}-1} \de\btheta.
%    \end{align*}
    Our proof technique relies on the fact that, if we condition on an observed sample of $\samplesize{1}$ units from population $1$ and $\samplesize{2}$ units from population $2$, the \emph{posterior} distribution of the variants frequencies in our Bayesian hierarchical model can be obtained by thinning the Poisson process underlying the model. I.e., for variants $\psi_{\variantidx}$, $\ell=1,2,\ldots$, the corresponding variants' frequency (posterior) law $f_{\variantidx}(\vecvariantfreq{})$ can be obtained by thinning the rate measure $\ratemeasproposed$  by the probability of the realized sample induced by the likelihood component of our model:
%    For a variant $\psi_{\variantidx}$ observed in the first $\samplesize{1}+\samplesize{2}$ samples,
    \begin{align*}
        f_{\variantidx}(\vecvariantfreq{})&\propto\ratemeasproposed(\de\vecvariantfreq{})\cdot \prod_{\popidx=1,2} \prod_{\unitidx=1}^{\samplesize{\popidx}} \BER(x_{\popidx, \unitidx, \variantidx} \mid \variantfreqperpop{\popidx})\\
        &\propto \frac{(\variantfreqperpop{1}+\variantfreqperpop{2}^{\rate{2}/\rate{1}})^{-\rate{1}}}{(\variantfreqperpop{1}+\variantfreqperpop{2})^{\corr{1}+\corr{2}}}\variantfreqperpop{1}^{\corr{1}+\sum_{\unitidx=1}^{\samplesize{1}} x_{1,\variantidx, \unitidx}-1}\variantfreqperpop{2}^{\corr{2} +\sum_{\unitidx=1}^{\samplesize{2}} x_{2,\variantidx, \unitidx} -1}\\
        &  \qquad \qquad \qquad \qquad  \times(1-\variantfreqperpop{1})^{\conc{1}+\sum_{\unitidx=1}^{\samplesize{1}} (1-x_{1,\variantidx, \unitidx})-1}(1-\variantfreqperpop{2})^{\conc{2}+\sum_{\unitidx=1}^{\samplesize{2}} (1-x_{2,\variantidx, \unitidx})-1}.
    \end{align*}

    This distribution is within the same family as the one induced by the prior $\ratemeasproposed$ in the sense that the two share the same beta form, up to the same multiplicative leading term $\frac{(\variantfreqperpop{1}+\variantfreqperpop{2}^{\rate{2}/\rate{1}})^{-\rate{1}}}{(\variantfreqperpop{1}+\variantfreqperpop{2})^{\corr{1}+\corr{2}}}$. 
    Notice that for a given $\ell$, it is either the case that $\psi_\variantidx$ has been observed in at least one sample, or that is has not yet been observed. In the former case ($\psi_{\variantidx}$ observed), there exists at least one pair $(\popidx, \unitidx)$ for which $x_{\popidx, \unitidx, \variantidx}=1$. This implies that $f_{\variantidx}(\vecvariantfreq{})\le \variantfreqperpop{\popidx}\ratemeasproposed(\de\vecvariantfreq{})$. Since by assumption we have $\int \variantfreqperpop{\popidx}\ratemeasproposed(\de\vecvariantfreq{}) <\infty$, this further implies that the density above can be normalized.
    Viceversa, in the latter case (variant not yet observed), the distribution above simplifies to
    \begin{align*}
        f_{\variantidx}(\vecvariantfreq{})&\propto\ratemeasproposed(\de\vecvariantfreq{}) \prod_{\popidx=1,2} \prod_{\unitidx=1}^{\samplesize{\popidx}} \BER(0 \mid \variantfreqperpop{\popidx})\\
        &=\frac{\mass}{B(\corr{1}, \conc{1})B(\corr{2}, \conc{2})}
        \frac{(\variantfreqperpop{1}+\variantfreqperpop{2}^{\rate{2}/\rate{1}})^{-\rate{1}}}{(\variantfreqperpop{1}+\variantfreqperpop{2})^{\corr{1}+\corr{2}}}
        \variantfreqperpop{1}^{\corr{1}-1}\variantfreqperpop{2}^{\corr{2}-1}(1-\variantfreqperpop{1})^{\conc{1}+\samplesize{1}-1}(1-\variantfreqperpop{2})^{\conc{2}+\samplesize{2}-1} \de\btheta,
    \end{align*}
    which is again within the same family of our PPP prior, satisfying~\cref{req:finite_mean} and~\cref{req:inf_mass}. 
\end{proof}

\section{Competing Bayesian nonparametric baselines for prediction in multi-population settings}
\label{app:method_for_prediction}
% !TEX root = ../double_trouble_supp.tex

In this section we introduce in detail the two alternative Bayesian nonparametric baselines used in our experiments of \Cref{sec:simulation,sec:real_data} to benchmark our novel proposed approach when forming predictions for future variant counts in the presence of multiple populations.

\subsection{Dependent three parameter beta processes [d3BP]} \label{sec:d3BP}

%\lom{why do we call this model ``dependent''? we should justify this naming here}

The first competing model we consider is one where we ignore the population label, and ``pool'' our observations as if they were being generated from the same population. We call this (fully) dependent model because we can view it as sample one population's variant frequency and the other fully dependent on it. We use a three parameter beta process prior to model the variants' frequencies, and again assume that samples are drawn from Bernoulli processes conditionally given the latent variant frequencies. Formally, for $P_0$ an arbitrary diffuse measure on the space of variants' labels $\Psi$,
\begin{align}
\begin{split}
	\Theta &\sim \PPP(\nu_{\mathrm{3BP}}(\de \theta) \times P_0(\de \psi)) \\
	\genericmeasurevariantcount \mid\Theta &\sim \BeP(\randommeas). 
\end{split} \label{eq:d3BP}
\end{align}

where 
\begin{align*}
    \nu_{\mathrm{3BP}}(\de \theta)=\alpha \frac{\Gamma(1+c)}{\Gamma(1-\sigma)\Gamma(c+\sigma)}\theta^{-1-\sigma}(1-\theta)^{c+\sigma-1}1(\theta\in (0,1))\de\theta,
\end{align*}

with $X_{p,n}$ i.i.d.\ across $p$ and $n$.
This is equivalent to adopting the approach of \citet{Masoero2021} after discarding information about the individuals' population. 
Recall that, given a tuple of integers $\bm{k} = [k_1, k_2]^\top$, we call a ``$\bm{k}$-ton'' any variant appearing exactly $k_\popidx$ times in population $\popidx$.
%\footnote{\lom{Not clear. It read ``We use the same definition of k-tons as number of variants exists exact $k$ times.'' Please double check. Are we also sure we want to call this stuff $k$-ton? }}

\begin{myProposition}[Predicting the number of future $\bm{k}$-tons under the d3BP]
	%\footnote{\lom{conditions were missing. I re wrote this. please double check}}
	Assume to have collected $N_1$ samples $X_{1, 1:N_1}$ from population $1$, and $N_2$ samples $X_{2, 1:N_2}$ from population $2$ under the data generating process of \Cref{eq:d3BP}. Then the number of new variants that have not been observed in any of the first $N_1+N_2$ samples and appear exactly $\bm{k}$ times in additional $\bm{M} = [M_1, M_2]^\top$ samples collected from population $1$ and $2$ respectively is given by:
    \label{prop:kr_ton_diag}
    \begin{align*}
        \news{\bm{N}}{\bm{M}, \bm{k}} \mid X_{1,1:N_1},X_{2,1:N_2} \sim \Poisson\left(\lambda^{(\bm{M},\bm{k})}_{\mathrm{d3BP},\bm{N}}\right),
    \end{align*} 
    where, for $(a)_{b\uparrow}:=\Gamma(a+b)/\Gamma(a+1)$ the rising factorial,
    \begin{align*}
        \lambda^{(\bm{M},\bm{k})}_{\mathrm{d3BP},\bm{N}}=&\alpha \binom{M_1}{k_1}\binom{M_2}{k_2} \frac{(c+\sigma)_{(N_1+N_2+M_1+M_2-k_1-k_2)\uparrow }(1-\sigma)_{(k_1+k_2-1)\uparrow}}{(c+1)_{(N_1+N_2+M_1+M_2-1)\uparrow }}.
            \end{align*}
%    \lom{did we define the rising/falling factorials? (Pochammer symbols?)}
\end{myProposition}

The proof largely follows \citet{Masoero2021}, by observing that the variant frequencies not seen being a thinned Poisson point process. %by observing that the law of the \emph{total} number of new rare variants across both population coincides with the one induced by a simple, single three parameter beta process. The allocation of the variants across the two population is instead defined by the hypergeometric distribution.
%\lom{This proof is basically missing. I am happy writing it down, unless it was written elsewhere and got lost?}
\begin{proof}
    Under this model, variants' frequencies follow a 3BP with rate measure $\nu_{\mathrm{3BP}}(\de\theta)$. Then, the  frequencies of those variants that (i) have not been seen in the first $N_1+N_2$ samples $X_{1,1:N_1}$ and $X_{2,1:N_2}$ and (ii) will be seen exactly $\bm{k}$ times in the subsequent $\bm{M}$ samples, come from a thinned Poisson point process with rate measure $\nu_{\mathrm{3BP}}(\de\theta)\Bernoulli(0|\theta)^{\futuresamplesize{1}}\Bernoulli(0|\theta)^{\futuresamplesize{2}}$ and is independent of the collection of frequencies that did appear in the first $N_1,N_2$ samples, 
        \begin{align}
        \nu_{\mathrm{3BP}}(\de\theta)&\Bernoulli(0 \mid \theta)^{\samplesize{1}}\Bernoulli(0 \mid \theta)^{\samplesize{2}} 
        \nonumber\\
        &
        \binom{\futuresamplesize{1}}{k_1}\Bernoulli(1 \mid \theta)^{k_1}\Bernoulli(0 \mid \theta)^{\futuresamplesize{1}-k_1}
        \nonumber\\
        & 
        \binom{\futuresamplesize{2}}{k_2}\Bernoulli(1 \mid \theta)^{k_2}\Bernoulli(0 \mid \theta)^{\futuresamplesize{2}-k_2}\nonumber\\
        =&\binom{\futuresamplesize{1}}{k_1}\binom{\futuresamplesize{2}}{k_2}\nu(d\theta)\Bernoulli(1 \mid \theta)^{k_1+k_2}\Bernoulli(0 \mid \theta)^{\futuresamplesize{1}+\futuresamplesize{2}+\samplesize{1}+\samplesize{2}-k_1-k_2}.
        \end{align}
    By the properties of Poisson point processes, the number of $\bm{k}$-tons follows a Poisson distribution with expectation being the integral of the thinned rate measure. Following \citet[Proposition 4]{Masoero2021}, the integral is given by
    \begin{align*}
        \lambda^{(\bm{M}, \bm{k})}_{\mathrm{d3BP}, \bm{N}}=&\alpha \binom{\samplesize{1}}{k_1}\binom{\samplesize{2}}{k_2}\frac{(c+\sigma)_{(\futuresamplesize{1}+\futuresamplesize{2}+\samplesize{1}+\samplesize{2}-k_1-k_2)\uparrow }(1-\sigma)_{(k_1+k_2-1)\uparrow}}{(c+1)_{(\futuresamplesize{1}+\futuresamplesize{2}+\samplesize{1}+\samplesize{2}-1)\uparrow }}.
    \end{align*}
\end{proof}

We can compute the distribution of the total number of future variants by summing across $\bm{k}$-tons, or, equivalently, by leveraging the fact that the results for the pooled population coincide with those provided in \citet{Masoero2021}.

\subsection{i3BP} \label{sec:i3BP}

The alternative baseline method we consider is one in which in each of the two populations, variants frequencies follow a three parameter beta process, independent of the other population. In each population, variants occur according to independent Bernoulli processes.  In this case, the two populations are modeled as completely independent, and we therefore call this model the ``independent'' 3BP (i3BP):
\begin{align}
\begin{split}
	\Theta_\popidx &\overset{\text{ind}}{\sim} \PPP(\nu_{\mathrm{3BP}}(\de \theta; \mass_\popidx, c_\popidx, \sigma_\popidx) \times P_0(\de \psi)) \\
	X_{p,n}\mid\Theta &\sim \BeP(\Theta_\popidx).
\end{split} \label{eq:i3BP}
\end{align}
Here, we overload the notation $\nu_{\mathrm{3BP}}(\de \theta; \mass_\popidx, c_\popidx, \sigma_\popidx)$ to emphasize that in each population $\popidx$ the variants frequencies are driven by a different three-parameter beta process with its own hyperaprameters.
 Because of independency, and the fact that we use a diffuse measure for locations, the probability that the same variant is observed in both population under this model is zero. As a consequence, under this model the total number of \emph{shared} variants (past, present, future) is zero. Therefore, we can't use this model to make any prediction about shared variants, but we still could use it to predict the \emph{total} number of new future variants, by summing the predictions arising from each population and ignoring the fact that in the real data some of the past and future variants are shared. 
%We treat two population independently model as 3BP's following \citet{Masoero2021} and add the number of new variants together to get the discover curve. 

%\lom{do we want to have a proposition here like in the other subsection? we should be precise here... let's introduce the parameters etc}

\begin{myProposition}[total number in i3BP]
    Assume to have collected $N_1$ samples $X_{1, 1:N_1}$ from population $1$, and $N_2$ samples $X_{2, 1:N_2}$ from population $2$ under the data generating process of the i3BP in \Cref{eq:i3BP}. The variant frequencies in population 1 and 2 are governed by two independent 3BP, the first with hyperparameters $\alpha_1, \sigma_1, c_1$ and the second with hyperparameters $\alpha_2, \sigma_2, c_2$. The total number of new variants that is going to be observed in $\bm{M} = [M_1, M_2]^\top$ additional samples follows a Poisson distribution with parameter $\lambda^{(\bm{M})}_{\text{i3BP},\bm{N}}$ where 
    \begin{align*}
        \lambda^{(\bm{M})}_{\mathrm{i3BP},\bm{N}} = \alpha_1 \sum_{k_1=1}^{\futuresamplesize{1}} \frac{(c_1+\sigma_1)_{(\samplesize{1}+k_1-1)\uparrow}}{(c_1+1)_{(\samplesize{1}+k_1-1)\uparrow}}+ \alpha_2 \sum_{k_2=1}^{\futuresamplesize{2}} \frac{(c_2+\sigma_2)_{(\samplesize{2}+k_2-1)\uparrow}}{(c_2+1)_{(\samplesize{2}+k_2-1)\uparrow}}.
    \end{align*}
\end{myProposition}

\begin{proof}
    The result is a trivial application of \citet[Proposition 1]{Masoero2021}. Indeed, by \citet[Proposition 1]{Masoero2021},  the number of new variants in population $\popidx$ is  Poisson distributed with mean $\alpha_\popidx \sum_{k=1}^{\futuresamplesize{\popidx}} \frac{(c_\popidx+\sigma_\popidx)_{(\samplesize{\popidx}+k-1)\uparrow}}{(c_\popidx+1)_{(\samplesize{\popidx}+k-1)\uparrow}}$. Moreover, because the two populations are independent and the sum of two independent Poisson random variables is a Poisson random variable with parameter given by the sum of the two parameters, the thesis follows.
    \end{proof}

%\section{Method for fitting the hyperparameter}
%\label{app:method_for_fitting}
%\input{tex/fitting}
%\yunyi{this is no longer needed, as we did not use any new algorithm to fit hyperparameter in baseline methods}

\section{Proofs for the posterior predictive quantities in the two-population proposed model}
\label{app:detail_on_prediction}
% !TEX root = ../double_trouble_supp.tex

In this section we provide the proofs for the results stated in \Cref{sec:prediction_fixed}.

\subsection{Proof of Proposition \ref{prop:predict_kr_tons}}
\label{sec:kton_proof}
%Here we give the proof of Proposition \ref{prop:predict_kr_tons}. 

\begin{proof}
    We assume to have collected $\vecsamplesize= [\samplesize{1}, \samplesize{2}]^\top$ samples $X_{1,1:\samplesize{1}}, X_{2, 1:\samplesize{2}}$ from the model of \Cref{eq:generative_model} where $\ratemeasproposed$ is  the rate measure defined in \Cref{eq:proposed_levy}. It follows from the same method used in \citet[Proposition 2]{Masoero2022} on thinning of a Poisson point process, that, given the data $X_{1,1:\samplesize{1}}, X_{2,1:\samplesize{2}}$, the  frequencies corresponding to variants that (i) have not yet been observed in any of the first $\samplesize{1}+\samplesize{2}$ and are going to be observed exactly $\bm{k} = [k_1, k_2]^\top$ times in each population when additional $\vecfuturesamplesize = [\futuresamplesize{1}, \futuresamplesize{2}]^\top$ samples are collected, with $0 \le k_{\popidx} \le \futuresamplesize{\popidx}$ for $\popidx = 1, 2$ and $k_1+k_2>0$, follows a thinned Poisson point process with rate measure:
\begin{equation}
    \begin{aligned}
    \ratemeasproposed(\de\boldsymbol{\theta})&\Bernoulli(0 \mid \variantfreqperpop{1})^{\samplesize{1}}\Bernoulli(0 \mid \variantfreqperpop{2})^{\samplesize{2}}\\
    &\binom{\futuresamplesize{1}}{\occurrence{1}}\Bernoulli(1 \mid \variantfreqperpop{1})^{\occurrence{1}}\Bernoulli(0 \mid \variantfreqperpop{1})^{\futuresamplesize{1}-\occurrence{1}}\\
    & \binom{\futuresamplesize{2}}{\occurrence{2}}\Bernoulli(1 \mid \variantfreqperpop{2})^{\occurrence{2}}\Bernoulli(0 \mid \variantfreqperpop{2})^{\futuresamplesize{2}-\occurrence{2}}\\
    =&\mass\binom{\futuresamplesize{1}}{\occurrence{1}}\binom{\futuresamplesize{2}}{\occurrence{2}}\frac{(\variantfreqperpop{1}+\variantfreqperpop{2}^{\rate{2}/\rate{1}})^{-\rate{1}}}{(\variantfreqperpop{1}+\variantfreqperpop{2})^{\corr{1}+\corr{2}}}\\ 
    &\times\variantfreqperpop{2}^{\corr{2}+\occurrence{2}-1}(1-\variantfreqperpop{2})^{\conc{2}+\futuresamplesize{2}+\samplesize{2}-\occurrence{2}-1}\variantfreqperpop{1}^{\corr{1}+\occurrence{1}-1}(1-\variantfreqperpop{1})^{\conc{1}+\samplesize{1}-\occurrence{1}-1}\\
    &\times\frac{1}{B(\corr{1},\conc{1})B(\corr{2},\conc{2})} \label{eq:rate_news_two_pop}
    \end{aligned}
\end{equation}
 Thus thus number of these variants, exist exactly $\bm{k}$ times in $\vecfuturesamplesize$ new samples, $\news{\vecsamplesize}{\vecfuturesamplesize, \bm{k}}$, is a Poisson distributed random variable whose parameter is the integral of the rate measure in \Cref{eq:rate_news_two_pop}:
\begin{align}
    \begin{split}
      \genericfreqnewsparam =&\mass\binom{\futuresamplesize{1}}{\occurrence{1}}\binom{\futuresamplesize{2}}{\occurrence{2}}\frac{B(\corr{2}+\occurrence{2},\conc{2}+\futuresamplesize{2}+\samplesize{2}-\occurrence{2})B(\corr{1}+\occurrence{1},\conc{1}+\futuresamplesize{1}+\samplesize{1}-\occurrence{1})}{B(\corr{1},\conc{1})B(\corr{2},\conc{2})}\\
    &\times\E_{Z,W}\left[\frac{(Z+W^{\rate{2}/\rate{1}})^{-\rate{1}}}{(Z+W)^{\corr{1}+\corr{2}}}\right],
    \end{split} \label{eq:expected_ktons}
\end{align}
where $B(a,b)$ is the beta function, and $\E_{Z,W}$ is the expectation over independent beta distributed random variables $Z, W$:
\begin{align*}
    Z\sim \Beta(\corr{1}+\occurrence{1},\conc{1}+\futuresamplesize{1}+\samplesize{1}-\occurrence{1}), \quad 
    W\sim \Beta(\corr{2}+\occurrence{2},\conc{2}+\futuresamplesize{2}+\samplesize{2}-\occurrence{2}).
\end{align*}
\end{proof}

\subsection{Total number of variants}
\label{sec:proof_total_number}
Next, we give the proof of~\Cref{coro:new_vars}.

\begin{proof}
    The proof closely resembles the proof of \citet[Proposition 1]{Masoero2022}. The idea is to consider the corresponding frequencies of newly appeared variants in one sample that does not appear in previous samples. These frequencies follow a(thinned) Poisson point process. 
    %\lom{I tried to fix this proof, but I am not sure I agree with what comes next --- more generally I am not sure I understand the argument rigorously. Can you take a pass to polish first?}
    Under \Cref{eq:generative_model} where $\nu$ is the rate measure defined in \Cref{eq:proposed_levy}, the frequencies $\vecvariantfreq{\variantidx}$ follow a Poisson point process with rate measure $\ratemeasproposed(\de\vecvariantfreq{})$. 
    Under independent two population Bernoulli process (where each sample is an independent Bernoulli trails), given an individual from population $\popidx$, variant $\variantidx$ appears with probability $\variantfreq{\variantidx}{\popidx}$ independently across all samples of population 1 and 2. The probability a variant did not appear in those samples is $\BER(0 \mid \variantfreq{\variantidx}{1})^{\samplesize{1}}\BER(0 \mid \variantfreq{\variantidx}{2})^{\samplesize{2}}$. Therefore, the collection of variants that did not appear in the first $\vecsamplesize=(\samplesize{1}, \samplesize{2})$ are characterized by a collection of frequencies whose distribution is given by a thinned multivariate Poisson point process with rate $\ratemeasproposed(\de\vecvariantfreq{})\BER(0 \mid \variantfreqperpop{1})^{\samplesize{1}}\BER(0 \mid \variantfreqperpop{2})^{\samplesize{2}}$ and are independent to the collection of frequencies that did appear in the first $\vecsamplesize$ samples.
    
    Within this collection of variant frequencies, there are some will appear in the next sample from population 1 with probability $\BER(1|\variantfreqperpop{1})$ the collection of these variant frequencies is inturn a furthering thinning. That is, the variant frequencies whose corresponding variants did not appear in the first $\vecsamplesize=(\samplesize{1}, \samplesize{2})$ samples but is going to appear if we were to collect an additional sample from population 1 samples is characterized by a thinned Poisson point process with rate measure 
    $$
    \ratemeasproposed(\de\vecvariantfreq{})\BER(0|\variantfreqperpop{1})^{\samplesize{1}}\BER(0|\variantfreqperpop{2})^{\samplesize{2}} \BER(1|\variantfreqperpop{1})
    $$
    and independent of the collection of frequencies that did not appear in the first $(\samplesize{1}, \samplesize{2})$ samples.

    Similarly if the first follow up sample is from population 2 the (thinned) rate measure is  
    $$
    \ratemeasproposed(\de\vecvariantfreq{})\BER(0|\variantfreqperpop{1})^{\samplesize{1}}\BER(0|\variantfreqperpop{2})^{\samplesize{2}} \BER(1|\variantfreqperpop{2})
    $$
    and is independent of the collection of frequencies that did not appear in the first $(\samplesize{1}, \samplesize{2})$ samples.

    Recursively, suppose we first sample follow up samples from population 1, for $\followupidx{1}\ge 1$, the collection of variant frequencies corresponding to variants that did not appear in the first $(\samplesize{1}+\followupidx{1}-1, \samplesize{2})$ samples but then did appear in the $\followupidx{1}$th follow-up samples of population 1 comes from a thinned Poisson point process with rate measure
    \begin{align*}
        \ratemeasproposed(\de\vecvariantfreq{})\BER(0|\variantfreqperpop{1})^{\samplesize{1}+\followupidx{1}-1}\BER(0|\variantfreqperpop{2})^{\samplesize{2}} \BER(1|\variantfreqperpop{1})
    \end{align*}

    The number of such variants, by property of Poisson point process, is a Poisson distribution with mean being the integral of the rate measure, and by~\Cref{eq:expected_ktons}, it is exactly $\newsparam{(\samplesize{1}+\followupidx{1}-1,\samplesize{2})}{[1,0],[1,0]}$. 
    
    Now suppose we finished all $\futuresamplesize{1}$ follow up of population 1 and start the follow up sample of population 2. For $\followupidx{2}\ge 1$, the collection of variant frequencies corresponding to variants that did not appear in the first $(\samplesize{1}+\futuresamplesize{1}, \samplesize{2}+\followupidx{2}-1)$ samples but then did appear in the $\followupidx{2}$th follow-up samples of population 2 comes from a thinned Poisson point process with rate measure
    \begin{align*}
        \ratemeasproposed(\de\vecvariantfreq{})\BER(0|\variantfreqperpop{1})^{\samplesize{1}+\futuresamplesize{1}}\BER(0|\variantfreqperpop{2})^{\samplesize{2}+\followupidx{2}-1} \BER(1|\variantfreqperpop{2})
    \end{align*}

    The number of such variants is a Poisson distribution with mean being the integral of the rate measure, and by~\Cref{eq:expected_ktons}, it is exactly $\newsparam{(\samplesize{1}+\futuresamplesize{1},\samplesize{2}+\followupidx{2}-1)}{[0,1],[0,1]}$. 

    Each of these Poisson point processes are independent. Since the total number of new variants $\genericnews$ are sum of new variants discovered at these $\futuresamplesize{1}+\futuresamplesize{2}$ samples, which are independently Poisson distributed, the total number of new variants are also Poisson distributed with mean  
    \begin{align*}
        \sum_{\followupidx{1}=1}^{\futuresamplesize{1}}\newsparam{(\samplesize{1}+\followupidx{1}-1,\samplesize{2})}{(1,0),(1,0)}+\sum_{\followupidx{2}=1}^{\futuresamplesize{2}} \newsparam{(\samplesize{1}+\futuresamplesize{1},\samplesize{2}+\followupidx{2}-1)}{(0,1),(0,1)}
    \end{align*}
    The recursive scheme of population 2 following population 1 is not the only method to get the distribution, one can also start with population 2 then population 1 or any other recursive scheme to get $\vecfuturesamplesize=(\futuresamplesize{1}, \futuresamplesize{2})$.
\end{proof}

\section{Details on simulating data}
\label{app:taking_sample}
% !TEX root = ../double_trouble_supp.tex

\subsection{ Sampling from the proposed two population model via (truncated) Poisson point processes}
\label{app:sample_from_pois}

We here provide a practical algorithm to generate synthetic data from the proposed hierarchical model introduced in \Cref{eq:generative_model} where the L\'{e}vy rate measure is given in \Cref{eq:proposed_levy}.

At a high level, there exist two main strategies to generate data from nonparametric hierarchical models like the one in \Cref{eq:generative_model}. The first, is to devise an exact ``marginal'' representation, where the infinitude of latent atoms is integrated out and observations are generated by leveraging an iterative predictive ``marginal'' scheme (e.g., the Chinese restaurant process, the Indian buffet process and related ``urn-schemes''). The other approach is to rely on an approximate ``conditional'' representation, in which the infinitude of latent atoms is approximated via a truncated representation of the infinite measure \citep{campbell2019truncated}, and observations are generated conditionally i.i.d.\ on this finite representation.

In principle we can derive a marginal representation of the proposed process following \citet{Broderick2018}. However such marginal representation would require us to solve many numerical integrals for every sample to determine whether a pre existing variant should exist in a follow up sample. 
To avoid this computational overhead, we decide to adopt the approximate (truncated) approach: we draw a truncated, approximate Poisson point process and obtain a sample $X_{\popidx, \unitidx}$ by repeatedly sample Bernoulli trials using the sampled variant frequencies. 

%Evaluating the marginal representation of the in equation~\ref{eq:marginal_representation} is challenging due to the need of repeat evaluating of the expectations to decide if a previously seen variant show up in a new sample. Another method to sample the process is to first sample from the Poisson point process with some truncations. 

A general approach to sample inhomogeneous Poisson point process with an arbitrary target rate is to find a ``proposal'' rate measure dominating the target, and then adopt rejection sampling \citep{saltzman2012simulating}. 
For our specific proposal, we truncated the domain to be larger than $10^{-10}$, variants with frequency less than this will have little probability to exist in 500 samples as in our simulation settings. % such that in our simulation with 500 samples the chance of observing the variant being $10^{-7}$. \footnote{\lom{not clear}}

To sample this truncated measure we chose a proposal process being piecewise constant that dominates our (truncated) rate. To do so, we span a (log-scaled) mesh using $\{10^{-10},10^{-9},\dots, 10^{0}\}$ and we sample points within each grid independently. In each grid $\mathcal{G}$, we sample a homogeneous Poisson point process with rate being the maximum rate in that grid ($\lambda_m=\max_{\vecvariantfreq{}\in \mathcal{G}} \ratemeasproposed(\vecvariantfreq{})$) and we accept a sample at $\vecvariantfreq{\variantidx}$ with probability $\ratemeasproposed(\vecvariantfreq{\variantidx})/\lambda_m$.  

These procedure will generate an approximated realization of $\vecvariantfreq{\variantidx}$, then we sample the Bernoulli process condition on these weights.

\subsection{d3BP}
To obtain a d3BP-Bernoulli process of size $\vecsamplesize = [\samplesize{1}, \samplesize{2}]^\top$ we first obtain $\samplesize{1} + \samplesize{2}$ samples from a 3BP-Bernoulli process using the Indian buffet representation \citep{Broderick2012,Masoero2021}. Then, we randomly split individuals into two groups of size $\samplesize{1}$ and $\samplesize{2}$.

\section{Additional experiments}
\label{app:additional_experiments}
% !TEX root = ../double_trouble_supp.tex

In this section we provide additional experimental evidence of the predictive performance on synthetic and real genomic data of our newly proposed BNP model as well as a number of preexisting competing (single population) methods.
Specifically, we consider the d3BP, i3BP (as defined in \Cref{sec:d3BP,sec:i3BP}), as well as the Good-Toulmin (GT) estimator --- recently employed by \citet{chakraborty2019using} in the context of rare genomic variants discovery --- and the fourth order Jackknife estimator (J4) \citet{gravel2014predicting}. As well as linear programming \citep[lp,]{zou2016quantifying}, and the scaled process (SP) method by \citet{Camerlenghi2021}. In real data experiments we choose the best performing single population methods and apply it in two different ways namely assuming populations do not share variants or are in fact from a single population.

\subsection{Additional synthetic experiments}
\label{sec:misspecified}

Here we generate data from i3BP and d3BP respectively to test if our model has the ability to handle data generated from these misspecified models. For i3BP, we generate data using three parameter Indian buffet process, \citep{Broderick2012, Masoero2022}, assuming there is no shared variants. For d3BP, we generate variants from a single Indian buffet process then randomly split individuals as two populations. 

%\tb{Add summary here about how we generate data according to the independent and dependent extensions of single-population models}

\paragraph{Data from i3BP}

For i3BP, we sample from parameter $(\alpha_1,c_1,\sigma_1,\alpha_2,c_2,\sigma_2)=(20,1,.6,20,1,.3)$, and $(40,1,.1,40,1,.3)$ where $(\alpha_\popidx, c_\popidx, \sigma_\popidx)$ are the corresponding 3BP's mass, concentration and discount parameters in population $\popidx$. 

\begin{figure}[htp]
    \centering
    \includegraphics[width = 0.9\textwidth]{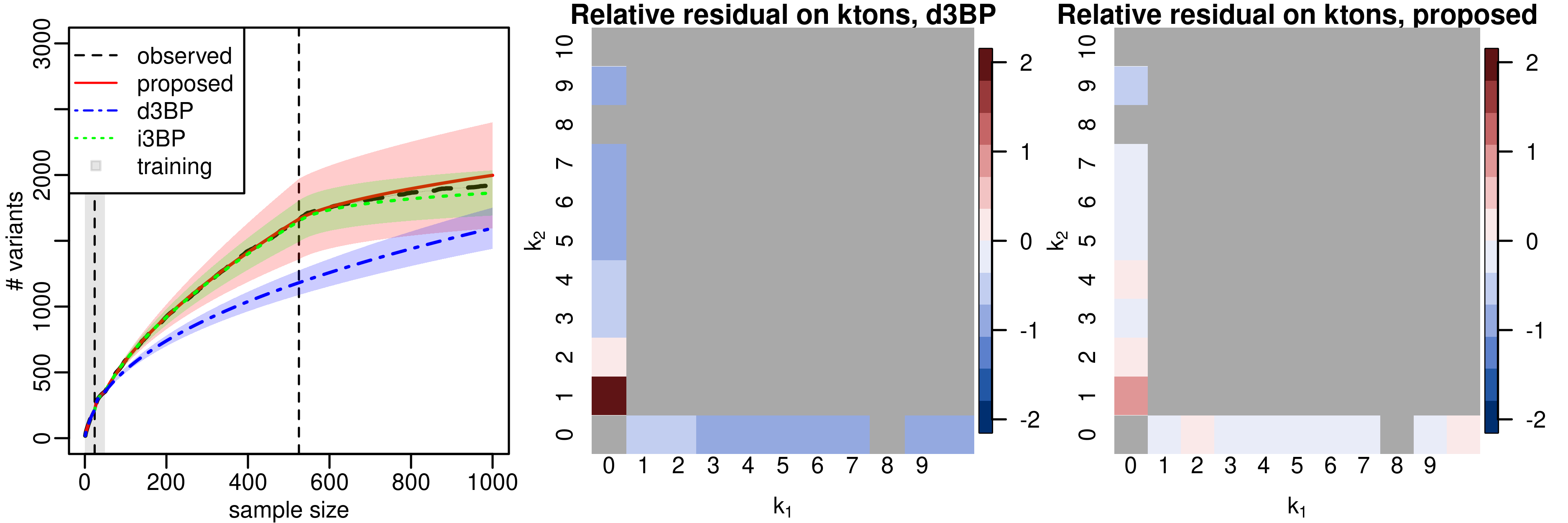}
    \includegraphics[width = 0.9\textwidth]{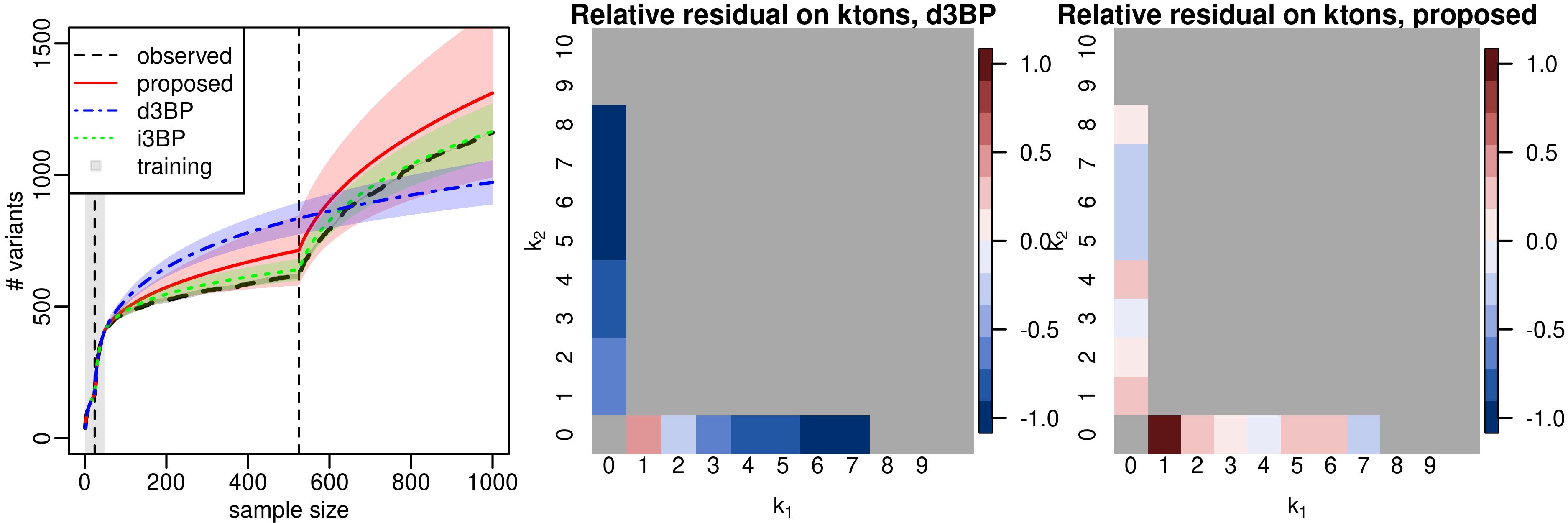}
    \caption{Prediction when data is sampled from i3BP. Each row represents a different simulated dataset, as described in \cref{sec:simulation}. \emph{Left}: The dashed black line shows the observed number of variants as a function of the sample number. The samples appear in the following order: the pilot from population 1, the pilot from population 2, the follow-up from population 1, the follow-up from population 2. A vertical dashed gray line separates the two populations in the pilot and follow-up, respectively. The gray shading covers the pilot data. All curves agree exactly on the pilot data. In the follow-up region, we plot mean (lines) and 1 standard deviation intervals (shaded regions) from three methods: our proposed method (solid red), the dependent version of the single-population 3BP (dash-dot blue), and the independent version of the 3BP (dotted green).  \emph{Center and Right}: For each $\vecoccurrence$ with component values up to 10, we plot the relative residual for predictions from our method (right) and the dependent 3BP (center). Note that the color scales are fixed across a row but vary across a column. A square is gray when the observed value is strictly less than 2 in at least one fold and black for $\vecoccurrence = (0,0)$.  
    }
    \label{fig:i3BP}
\end{figure}

\Cref{fig:i3BP} visualize the ground truth growth curve and predicted growth curve by all three models we tested when data is from i3BP, as well as the heat map of relative residual on the number of k-tons. In this setting our proposed method works as well as the correctly specified i3BP model on average and had larger variance, potentially due to the wrong assumption that all variants will eventually become shared. Meanwhile d3BP vastly underestimate number of ktons especially for $k_1,k_2>1$ as it assumes variants share the same frequencies. The overestimating of singletons in d3BP methods is not surprising since singletons are more compatible with the assumptions of all variants share the same variants and by chance we discovers it in one population. 

\paragraph{d3BP data}

We sample d3BP with parameter $(\alpha,c,\sigma)=(20,1,.5),(40,1,.1)$  respectively.

\begin{figure}[htp]
    \centering
    \includegraphics[width = 0.9\textwidth]{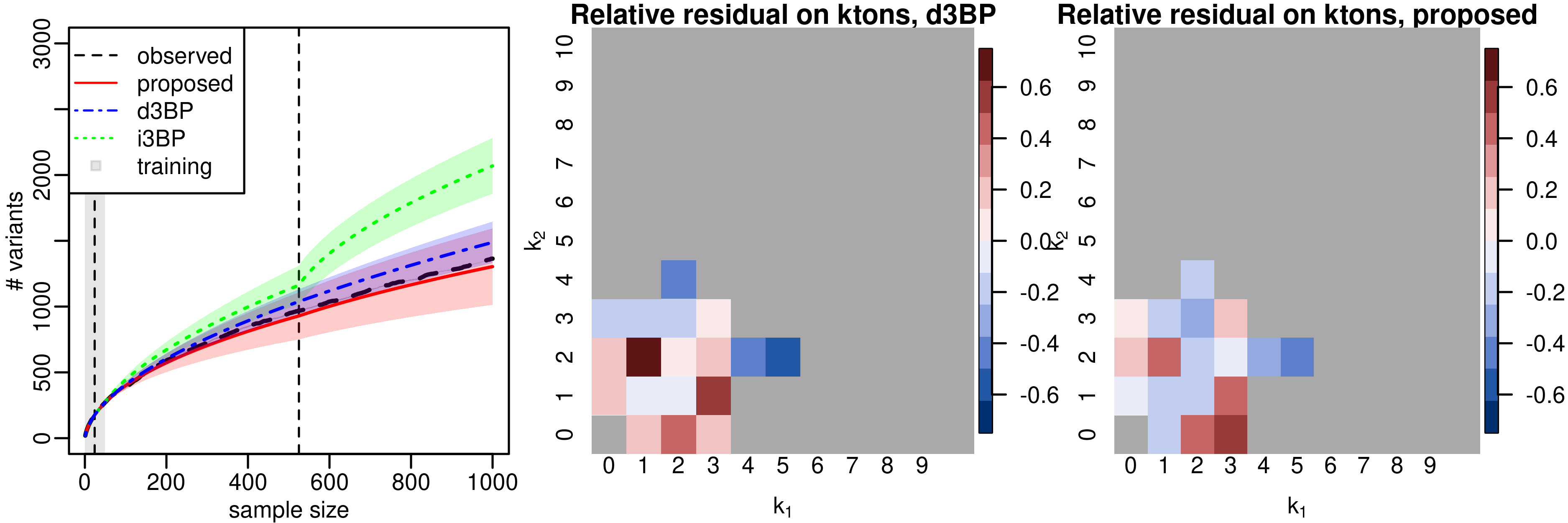}
    \includegraphics[width = 0.9\textwidth]{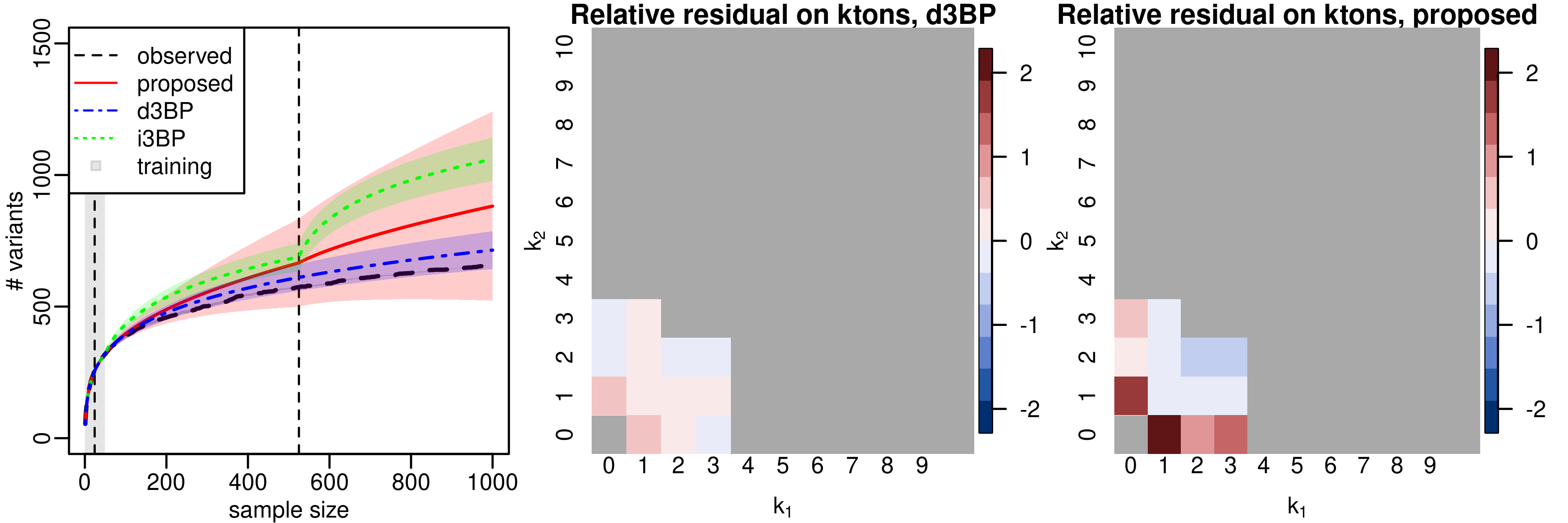}
    \caption{Prediction when data is sampled from d3BP. Each row represents a different simulated dataset, as described in \cref{sec:simulation}. \emph{Left}: The dashed black line shows the observed number of variants as a function of the sample number. The samples appear in the following order: the pilot from population 1, the pilot from population 2, the follow-up from population 1, the follow-up from population 2. A vertical dashed gray line separates the two populations in the pilot and follow-up, respectively. The gray shading covers the pilot data. All curves agree exactly on the pilot data. In the follow-up region, we plot mean (lines) and 1 standard deviation intervals (shaded regions) from three methods: our proposed method (solid red), the dependent version of the single-population 3BP (dash-dot blue), and the independent version of the 3BP (dotted green).  \emph{Center and Right}: For each $\vecoccurrence$ with component values up to 10, we plot the relative residual for predictions from our method (right) and the dependent 3BP (center). Note that the color scales are fixed across a row but vary across a column. A square is gray when the observed value is strictly less than 2 in at least one fold and black for $\vecoccurrence = (0,0)$. }
    \label{fig:d3BP}
\end{figure}

\Cref{fig:d3BP} visualize the ground truth growth curve and predicted growth curve by all three models we tested when data is from d3BP, as well as the heat map of relative residual on the number of k-tons. The data has relatively large amount of sharing thus we see a strong double counting problem faced by i3BP method that assumes no sharing. Our proposed method also struggle when the power law rate is small with perfect correlation. This is understandable since our method restrict the form of correlation and does not put mass perfectly along the diagonal especially for relatively common variants (\Cref{fig:log_rate_sigma}, third panel). Specifically when the power-law rate is low, the total number of variants are not dominated by extremely rare singletons and our method's mispecification for relatively common variants could struggle.

\subsection{Performance of single-population models on real data}

\label{sec:single_pop_methods}

We start by providing additional results on the performance of the array of model considered at the task of predicting the total number of genomic variance on real data when a single population is present.
For these comparisons, we tested fourth-order jackknife \citep[4jk,]{gravel2014predicting}, the Good-Toulmin estimator \citep[GT,]{Orlitsky2016,chakraborty2019using}, linear programming \citep[lp,]{zou2016quantifying}, and the 3BP approach of \citet{Masoero2022} as well as the scaled process (SP) method by \citet{Camerlenghi2021}.
We consider both data coming from the GnomAD dataset \citep{karczewski2020mutational}, as well as cancer data from the cancer genome atlas (TCGA) and the MSK-impact dataset \citep{cheng2015memorial}.
We observe that 3BP outperforms jackknife GT, and linear programming method in cancer datasets and perform very similarly with scaled process method which has larger variance. Jackknife performs better than other competing methods in GnomAD datasets. 

We provide a summary of our findings with results on a few cancer types in \Cref{fig:brca_luad_singlepop} and for a few subpopulations in the GnomAD dataset in \Cref{fig:seu_bgr_singlepop}. 

\begin{figure}[H]
    \centering
    \includegraphics[width = 0.6\textwidth]{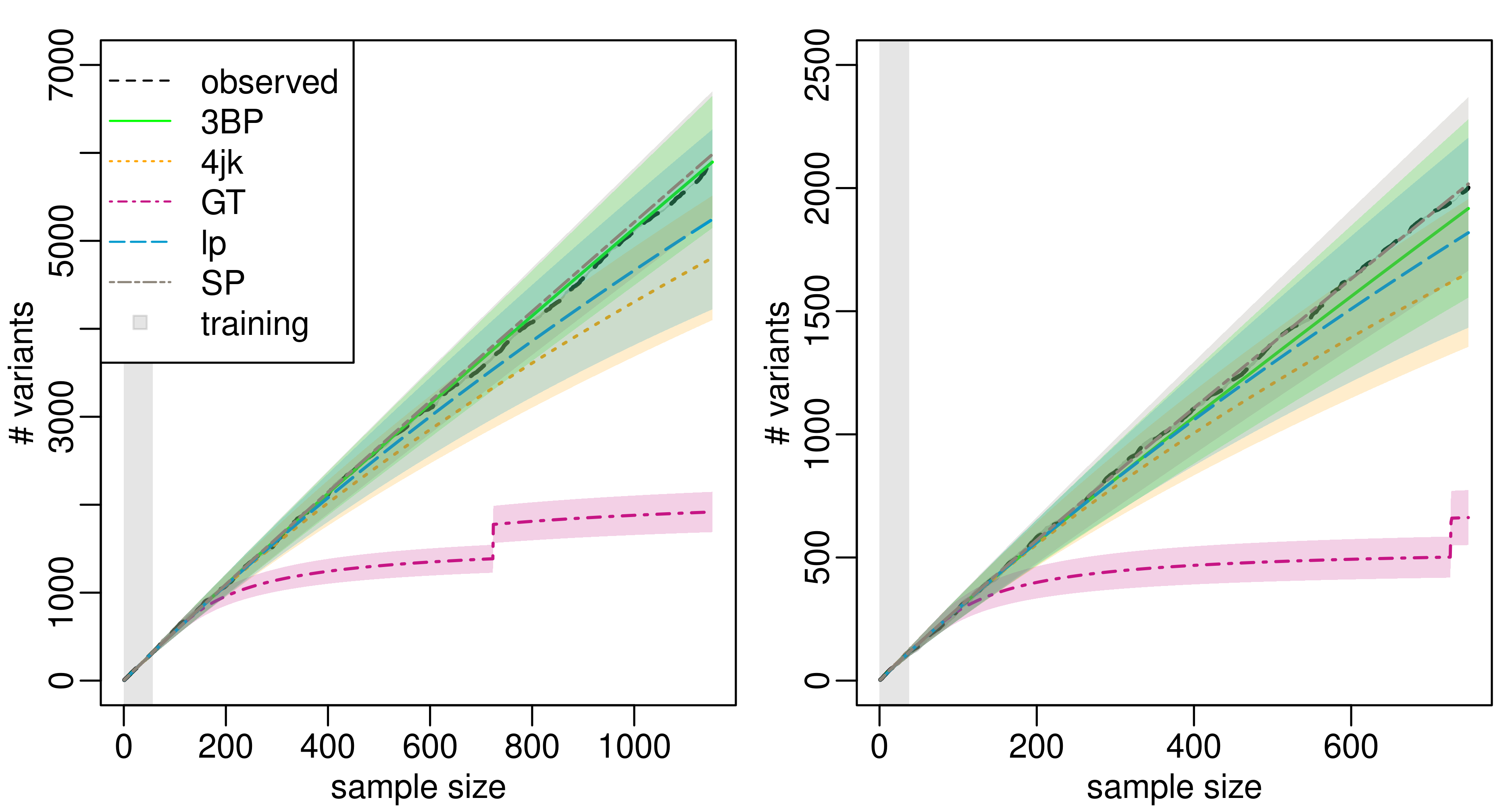}
    \includegraphics[width = 0.6\textwidth]{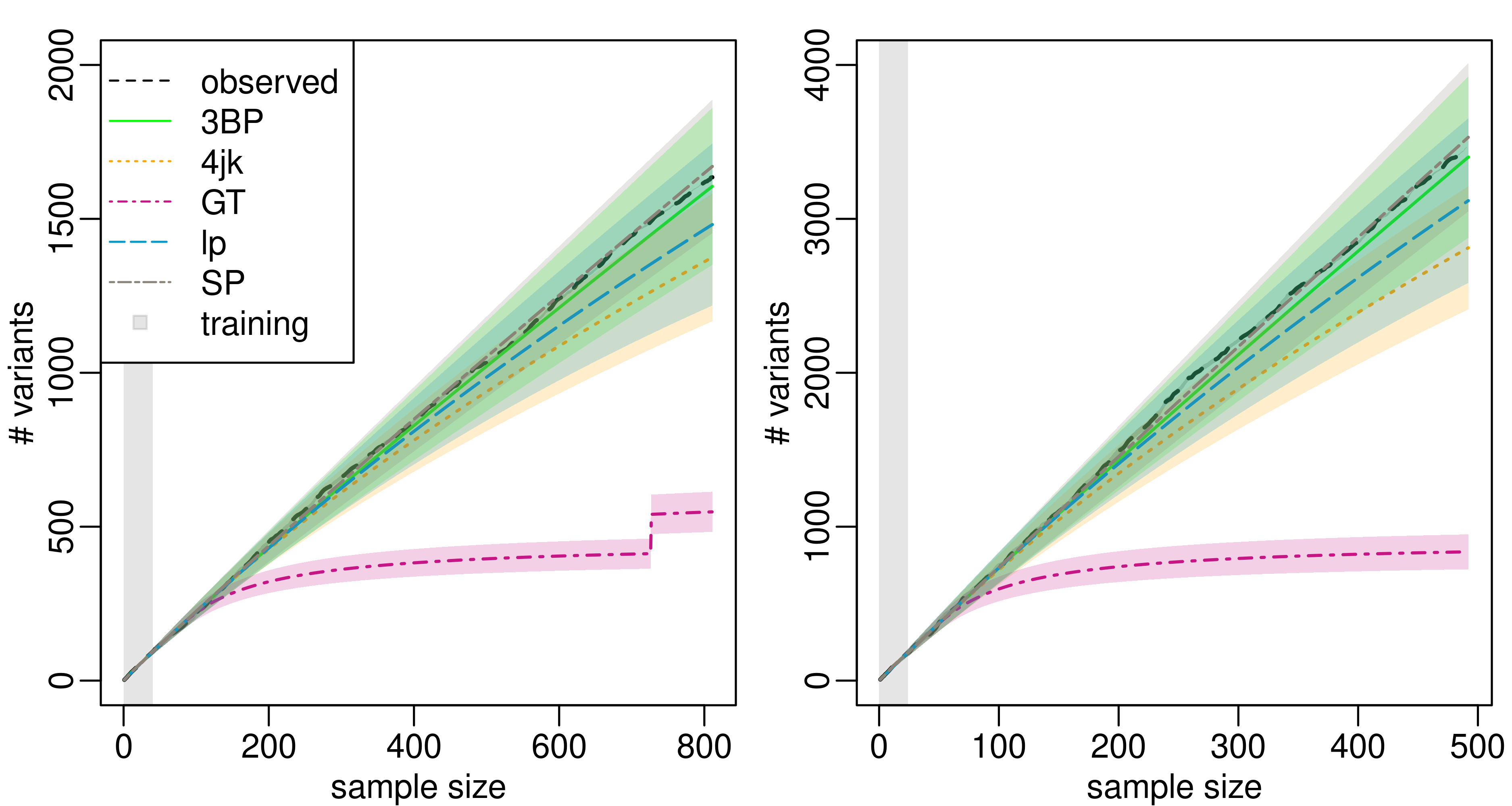}
    \caption{Prediction for the variant counts in, breast (left column) and lung cancer (right column). For  MSK-IMPACT dataset (first row) and TCGA dataset (second row) respectively. Shaded area being mean $\pm$ sd. The 3BP method outperforms alternatives in this dataset. }
%    \label{fig:bidc_la2_singlepop}
    \label{fig:brca_luad_singlepop}
\end{figure}

\begin{figure}[H]
    \centering
    \includegraphics[width = 0.9\textwidth]{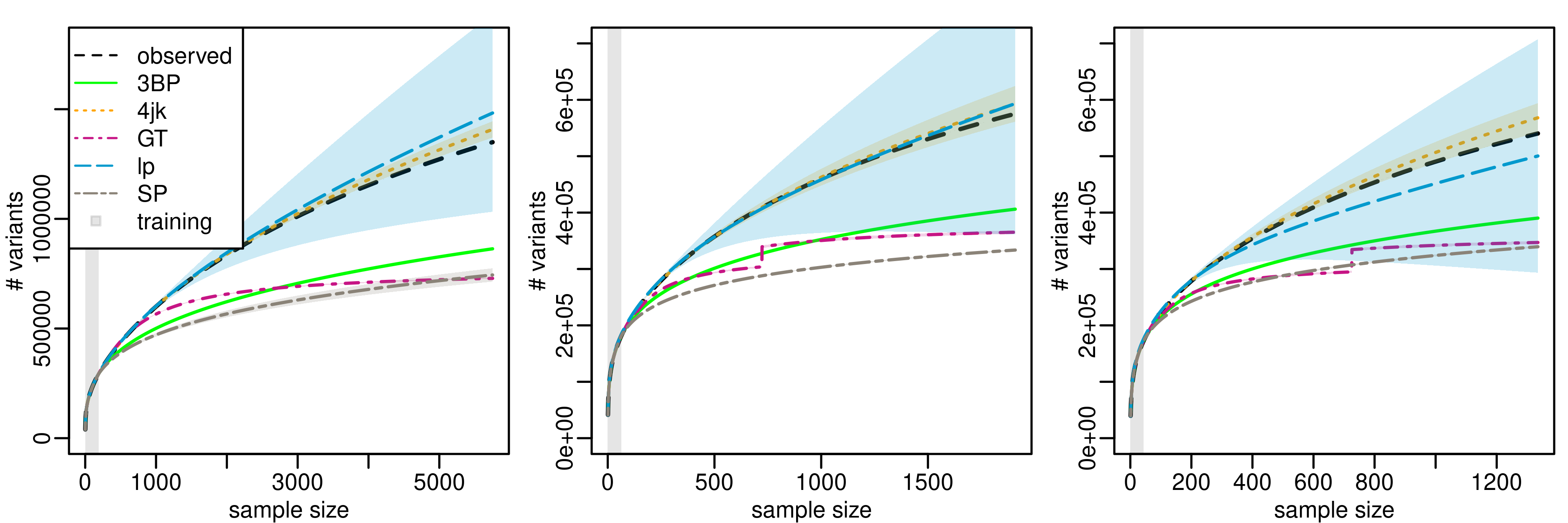}
    \caption{Prediction for the GnomAD dataset, Left column: Southern European, Middle column: Korean,  Right column: Bulgarian, shaded area being mean $\pm$ sd. The fourth order Jackknife outperforms competing alternatives in this dataset. 
    }
    %\label{fig:seu_bgr_singlepop}
    \label{fig:seu_bgr_singlepop}
\end{figure}

\subsection{Raw number of shared variants and residuals in simulation and real data experiments}
\label{sec:raw_responses}

In this section we report raw number of shared variants and raw magnitude of residuals without normalization as in the main text. The general observations are similar to that of normalized residuals and our proposed methods performs well overall.

\subsubsection{Simulations}
\label{sec:raw_responses_simulated}

\Cref{fig:proposed_kton} to~\Cref{fig:d3BP_kton} showed the raw number of k-tons and residuals with simulated data. The observations are similar to those of normalized results that our model could do as well as well specified i3BP and may struggle when d3BP with low power law rate was used to simulate data. 

\begin{figure}[H]
    \centering
    \includegraphics[width = 0.9\textwidth]{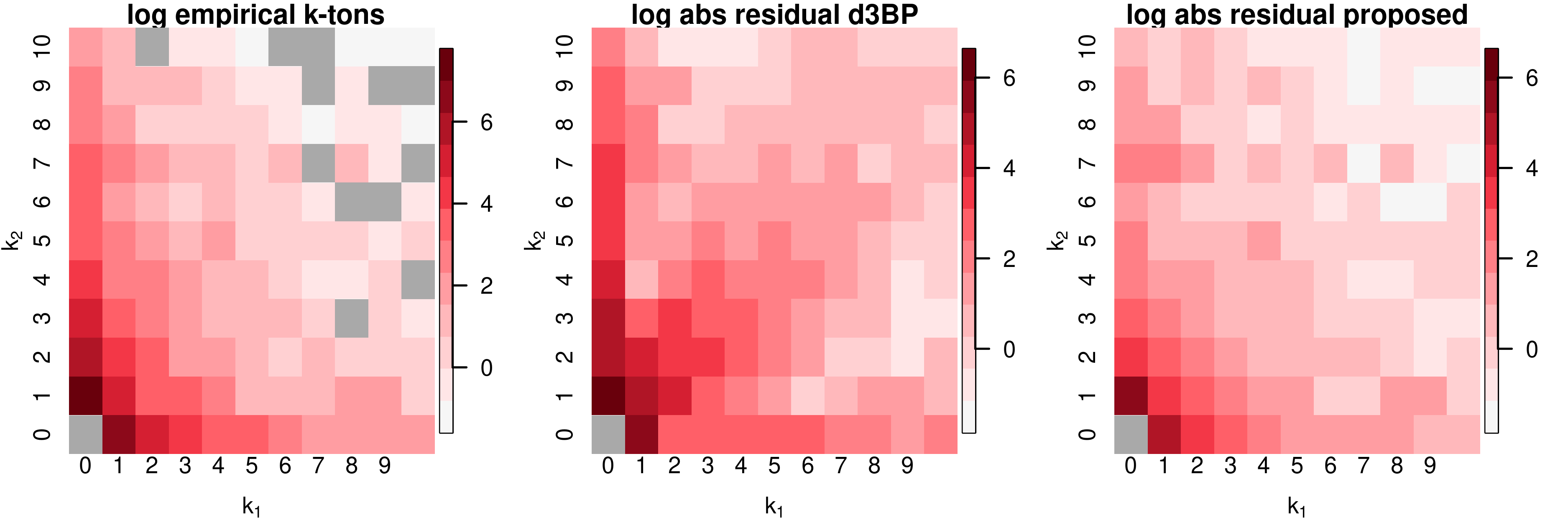}
    \includegraphics[width = 0.9\textwidth]{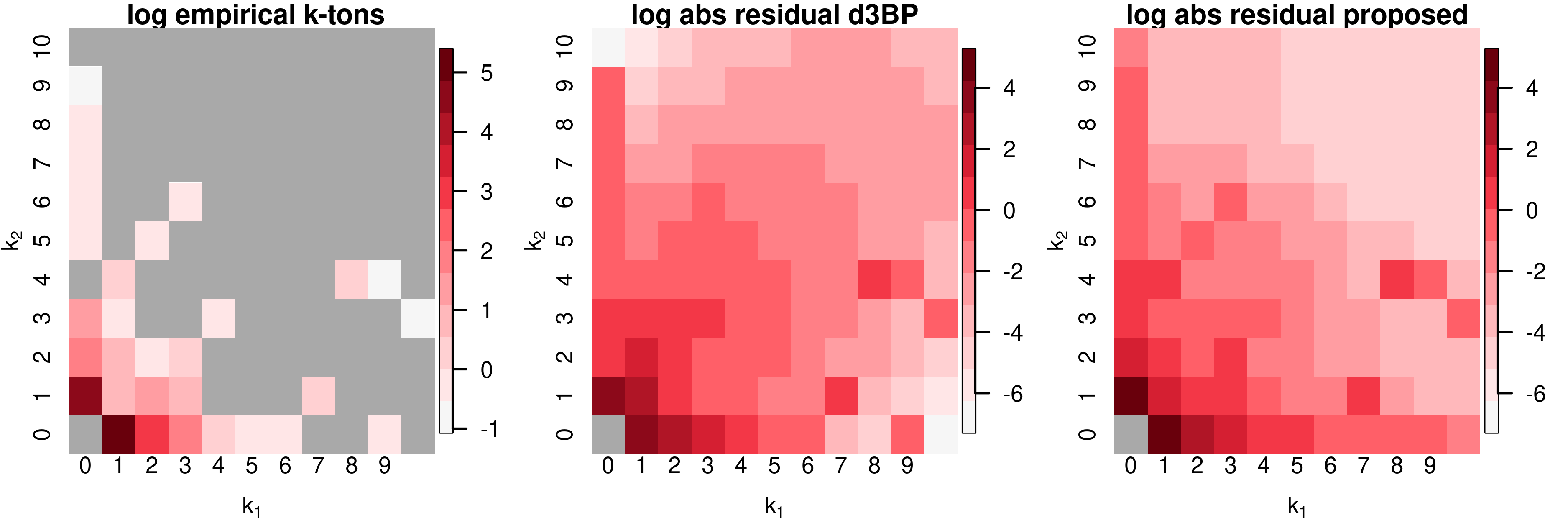}
    \caption{Ground truth of ktons and model prediction when the data sampled from the proposed model. Left most:log over average observed ktons. A square is gray when it the observed value is 0 for all folds or $\vecoccurrence=(0,0)$. Middle and right: Log average absolute residual in d3BP and proposed model.  %$(\mass,\rate{1}, \rate{2},\corr{1},\corr{2}, \conc{1}, \conc{2})=(100,0.4,0.6,0.5,0.5,1,1)$ and $(\mass,\rate{1}, \rate{2},\corr{1},\corr{2}, \conc{1}, \conc{2})=(1,0.8,0.7,0.5,0.5,1,1)$.
    }
    \label{fig:proposed_kton}
\end{figure}

\begin{figure}[H]
    \centering
    \includegraphics[width = 0.9\textwidth]{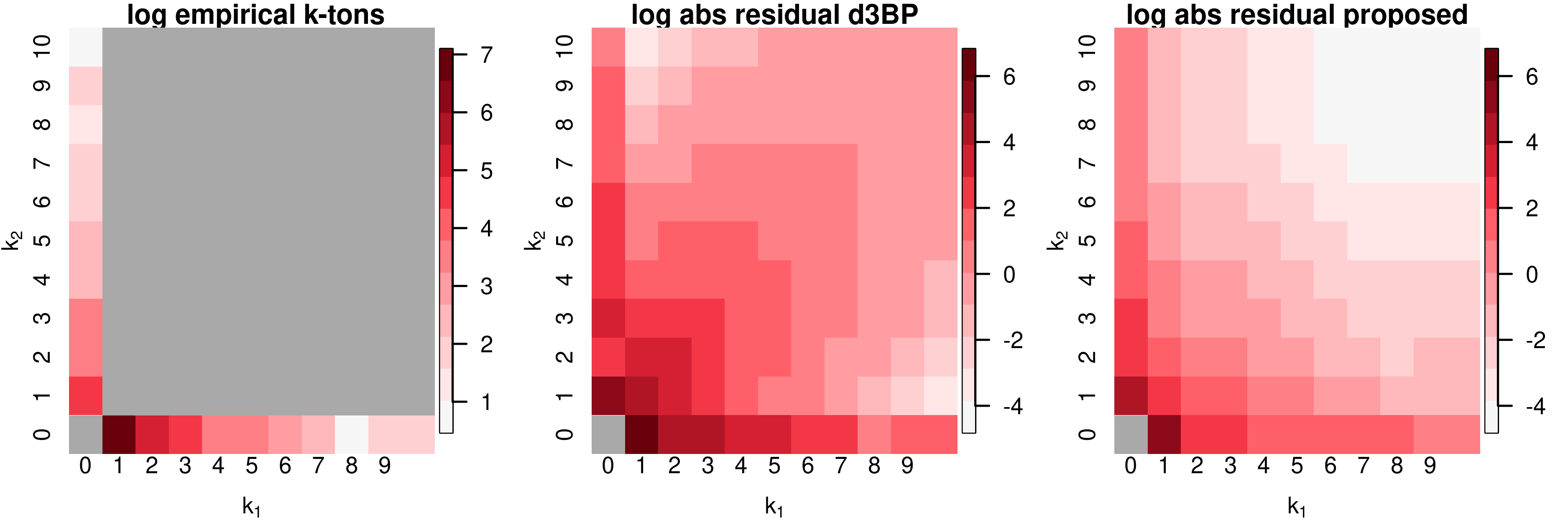}
    \includegraphics[width = 0.9\textwidth]{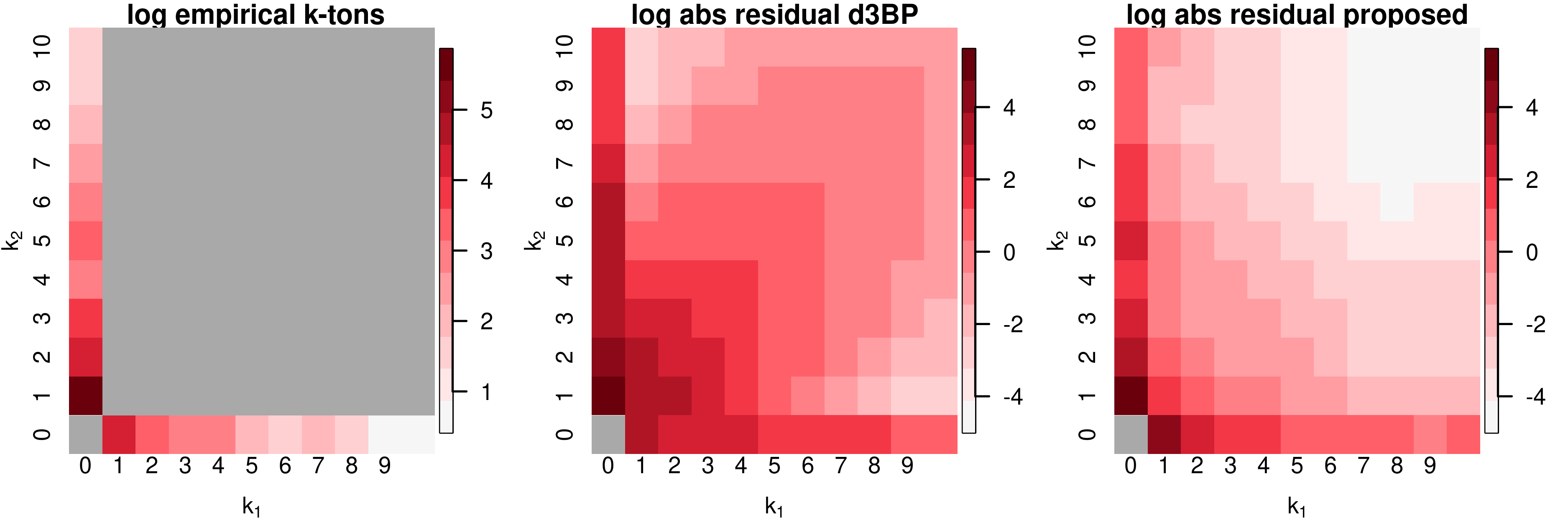}
    \caption{Ground truth of ktons and model prediction when the data sampled from i3BP. Left most:log over average observed ktons. A square is gray when it the observed value is 0 for all folds or $\vecoccurrence=(0,0)$. Middle and right: Log average absolute residual in d3BP and proposed model. 
    }
    \label{fig:i3BP_kton}
\end{figure}

\begin{figure}[H]
    \centering
    \includegraphics[width = 0.9\textwidth]{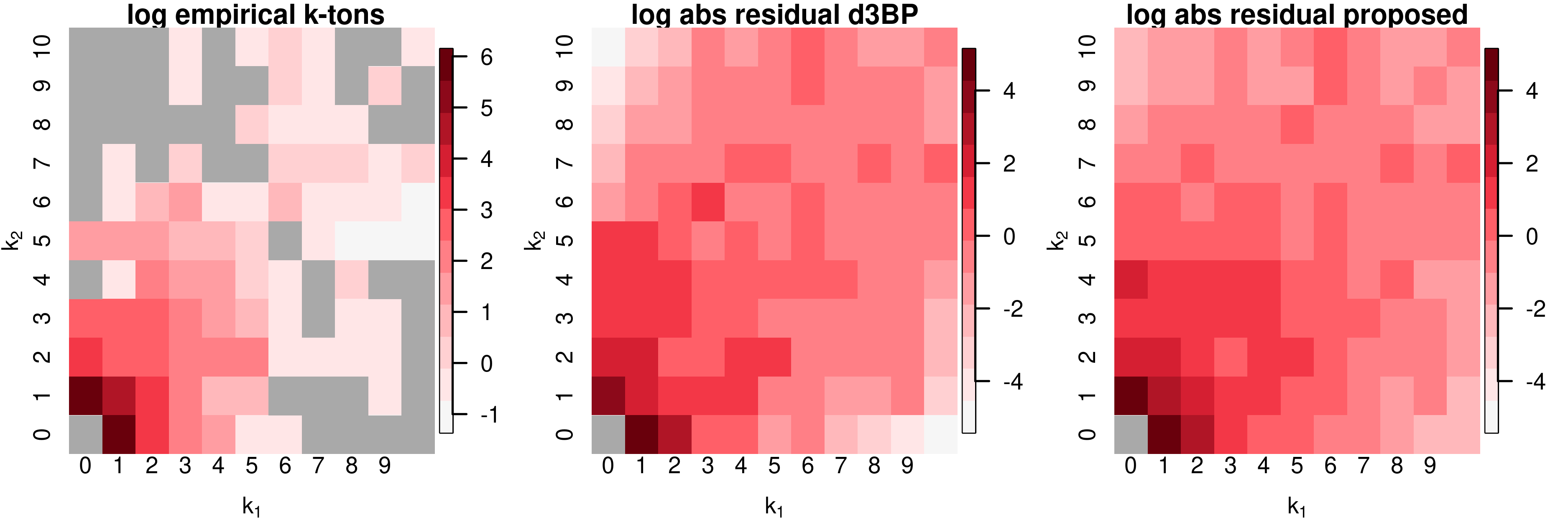}
    \includegraphics[width = 0.9\textwidth]{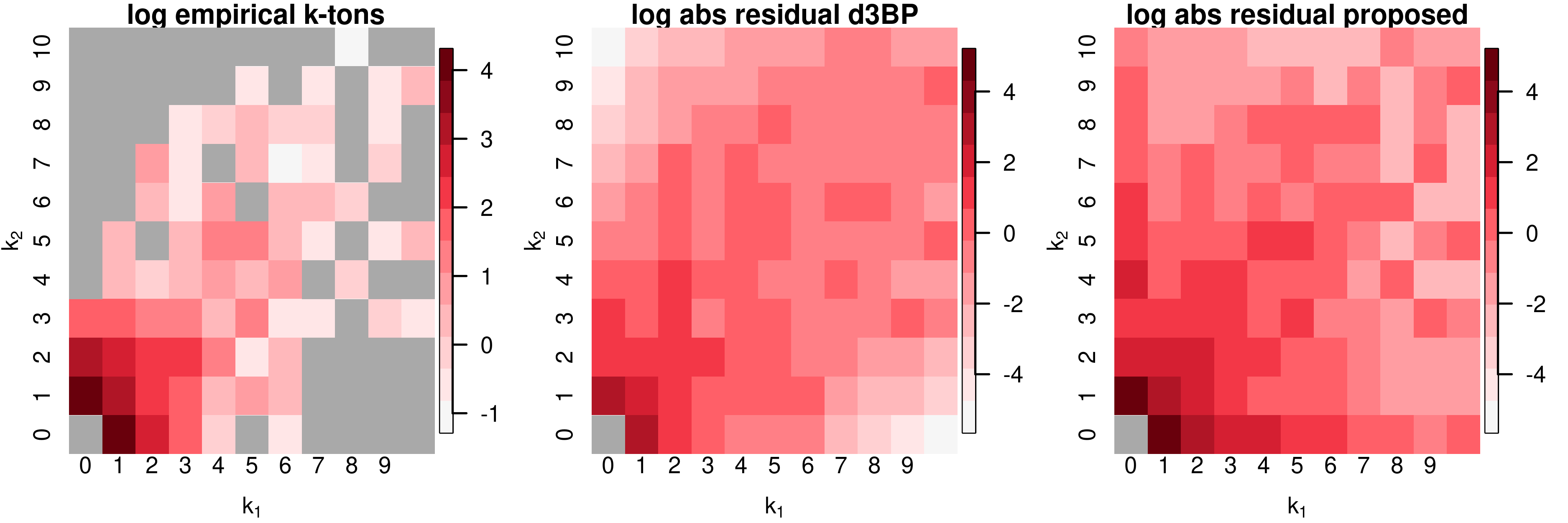}
    \caption{Ground truth of ktons and model prediction when the data sampled from d3BP. Left most:log over average observed ktons. A square is gray when it the observed value is 0 for all folds or $\vecoccurrence=(0,0)$. Middle and right: Log average absolute residual in d3BP and proposed model.}
    \label{fig:d3BP_kton}
\end{figure}

\subsection{Real data experiment}

\Cref{fig:brca_luad_kton} to~\Cref{fig:seu_bgr_kton} showed the raw number of k-tons that exists more than twice in all samples and magnitude of raw residuals with simulated data. We have the similar observation as normalized version that d3BP methods tend to have larger residual around the diagonal since by assumption of d3BP all variants are shared with the same frequencies. 

\begin{figure}[H]
    \centering
    \includegraphics[width = 0.9\textwidth]{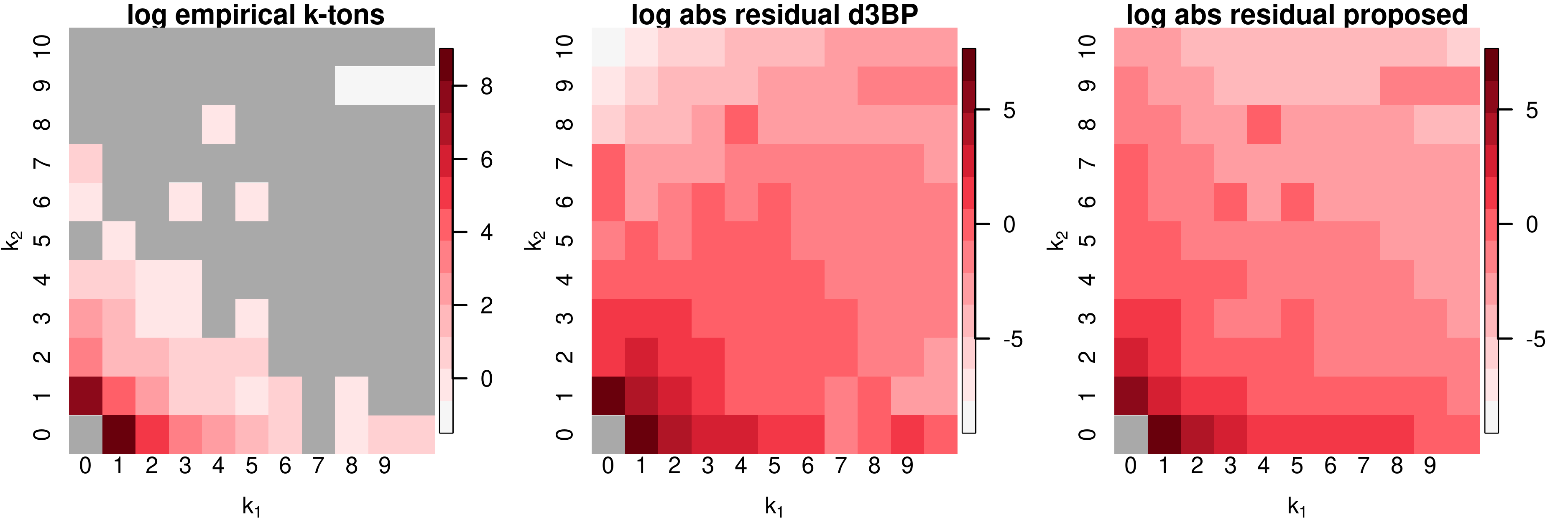}
    \includegraphics[width = 0.9\textwidth]{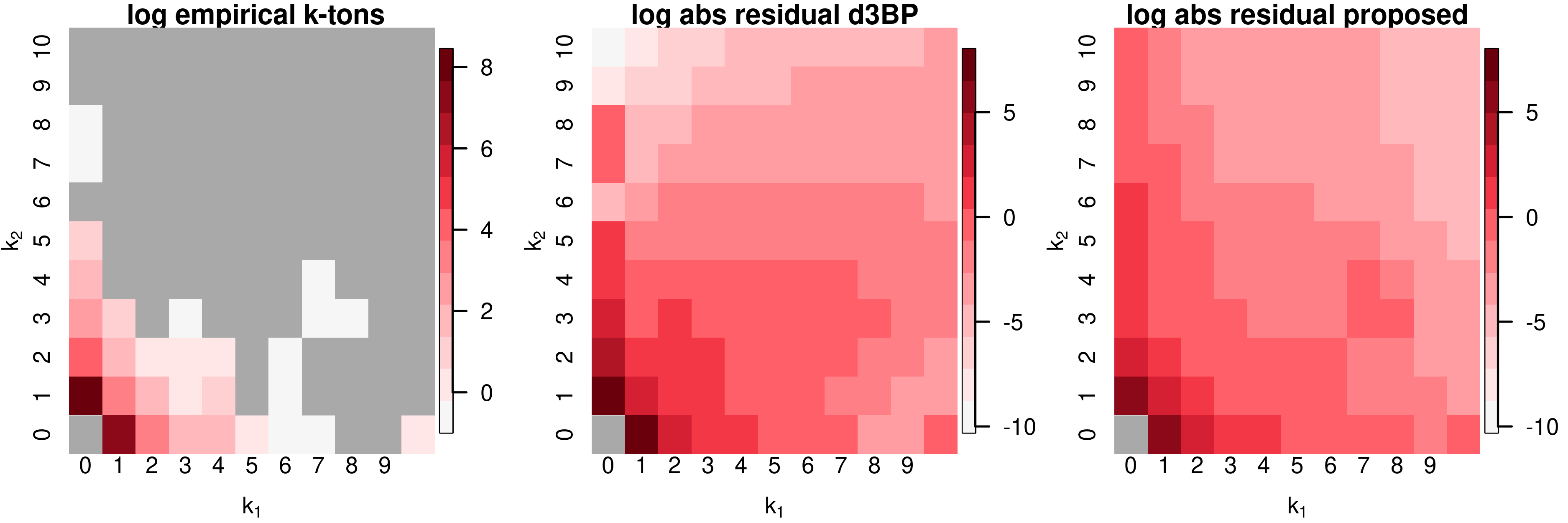}
    \caption{Ground truth of ktons and model prediction using cancer genome dataset. Left most:log over average observed ktons. A square is gray when it the observed value is 0 for all folds or $\vecoccurrence=(0,0)$. Middle and right: Log average absolute residual in d3BP and proposed model. }
    %\label{fig:brca_luad_kton}
    \label{fig:brca_luad_kton}
\end{figure}

\begin{figure}[H]
    \centering
    \includegraphics[width = 0.9\textwidth]{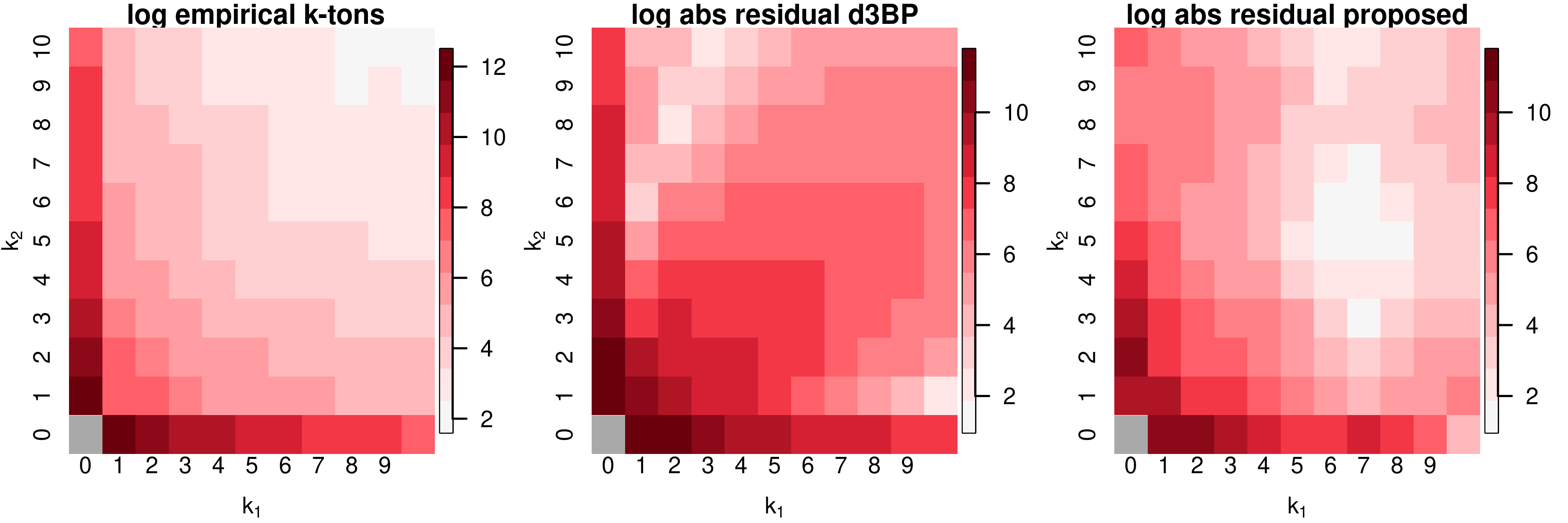}
    \includegraphics[width = 0.9\textwidth]{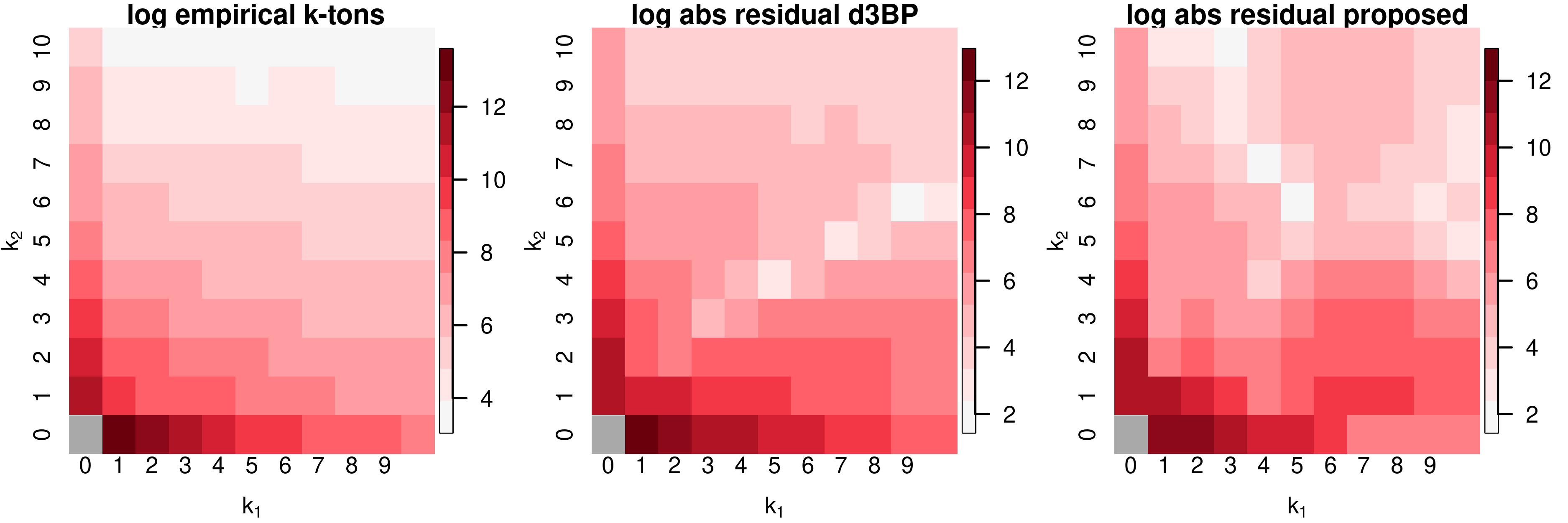}
    \caption{Ground truth of ktons and model prediction using GnomAD dataset. Left most:log over average observed ktons. A square is gray when it the observed value is 0 for all folds or $\vecoccurrence=(0,0)$. Middle and right: Log average absolute residual in d3BP and proposed model. }
    %\label{fig:seu_bgr_kton}
    \label{fig:seu_bgr_kton}
\end{figure}

\subsection{Our method avoids double counting when variants grow slowly}

Here we show our model's behavior when fit to data with low power law rate of d3bp and from our proposed model. Our proposed method could handle both dependency models when the growth is slow while i3bp encounters sever double counting problems.

\begin{figure}[H]
    \centering
    \includegraphics[width = 0.9\textwidth]{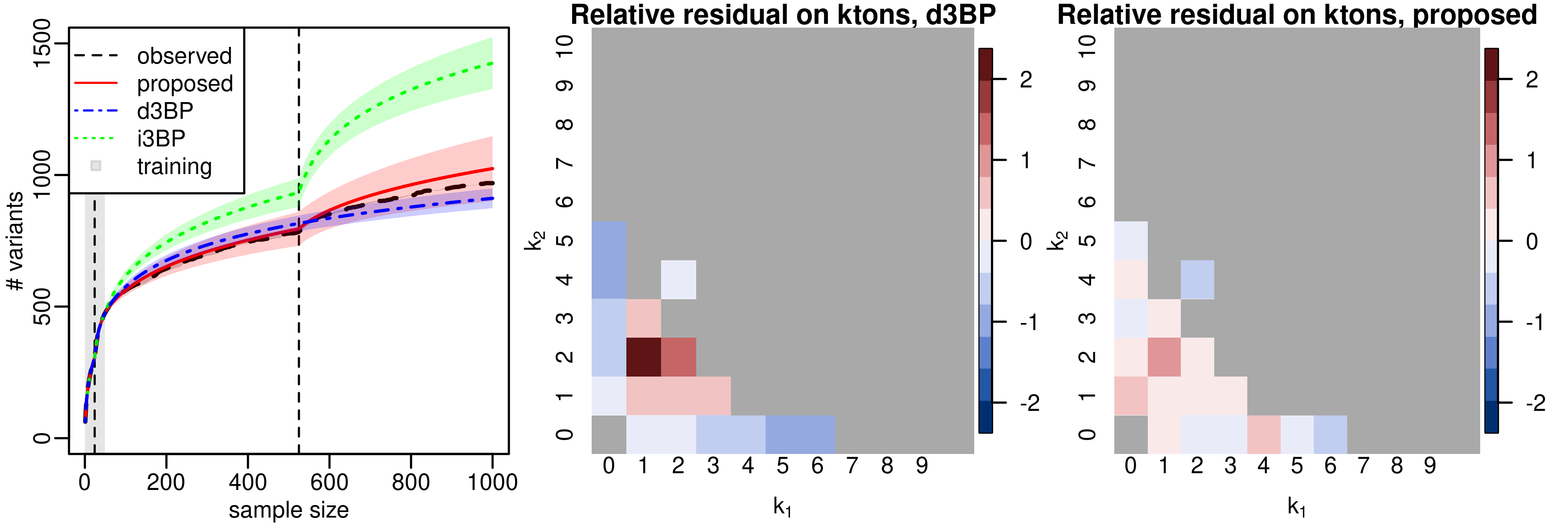}
    \includegraphics[width = 0.9\textwidth]{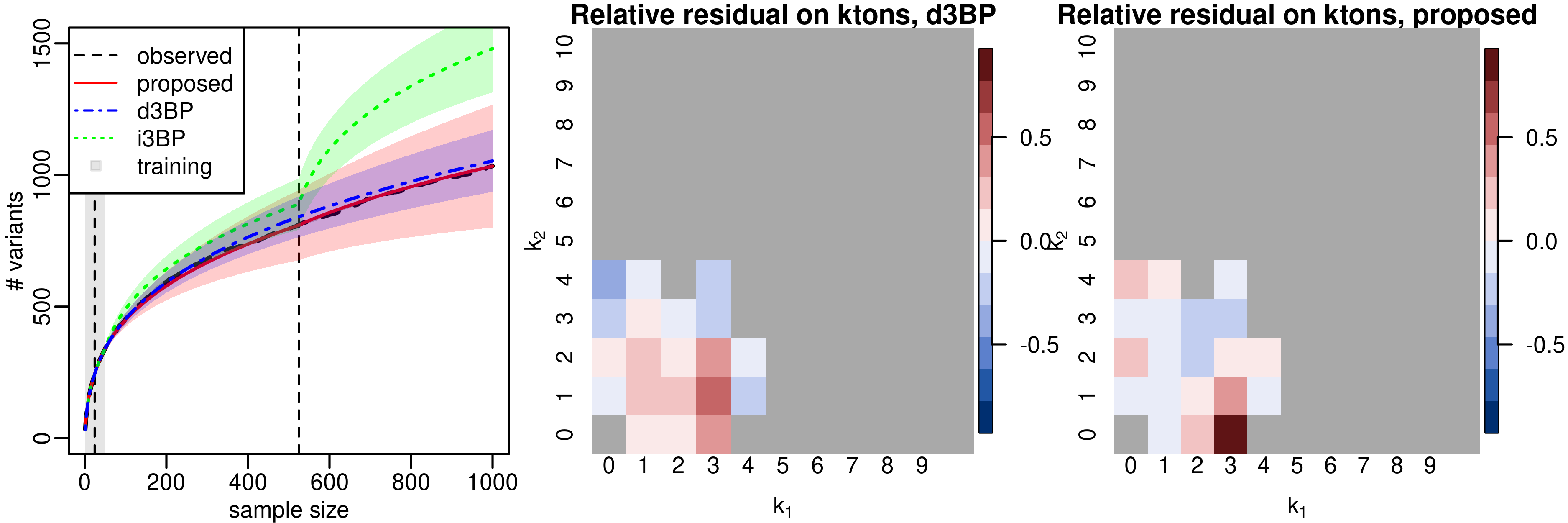}
    \caption{Prediction for slow growth correlated models. Top: proposed model $(\mass,\rate{1}, \rate{2}, \corr{1}, \corr{2}, \discount{1}, \discount{2})=(100,0.4,0.6,0.5,0.5,1,1)$. Bottom: d3bp $(\alpha,c,\sigma)=(40,1,.3)$.}
    \label{fig:slow_growth}
\end{figure}

\subsection{Larger trimming during training}
\label{sec:large_triming}
Determine of the hyperparameters of our model involving maximizing likelihood over trimmed k-tons. We present additional results here with larger trimming in dataset with enough training data to allow $k=15,20$ namely the GnomAD and MSK-IMPACT dataset. The results was not influenced by the trimming (Fig.~\ref{fig:seu_bgr_15}~\ref{fig:seu_bgr_20}~\ref{fig:bidc_la2_15}). 

\begin{figure}[H]
    \centering
    \includegraphics[width = 0.9\textwidth]{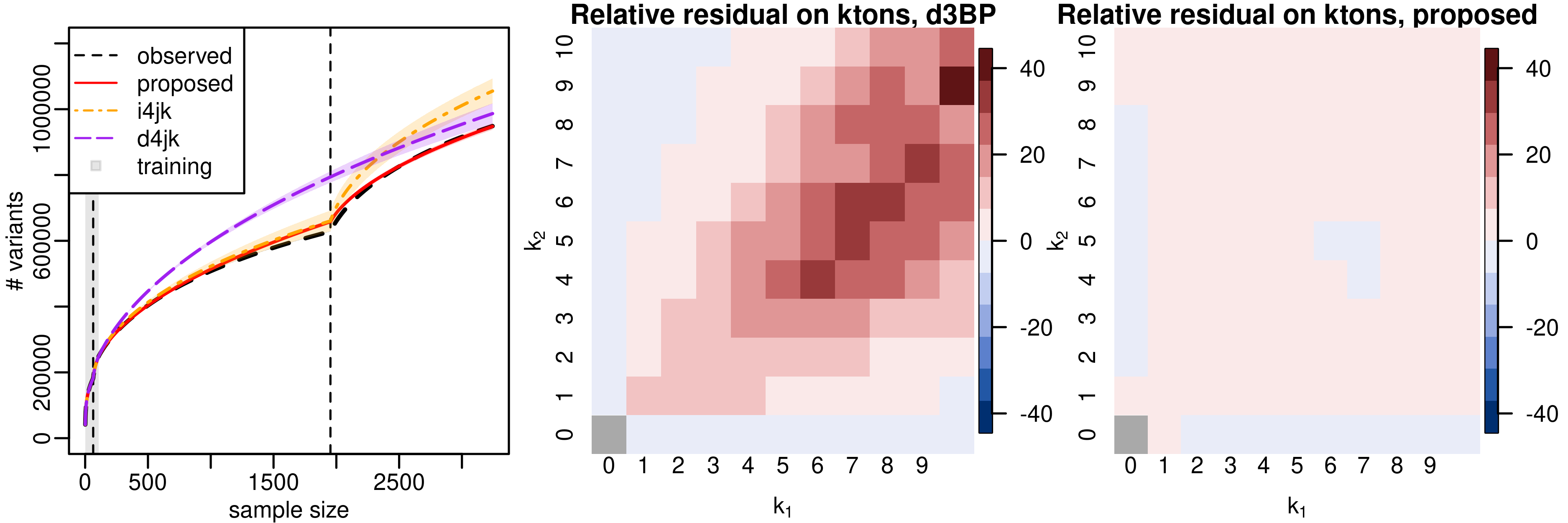}
    \includegraphics[width = 0.9\textwidth]{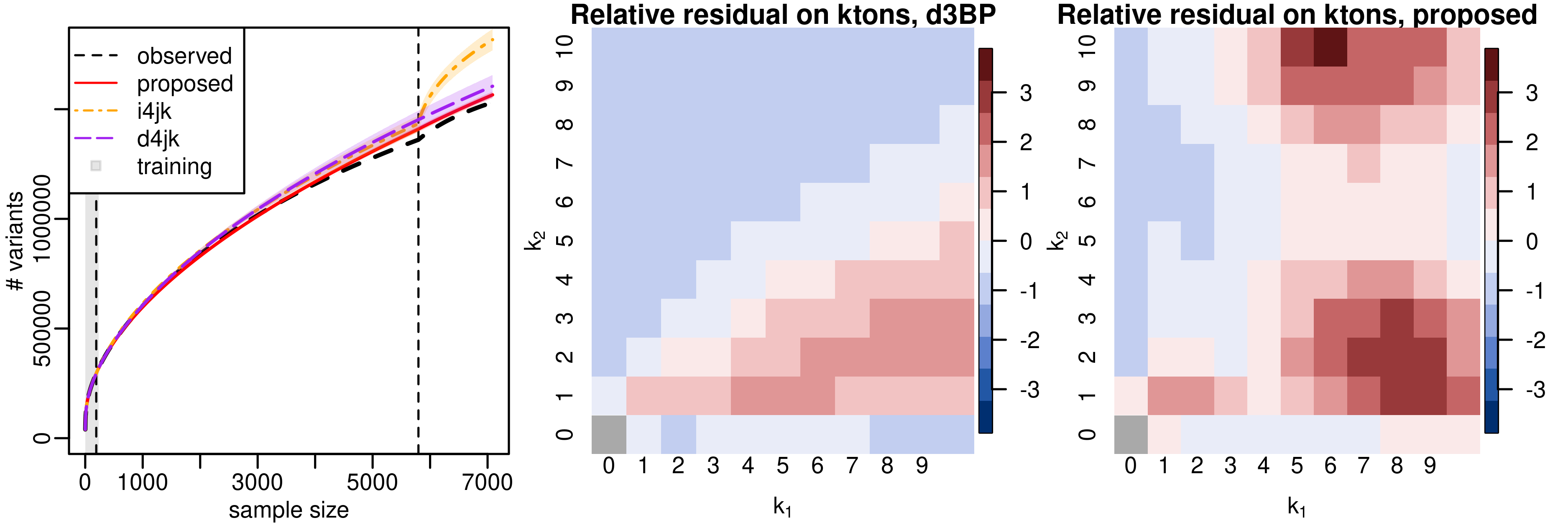}
    \caption{Prediction for the GnomAD dataset, Korean--Bulgarian and Southern European--Bulgarian respectively. The trimming here are set to 15. The results are the same as trimming at 10 in the main text \cref{fig:pop_gen}.}
    %\label{fig:seu_bgr_15}
    \label{fig:seu_bgr_15}
\end{figure}

\begin{figure}[H]
    \centering
    \includegraphics[width = 0.9\textwidth]{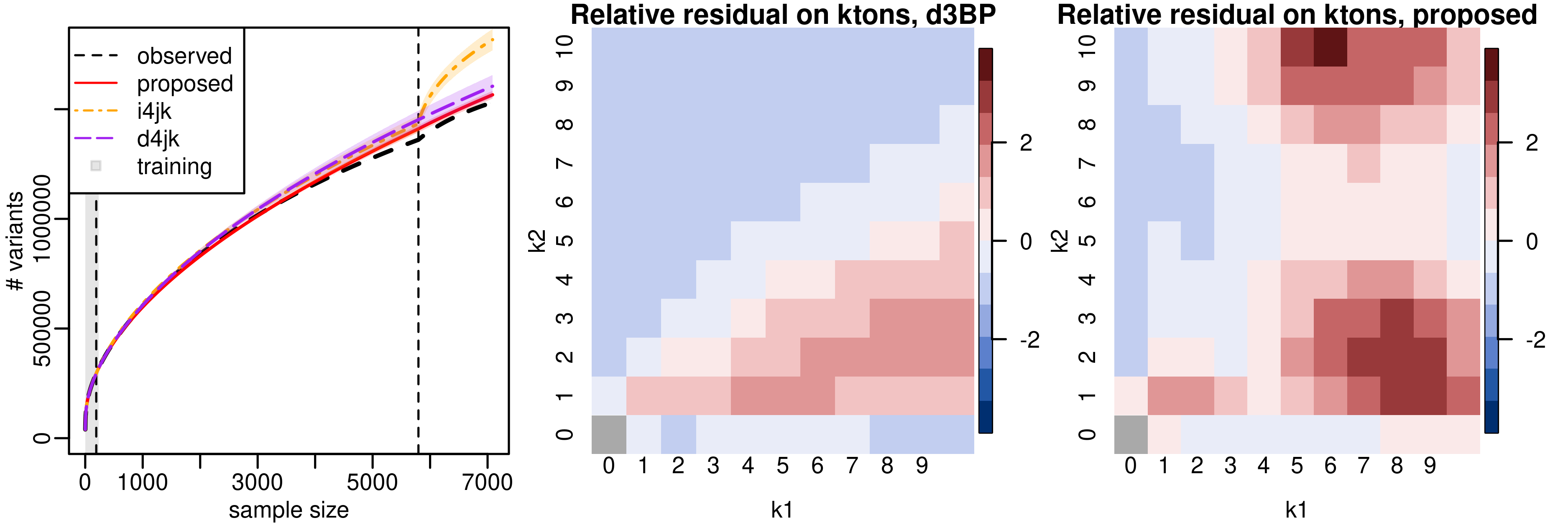}
    \caption{Prediction for the GnomAD dataset Southern European--Bulgarian respectively. The trimming here are set to 20. The results are the same as trimming at 10 in the main text \cref{fig:pop_gen}. The Korean population does not have enough training sample to allow a 20 trimming.}
    \label{fig:seu_bgr_20}
\end{figure}

\begin{figure}[H]
    \centering
    \includegraphics[width = 0.9\textwidth]{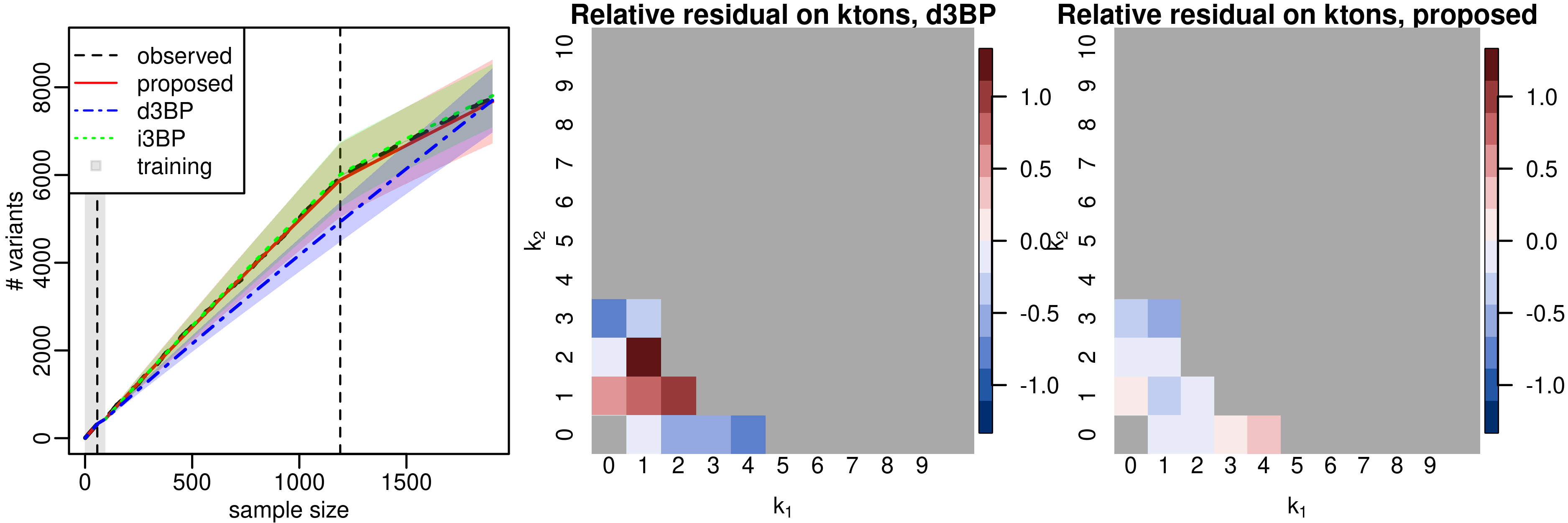}
    \caption{Prediction for the variant counts in, breast and lung cancer. For  MSK-IMPACT dataset. This experiment used trimming of 15 during training and the results are not influenced. }
    \label{fig:bidc_la2_15}
\end{figure}

\section{Power law behavior}
\label{app:power_law}
% !TEX root = ../double_trouble_supp.tex

\subsection{Power law bounds in projection scheme}
\label{sec:projection_scheme_proof}

In this section we provide details and proofs on the asymptotic properties of number of variants seen.
Here the asymptotics are obtained as more samples are being collected, as introduced in~\Cref{sec:projection_powerlaw}. 
Specifically we prove that the expected number of variants seen can be both upper and lower bounded by (deterministic) functions whose asymptotic behavior is that of a power-law, i.e.\ behave like some function $K(N)$ which, as the sample size $N$ increases, satisfies $K(N) \sim CN^\sigma$ for a constant $C$ and $\alpha\in (0,1)$, where we are using the notation $f(t)\sim g(t) \iff \lim_{t\to \infty} f(t)/g(t) = 1$.
 
The main tool we use in our proof is an existing result for power-law asymptotics, developed for the case of single-population models \citep[Proposition 6.1]{Broderick2012}. 
We here first restate \citet[Proposition 6.1]{Broderick2012}. 
This result explains, under the assumption that we have a single population, and the data is drawn from a Bayesian nonparametric model like the one in \Cref{eq:d3BP}, the relationship between the rate of growth of the total number of observed variants and the tail behavior of the underlying Poisson point process rate measure. 
In what follows, we use the following asymptotic relation notation:  $f(t)\sim g(t) \iff \lim_{t\to \infty} f(t)/g(t) = 1$.

A useful tool in our asymptotic analysis is the Laplace transformation of the rate measure's tail, or ``Poissonized'' expectation $\Phi(t)$. In particular, Proposition 6.1, Lemma 6.3 and Lemma 6.4 in~\citet{Broderick2012} (i) established the asymptotic behavior of the (deterministic) Laplace transform $\Phi$, and (ii) linked this asymptotic behavior to the asymptotic number of variants observed from the model. These results will play a crucial role in our analysis, so we here re-state them for completeness.
%that the expectation of number of variants can be well approximated by the Poissionized version and number of variants concentrates around expectations. We restate the

\Cref{lemma:broderick12_6.1} connects tail of the rate measure to ``Poissionized'' expectation 
\begin{myLemma}[Proposition 6.1 of~\citet{Broderick2012}]
    \label{lemma:broderick12_6.1}
    Let $\nu(\de \theta)$ be a (univariate) rate measure such that for $x>0$ there exists $\alpha\in (0,1)$ such that:
    	\begin{equation}
            \label{eq:integral_req}
            \int_0^x \nu(\de\theta)\sim \frac{\alpha}{1-\alpha} x^{1-\alpha}l\left(\frac{1}{x}\right),
        \end{equation}
    		
	for a slowly varying function $l(\cdot)$. Define $\bar{\nu}(x)=\int_x^\infty \nu(\de\theta)$ and $\Phi(t)=t\int_0^\infty e^{-tx}\bar{\nu}(\de x)$. Then it holds
    %\footnote{\textcolor{red}{lom: do we want to leave the $\Gamma(1-\alpha)$? this is just a constant and has no dependence on $t$ so it could be dropped since the asymptotic equivalence would still hold}}
    \begin{align*}
        \Phi(t) \sim \Gamma(1-\alpha)t^\alpha l(t).
    \end{align*}
\end{myLemma}

\Cref{lemma:broderick12_6.3} established that ``Poissionization'' error vanishes as number of samples grows 
 %results as lemmas 
\begin{myLemma}[Lemma 6.3 of~\citet{Broderick2012}]
    \label{lemma:broderick12_6.3}
    Assume variant frequencies are from PPP with rate measure $\nu(\de\theta)$ such that $\int \nu(\de\theta)=\infty$, $\int \theta \nu(\de\theta)<\infty$, and variants are independent Bernoulli random variables with variant frequencies. Define $\bar{\nu}(x)=\int_x^\infty \nu(\de\theta)$ and $\Phi(t)=t\int_0^\infty e^{-tx}\bar{\nu}(\de x)$. Denote the number of variants seen with sample size $N$ as $U_0^N$, then for $N\to\infty$
    \begin{align*}
        \left|\E U_0^N- \Phi(N) \right|\to 0
    \end{align*}
\end{myLemma}

\Cref{lemma:broderick12_6.4} showed that number of variants seen concentrates around the expectation
\begin{myLemma}[Lemma 6.4 of~\citet{Broderick2012}]
    \label{lemma:broderick12_6.4}
    For the same assumption as~\cref{lemma:broderick12_6.3}, we have 
    \begin{align*}
        U_0^N \sim \E U_0^N \text{  a.s.}
    \end{align*}
\end{myLemma}

We are now ready to prove \Cref{prop:power_law_proj}.
Recall that  \Cref{prop:power_law_proj} deals with the ``projection scheme'', in which we are sampling from one population only. 
In this case, the problem can be reduced to the corresponding single population case: we can indeed bound the rate of growth of the number of variants with the corresponding rate of growth arising from the population from which we are sampling.
This can be accomplished by adopting known techniques developed for the the single-population setting.

\begin{proof}[Proof of \Cref{prop:power_law_proj}]

\textbf{Reduction to one population.} We start by considering the ``projection scheme'', in which we sample only from one population. 
In this case, given a two dimensional rate measure $\nu(\de \variantfreqperpop{1} \de \variantfreqperpop{2})$, the variants' frequencies when sampling e.g.\ only from population $1$ are obtained using the mapping theorem, by ``projecting'' $\nu$ onto the first axis indexed by $\variantfreqperpop{1}$ by integrating the measure $\nu$ with respect to $\variantfreqperpop{2}$ (from which the name ``projection'' scheme).

Following the approach of \citet{Broderick2012}, we define the ``projected'' single population rate measures, i.e., 
\begin{align*}
    \levync_{1}(\de \variantfreqperpop{1}) = \int_{s=0}^{s=1}\ratemeasproposed( \de\variantfreqperpop{1}\de s),\quad\text{and}\quad
    \levync_{2}(\de \variantfreqperpop{2}) = \int_{s=0}^{s=1}\ratemeasproposed( \de s \de\variantfreqperpop{2}).
\end{align*}

These measures are responsible for the generation of variant frequencies for each population respectively. 
In principle, we could directly apply~\Cref{lemma:broderick12_6.1} to evaluate the growth of the number of variants.
One practical difficulty is that directly checking~\cref{eq:integral_req} in conditions of~\cref{lemma:broderick12_6.1} is difficult.%\footnote{\textcolor{red}{lom: can we say exactly which integral is hard? there are at least three integrals in there :)}} 
Therefore, due to this difficulty, instead of directly proving that the process leads to a power law behavior for the number of variants under the projection scheme, we will provide a bound for the expected number of variants seen.
We accomplish so by bounding the rate measure and showing that the expected number of variants arising when adopting such bounds grow like a power law in the asymptotic regime.

%\textcolor{red}{lom: do we mean a generic $\levync$ or our proposed one $\nu_{\mathrm{prop}}$?}

To proceed, we first connect the expected number of variants in the projection scheme with the underlying projected rate measure. 
Intuitively, variants observed for the first time when the sample grows arbitrarily large have to be arbitrarily rare (otherwise they would have been seen in pervious samples). 
Therefore, it is intuitive to link the growth in the number of (rare new) variants with the tail of the (projected) rate measure. To formalize this, define:
\begin{align}
\begin{split}
    \bar\levync_{1}(x) := \int_x^1\int_{\variantfreqperpop{2} = 0}^{\variantfreqperpop{2} = 1}\ratemeasproposed( \de\variantfreqperpop{1}\de \variantfreqperpop{2}), \quad \text{and} \quad
    \bar\levync_{2}(x) := \int_x^1\int_{\variantfreqperpop{1}=0}^{\variantfreqperpop{1}=1}\ratemeasproposed( \de\variantfreqperpop{1}\de \variantfreqperpop{2}).
\end{split} \label{eq:tail_function}
\end{align}

Denote with $\poissonize((t_1,0))$ and $\poissonize((0,t_2))$ 
%These quantities coincide with 
the Laplace transforms of the tail function introduced in \Cref{eq:tail_function}:
\begin{equation}
    \label{eq:laplace}
    \begin{aligned}
        \poissonize((t_1,0))&=t_1\int_{0}^1 e^{-t_1x} \bar\levync_{1}(x)\de x, \quad \text{and} \quad
        \poissonize((0,t_2))&=t_2\int_{0}^1 e^{-t_2x} \bar\levync_{2}(x) \de x.
    \end{aligned}
\end{equation}
By applying \cref{lemma:broderick12_6.3}, it follows that $\poissonize((t_1,0))$ and $\poissonize((0,t_2))$ coincide with the asymptotic mean of the number of variants seen when total sample size is $0$ in one population, and $t_\popidx$ in the other population $\popidx$:
%\textcolor{red}{lom: not clear; also what do we mean by asymptotic here? what is diverging?}
%Lemma 6.3 of~\citet{Broderick2012} suggests that the $\poissonize$ evaluated at $N_1$ and $N_2$ will converge to the expected number seen. \textcolor{red}{lom: i am not super clear on this; why is this not identically equal to $0$ for all $N_1$, $N_2$? also, if we cite pointwise a theorem/result, i think it is worth stating it in the appendix?}
\begin{equation}
    \label{eq:poissonized_error}
%    \begin{aligned}
        \lim_{N_1 \to \infty} \left|\poissonize((N_1,0))-\E \news{\zerovec}{(N_1,0)} \right| =  0,\quad\text{and}\quad
         \lim_{N_2 \to \infty} \left|\poissonize((0,N_2))-\E \news{\zerovec}{(0, N_2)} \right| =  0.
%    \end{aligned}
\end{equation}
Then applying~\Cref{lemma:broderick12_6.4}, we have $ \news{\zerovec}{\vecsamplesize} \sim \E \news{\zerovec}{\vecsamplesize}$ a.s..

%\footnote{\textcolor{red}{lom: in \Cref{prop:power_law_proj} we state that $ \news{\zerovec}{\vecsamplesize} \sim \E \news{\zerovec}{\vecsamplesize}$ a.s.; can we have a pointer as per why and how our proof shows that? I don't think we mention it explicitly, and this could be confusing!}} 

%Then we have it follows a Laplace transformation of the tail, i.e.

With~\Cref{eq:poissonized_error}, we can study the expected number of variants $\E \news{\zerovec}{(N_1,0)}$ by studying the asymptotic behavior of the so-called ``Poissonized'' version $\poissonize((t_1,0))$, $\poissonize((0,t_2))$. By using the Laplace transformation form in~\Cref{eq:laplace}, we can then apply~\Cref{lemma:broderick12_6.1} to connect the asymptotic rate of growth of the number of new variants to the tail of the rate measures. 

\paragraph{Power law upper bound, population $1$}
We here derive the asymptotic upper bound $\upperboundproject (\vecsamplesize)$ given in \Cref{eq:power-law_projection-upperbound}.
In particular, $\upperboundproject (\vecsamplesize)$ is an upper bound for $\poissonize((t_1,0))$ as $t_1 \to \infty$. 
Since $\poissonize((t_1,0))$ is what determines the rate of growth of rare variants in the projection scheme (when growing only population $1$, as showed in \citet[Lemma 6.3]{Broderick2012}), by bounding $\poissonize((t_1,0))$ we effectively bound the expectation $\E \news{\zerovec}{(N_1,0)}$, which is our stated goal (see \Cref{eq:projection_bounds}).

The idea underlying our proof is similar to the one used to prove that the proposed prior satisfies the integral requirements in \Cref{req:finite_mean,req:inf_mass}. 
Namely, we upper bound the rate measure using the two following observations:
\begin{align}
	\variantfreqperpop{1}+\variantfreqperpop{2}^{\rate{2}/\rate{1}}\ge \max\left\{ \variantfreqperpop{1}, \variantfreqperpop{2}^{\rate{2}/\rate{1}} \right\}
	\quad\text{and}\quad
	\variantfreqperpop{1}+\variantfreqperpop{2}\ge \max\left\{ \variantfreqperpop{1}, \variantfreqperpop{2} \right\}.
\label{eq:bounds_ineq}
\end{align}
In light of these two simple observations, for a positive $\delta>0$ it holds
\begin{align*}
    \frac{(\variantfreqperpop{1}+\variantfreqperpop{2}^{\rate{2}/\rate{1}})^{-\rate{1}}}{(\variantfreqperpop{1}+\variantfreqperpop{2})^{\corr{1}+\corr{2}}}&=(\variantfreqperpop{1}+\variantfreqperpop{2}^{\rate{2}/\rate{1}})^{-\rate{1}}(\variantfreqperpop{1}+\variantfreqperpop{2})^{-\corr{1}-\smallpositive}(\variantfreqperpop{1}+\variantfreqperpop{2})^{-\corr{2}+\smallpositive}\\
    &\le \variantfreqperpop{1}^{-\rate{1}}\variantfreqperpop{1}^{-\corr{1}-\smallpositive}\variantfreqperpop{2}^{-\corr{2}+\smallpositive}.
\end{align*}
Thus, letting
\[
	\mu_{\mathrm{upper}, 1}(\de\vecvariantfreq{}):= 
	\frac{\mass}{B(\corr{1}, \conc{1})B(\corr{2}, \conc{2})}
	\variantfreqperpop{1}^{-\rate{1}-\corr{1}-\smallpositive} \variantfreqperpop{1}^{\corr{1}-1}\variantfreqperpop{2}^{-\corr{2}+\smallpositive}\variantfreqperpop{2}^{\corr{2}-1}(1-\variantfreqperpop{1})^{\conc{1}-1}(1-\variantfreqperpop{2})^{\conc{2}-1} \de \btheta,
\]
it holds
\begin{align*}
    \ratemeasproposed(\de\vecvariantfreq{})\le\mu_{\mathrm{upper}, 1}(\de\vecvariantfreq{}).
\end{align*}

Further letting 
\[
	\bar\mu_{\mathrm{upper}, 1}(x):=\int_{x}^1\int_{0}^1 \mu_{\mathrm{upper}, 1}(\de\vecvariantfreq{}),
\]
it holds that 
\[
	\bar\nu_1(\de\vecvariantfreq{})\le \bar\mu_{\mathrm{upper}, 1}.
\]
Thus we have
\begin{align*}
    \poissonize((t_1,0))\le \upperbound((t_1,0)):= t_1\int_{0}^1 e^{-t_1x} \bar\mu_{\mathrm{upper}, 1}(x) \de x.
\end{align*}
%\lom{We should make clear how the equation directly above this comment and the next align block relate to each other (asymptotic relationship of Tauberian/Abelian stuff}

The asymptotic behavior of right hand side as $t_1\to\infty$ can be connected to the tail of $\mu_{\mathrm{upper}, 1}$ using~\cref{lemma:broderick12_6.1}. We check the conditions of~\cref{lemma:broderick12_6.1}  for $\mu_{\mathrm{upper}, 1}(\de\vecvariantfreq{})$. 

%By a repeat use of Tauberian and Abelian theorem same as \citet{Broderick2012}\footnote{\lom{let's either cite the exact paper, or poof of broderick12}}, the asymptotic of $\upperbound((t_1,0))$ is determined by 
\begin{align*}
    &\int_{0}^{x} \variantfreqperpop{1} \int_{0}^1 \mu_{\mathrm{upper}, 1}(\de\variantfreqperpop{1}\de \variantfreqperpop{2})  \\
    =&\frac{\mass}{B(\corr{1}, \conc{1})B(\corr{2}, \conc{2})}\int_{0}^x \variantfreqperpop{1} \variantfreqperpop{1}^{-\rate{1}-\corr{1}-\smallpositive} \variantfreqperpop{1}^{\corr{1}-1} (1-\variantfreqperpop{1})^{\conc{1}-1}\de\variantfreqperpop{1} \int_0^1 \variantfreqperpop{2}^{-\corr{2}+\smallpositive}\variantfreqperpop{2}^{\corr{2}-1}(1-\variantfreqperpop{2})^{\conc{2}-1} \de\variantfreqperpop{2}\\
    =&\frac{\mass B(\smallpositive, \conc{2})}{B(\corr{1}, \conc{1})B(\corr{2}, \conc{2})}\int_{0}^x \variantfreqperpop{1}^{-\rate{1}-\smallpositive} (1-\variantfreqperpop{1})^{\conc{1}-1}\de\variantfreqperpop{1} \\
    \sim& \frac{\mass B(\smallpositive, \conc{2})}{B(\corr{1}, \conc{1})B(\corr{2}, \conc{2})} \frac{1}{(1-\rate{1}-\smallpositive)} x^{1-\rate{1}-\smallpositive}
\end{align*}

%Consistent across the paper, by $a_{t}\sim b_{t}$ we mean that $\lim_{t\to\infty} a_t/b_t=1$.

%\lom{let's have a pointer for what we mean by $a\sim b$ --- ok to just point to main text if we have it there!}

Thus we have by~\cref{lemma:broderick12_6.1} that for all $\smallpositive\le \min\{\corr{1}, \corr{2}\}/2$, there exists a constant $\constantupper{1}>0$, independent of $t_1$, for which
\begin{align*}
    \poissonize((t_1,0))\le \upperbound((t_1,0))\sim \constantupper{1} t_1^{\rate{1}+\smallpositive} 
\end{align*}
%\textcolor{red}{lom: ``thus we have'' -- is this by which lemma? let's give a precise pointer - this should be a fully fleshed out proof, let's not have the reader re-derive steps}.
We have by~\cref{eq:poissonized_error} the approximation error of $\poissonize((\samplesize{1},0))$ to $\E\left[\news{\zerovec}{(\samplesize{1},0)} \right]$ goes to 0, and in turn
\begin{align*}
    \E\left[\news{\zerovec}{(\samplesize{1},0)} \right] \le \upperbound((\samplesize{1},0))\sim \constantupper{1} \samplesize{1}^{\rate{1}+\smallpositive} 
\end{align*}
%\textcolor{red}{lom: ``We have'' -- as above, what are we leveraging here?}.

\paragraph{Power law upper bound, population $2$}
Now we derive the upper bound for population $2$, again under the projection scheme. To do so, we construct a bound for $\poissonize((0, t_2))$ using a similar argument to the one used for $\poissonize((t_1, 0))$.

For $\smallpositive>0$, in light of the bounds of \Cref{eq:bounds_ineq},

\begin{align*}
    \frac{(\variantfreqperpop{1}+\variantfreqperpop{2}^{\rate{2}/\rate{1}})^{-\rate{1}}}{(\variantfreqperpop{1}+\variantfreqperpop{2})^{\corr{1}+\corr{2}}}&=(\variantfreqperpop{1}+\variantfreqperpop{2}^{\rate{2}/\rate{1}})^{-\rate{1}}(\variantfreqperpop{1}+\variantfreqperpop{2})^{-\corr{1}-\smallpositive}(\variantfreqperpop{1}+\variantfreqperpop{2})^{-\corr{2}+\smallpositive}\\
    &\le \variantfreqperpop{2}^{-\rate{2}}\variantfreqperpop{2}^{-\corr{2}-\smallpositive}\variantfreqperpop{1}^{-\corr{1}+\smallpositive},
\end{align*}
and in turn letting 
\[
	\mu_{\mathrm{upper}, 2}(\de\vecvariantfreq{}):= \frac{\mass}{B(\corr{1}, \conc{1})B(\corr{2}, \conc{2})}\variantfreqperpop{1}^{-\corr{1}+\smallpositive} \variantfreqperpop{1}^{\corr{1}-1}\variantfreqperpop{2}^{-\rate{2}-\corr{2}-\smallpositive}\variantfreqperpop{2}^{\corr{2}-1}(1-\variantfreqperpop{1})^{\conc{1}-1}(1-\variantfreqperpop{2})^{\conc{2}-1} \de \btheta,
\]
it holds that
\begin{align*}
    \ratemeasproposed(\de\vecvariantfreq{})\le \mu_{\mathrm{upper}, 2}(\de\vecvariantfreq{}).
\end{align*}

Denote with $\bar\mu_{\mathrm{upper}, 2}(x)=\int_{x}^1\int_{0}^1 \mu_{\mathrm{upper}, 2}(\de\vecvariantfreq{}) $, we have $\bar\nu_2(\de\vecvariantfreq{})\le \bar\mu_{\mathrm{upper}, 2} $. %\textcolor{red}{lom: check also this inequality}
Thus we have
\begin{align*}
    \poissonize((0, t_2))\le \upperbound((0, t_2)):= t_2\int_{0}^1 e^{-t_2x} \bar\mu_{\mathrm{upper}, 2}(x)\de x.
\end{align*}
We again check that the conditions of~\cref{lemma:broderick12_6.1} hold for $\mu_{\mathrm{upper}, 2}(\de\vecvariantfreq{})$. 
\begin{align*}
    &\int_{0}^{x} \variantfreqperpop{2} \int_{0}^1 \mu_{\mathrm{upper}, 2}(\de\variantfreqperpop{1}\de \variantfreqperpop{2})  \\
    =&\frac{\mass}{B(\corr{1}, \conc{1})B(\corr{2}, \conc{2})}\int_{0}^x \variantfreqperpop{2} \variantfreqperpop{2}^{-\rate{2}-\corr{2}-\smallpositive} \variantfreqperpop{2}^{\corr{2}-1} (1-\variantfreqperpop{2})^{\conc{2}-1}\de\variantfreqperpop{2} \int_0^1 \variantfreqperpop{1}^{-\corr{1}+\smallpositive}\variantfreqperpop{1}^{\corr{1}-1}(1-\variantfreqperpop{1})^{\conc{1}-1} \de\variantfreqperpop{1}\\
    =&\frac{\mass B(\smallpositive, \conc{2})}{B(\corr{1}, \conc{1})B(\corr{2}, \conc{2})}\int_{0}^x \variantfreqperpop{2}^{-\rate{2}-\smallpositive} (1-\variantfreqperpop{2})^{\conc{2}-1}\de\variantfreqperpop{2} \\
    \sim& \frac{\mass B(\smallpositive, \conc{2})}{B(\corr{1}, \conc{1})B(\corr{2}, \conc{2})} \frac{1}{(1-\rate{2}-\smallpositive)} x^{1-\rate{2}-\smallpositive}
\end{align*}

Thus for all $\smallpositive$, there exists a constant $\constantupper{2}$ which does not depend on $t_2$ such that:
\begin{align*}
    \poissonize((0, t_2))\le \upperbound((0,t_2))\sim \constantupper{2} t_2^{\rate{2}+\smallpositive}.
\end{align*}
In turn,
\begin{align*}
    \E\left[\news{\zerovec}{(0,\samplesize{2})} \right] \le \lowerbound((0,\samplesize{2}))\sim \constantupper{2} \samplesize{2}^{\rate{2}+\smallpositive}.
\end{align*}

\textbf{Power law lower bound}

To derive the asymptotic lower bound for the expected  number of new variants we adopt a similar idea to the one used to derive the upper bound.
Specifically, we find a lower bound for the model's rate measure, check that the lower bound satisfies the conditions in~\cref{lemma:broderick12_6.1}, and use that rate as a lower bound. 

To find this lower bound, we leverage the observation that 
\begin{equation}
	\variantfreqperpop{1}+\variantfreqperpop{2}\le 2\max\{ \variantfreqperpop{1},\variantfreqperpop{2} \}
	\quad\text{and}\quad
	\variantfreqperpop{1}+\variantfreqperpop{2}^{\rate{2}/\rate{1}}\le 2\max\{ \variantfreqperpop{1},\variantfreqperpop{2}^{\rate{2}/\rate{1}} \}.
\end{equation}

%In the later proof we will use $2\variantfreqperpop{1}$ or $2\variantfreqperpop{2}^{\rate{2}/\rate{1}}$ explicitly as the upper bound of sum depends on which one is larger. 

We start with $\rate{2}/\rate{1}\ge 1$. For $\poissonize((t_1,0))$, we use the fact that $\variantfreqperpop{1}+\variantfreqperpop{2}\le 2\max\{\variantfreqperpop{1}, \variantfreqperpop{2}\}$. We first consider subset of the unit square that $1/2 \ge \variantfreqperpop{1}\ge \variantfreqperpop{2}$, where we also have $\variantfreqperpop{1}\ge \variantfreqperpop{2}^{\rate{2}/\rate{1}}$ and $(1-\variantfreqperpop{2})\le \min\{1, 2^{1-\conc{2}}\}$. We have in this region, $\variantfreqperpop{1}=\max\{\variantfreqperpop{1},\variantfreqperpop{2}\}=\max\{\variantfreqperpop{1}, \variantfreqperpop{2}^{\rate{2}/\rate{1}}\}$. To construct a lower bound of the rate measure, we could use a rate measure whose density is 0 outside this range and within the range we used the bounds found, concretely we have:
\begin{align*}
        \frac{(\variantfreqperpop{1}+\variantfreqperpop{2}^{\rate{2}/\rate{1}})^{-\rate{1}}}{(\variantfreqperpop{1}+\variantfreqperpop{2})^{\corr{1}+\corr{2}}}&=(\variantfreqperpop{1}+\variantfreqperpop{2}^{\rate{2}/\rate{1}})^{-\rate{1}}(\variantfreqperpop{1}+\variantfreqperpop{2})^{-\corr{1}-\corr{2}}\\
        &\ge \left\{(2\variantfreqperpop{1})^{-\rate{1}} \right\} \left\{ (2\variantfreqperpop{1})^{-\corr{1}-\corr{2}} \right\}\bm{1}_{1/2 \ge \variantfreqperpop{1}\ge \variantfreqperpop{2}}
\end{align*}
To slightly simplify the integral we have when $\variantfreqperpop{2}\le 1/2$ $,(1-\variantfreqperpop{2})^{\conc{2}-1}\ge \min\{1, 2^{1-\conc{2}}\}$, combine the bounds, we can bound our proposed rate measure by 
\begin{align*}
    \ratemeasproposed(\de\vecvariantfreq{})\ge\mu_{s,1}(\de\vecvariantfreq{}):= \frac{\mass 2^{-\rate{1}-\corr{1}-\corr{2}}\min\{1, 2^{1-\conc{2}}\}}{B(\corr{1}, \conc{1})B(\corr{2}, \conc{2})}\variantfreqperpop{1}^{-\rate{1}-\corr{1}-\corr{2}} \variantfreqperpop{1}^{\corr{1}-1}\variantfreqperpop{2}^{\corr{2}-1}(1-\variantfreqperpop{1})^{\conc{1}-1} \bm{1}_{1/2 \ge \variantfreqperpop{1}\ge \variantfreqperpop{2}}
\end{align*}
Use the same arguments as before, it remains to check $\mu_{s,1}(\de\vecvariantfreq{})$ satisfy requirements of~\cref{lemma:integral_requirements}. Since we take $x\to 0$ we can focus on $x<1/2$. 
\begin{align*}
    &\int_0^x \variantfreqperpop{1} \int_0^1 \mu_{s,1}(\de\vecvariantfreq{})\\
    =& \frac{\mass 2^{-\rate{1}-\corr{1}-\corr{2}}\min\{1, 2^{1-\conc{2}}\}}{B(\corr{1}, \conc{1})B(\corr{2}, \conc{2})}\int_0^x \variantfreqperpop{1} \int_0^{\variantfreqperpop{1}} \variantfreqperpop{1}^{-\rate{1}-\corr{1}-\corr{2}} \variantfreqperpop{1}^{\corr{1}-1}\variantfreqperpop{2}^{\corr{2}-1}(1-\variantfreqperpop{1})^{\conc{1}-1}\de\vecvariantfreq{}\\
    =& \frac{\mass 2^{-\rate{1}-\corr{1}-\corr{2}}\min\{1, 2^{1-\conc{2}}\}}{B(\corr{1}, \conc{1})B(\corr{2}, \conc{2})}\int_0^x \variantfreqperpop{1}^{-\rate{1}-\corr{1}-\corr{2}} \variantfreqperpop{1}^{\corr{1}} (1-\variantfreqperpop{1})^{\conc{1}-1} \int_0^{\variantfreqperpop{1}} \variantfreqperpop{2}^{\corr{2}-1}\de\vecvariantfreq{}\\
    =&\frac{\mass 2^{-\rate{1}-\corr{1}-\corr{2}}\min\{1, 2^{1-\conc{2}}\}}{B(\corr{1}, \conc{1})B(\corr{2}, \conc{2})\corr{2}}\int_0^x \variantfreqperpop{1}^{-\rate{1}-\corr{1}-\corr{2}} \variantfreqperpop{1}^{\corr{1}} (1-\variantfreqperpop{1})^{\conc{1}-1} \variantfreqperpop{1}^{\corr{2}}\de\variantfreqperpop{1}\\
    \sim & \frac{\mass 2^{-\rate{1}-\corr{1}-\corr{2}}\min\{1, 2^{1-\conc{2}}\}}{B(\corr{1}, \conc{1})B(\corr{2}, \conc{2})\corr{2}}\int_0^x \variantfreqperpop{1}^{-\rate{1}} \de\variantfreqperpop{1}\\
    =& \frac{\mass 2^{-\rate{1}-\corr{1}-\corr{2}}\min\{1, 2^{1-\conc{2}}\}}{B(\corr{1}, \conc{1})B(\corr{2}, \conc{2})\corr{2}} \frac{1}{1-\rate{1}}x^{1-\rate{1}}
\end{align*}
we have then by~\cref{lemma:broderick12_6.1} and~\cref{eq:poissonized_error}, similar to upper bounds
\begin{align*}
   \E\left[ \news{\zerovec}{(\samplesize{1}, 0)} \right] \ge \lowerbound((\samplesize{1}, 0))\sim \constantlower{1} \samplesize{1}^{\rate{1}} 
\end{align*}

For $\poissonize((0, t_2))$, first consider the region that $1/2 \ge \variantfreqperpop{2} \ge \variantfreqperpop{2}^{\rate{2}/\rate{1}}\ge \variantfreqperpop{1}  $, where we also have $(1-\variantfreqperpop{1})\le \min\{1, 2^{1-\conc{1}}\}$. In this region, $\variantfreqperpop{2}=\max\{\variantfreqperpop{1},\variantfreqperpop{2}\}$, $\variantfreqperpop{2}^{\rate{2}/\rate{1}}=\max\{\variantfreqperpop{1}, \variantfreqperpop{2}^{\rate{2}/\rate{1}}\}$. Same as before, we lower bound our proposed rate measure with a rate whose density is 0 outside this region, concretely we have
\begin{align*}
        \frac{(\variantfreqperpop{1}+\variantfreqperpop{2}^{\rate{2}/\rate{1}})^{-\rate{1}}}{(\variantfreqperpop{1}+\variantfreqperpop{2})^{\corr{1}+\corr{2}}}&=(\variantfreqperpop{1}+\variantfreqperpop{2}^{\rate{2}/\rate{1}})^{-\rate{1}}(\variantfreqperpop{1}+\variantfreqperpop{2})^{-\corr{1}-\corr{2}}\\
        &\ge \left\{ (2\variantfreqperpop{2}^{\rate{2}/\rate{1}} )^{-\rate{1}} \right\} \left\{(2\variantfreqperpop{2})^{-\corr{1} - \corr{2}} \right\}\bm{1}_{1/2 \ge \variantfreqperpop{2} \ge \variantfreqperpop{2}^{\rate{2}/\rate{1}}\ge \variantfreqperpop{1}} \\
        &= 2^{-\rate{1}-\corr{1}-\corr{2}} \variantfreqperpop{2}^{-\rate{2}-\corr{1}-\corr{2}}\bm{1}_{1/2 \ge \variantfreqperpop{2} \ge \variantfreqperpop{2}^{\rate{2}/\rate{1}}\ge \variantfreqperpop{1}}.
\end{align*}
Similarly we have when $\variantfreqperpop{1}\le 1/2$ $,(1-\variantfreqperpop{1})^{\conc{1}-1}\ge \min\{1, 2^{1-\conc{1}}\}$, combine the bounds, we have 
\begin{align*}
    \ratemeasproposed(\de\vecvariantfreq{})\ge\mu_{s,2}(\de\vecvariantfreq{}):= \frac{\mass 2^{-\rate{1}-\corr{1}-\corr{2}}\min\{1, 2^{1-\conc{1}}\}}{B(\corr{1}, \conc{1})B(\corr{2}, \conc{2})}\variantfreqperpop{2}^{-\rate{2}-\corr{1}-\corr{2}} \variantfreqperpop{2}^{\corr{2}-1}\variantfreqperpop{1}^{\corr{1}-1}(1-\variantfreqperpop{2})^{\conc{2}-1} \bm{1}_{1/2 \ge \variantfreqperpop{2} \ge \variantfreqperpop{2}^{\rate{2}/\rate{1}}\ge \variantfreqperpop{1}}
\end{align*}
Use the same arguments as before, it remains to check conditions of~\cref{lemma:broderick12_6.1} for $\mu_{s,2}(\de\vecvariantfreq{})$. We can focus on $x\le 1/2$ since we take $x\to 0$. 
\begin{align*}
    &\int_0^x \variantfreqperpop{2} \int_0^1 \mu_{s,2}(\de\vecvariantfreq{})\\
    =& \frac{\mass 2^{-\rate{1}-\corr{1}-\corr{2}}\min\{1, 2^{1-\conc{1}}\}}{B(\corr{1}, \conc{1})B(\corr{2}, \conc{2})}\int_0^x \variantfreqperpop{2} \int_0^{\variantfreqperpop{2}^{\rate{2}/\rate{1}}} \variantfreqperpop{2}^{-\rate{2}-\corr{1}-\corr{2}} \variantfreqperpop{2}^{\corr{2}-1}\variantfreqperpop{1}^{\corr{1}-1}(1-\variantfreqperpop{2})^{\conc{2}-1}\de\vecvariantfreq{}\\
    =& \frac{\mass 2^{-\rate{1}-\corr{1}-\corr{2}}\min\{1, 2^{1-\conc{1}}\}}{B(\corr{1}, \conc{1})B(\corr{2}, \conc{2})}\int_0^x \variantfreqperpop{2}^{-\rate{2}-\corr{1}-\corr{2}} \variantfreqperpop{2}^{\corr{2}} (1-\variantfreqperpop{2})^{\conc{2}-1} \int_0^{\variantfreqperpop{2}^{\rate{2}/\rate{1}}} \variantfreqperpop{1}^{\corr{1}-1}\de\vecvariantfreq{}\\
    =&\frac{\mass 2^{-\rate{1}-\corr{1}-\corr{2}}\min\{1, 2^{1-\conc{1}}\}}{B(\corr{1}, \conc{1})B(\corr{2}, \conc{2})\corr{1}}\int_0^x \variantfreqperpop{2}^{-\rate{2}-\corr{1}-\corr{2}} \variantfreqperpop{2}^{\corr{2}} (1-\variantfreqperpop{2})^{\conc{2}-1} \variantfreqperpop{2}^{\corr{1}\frac{\rate{2}}{\rate{1}}}\de\variantfreqperpop{2}\\
    \sim & \frac{\mass 2^{-\rate{1}-\corr{1}-\corr{2}}\min\{1, 2^{1-\conc{1}}\}}{B(\corr{1}, \conc{1})B(\corr{2}, \conc{2})\corr{1}}\int_0^x \variantfreqperpop{2}^{-(\rate{2}-\corr{1}(\frac{\rate{2}}{\rate{1}}-1))} \de\variantfreqperpop{2}\\
    =& \frac{\mass 2^{-\rate{1}-\corr{1}-\corr{2}}\min\{1, 2^{1-\conc{1}}\}}{B(\corr{1}, \conc{1})B(\corr{2}, \conc{2})\corr{1}} \frac{1}{1-(\rate{2}-\corr{1}(\frac{\rate{2}}{\rate{1}}-1))}x^{1-(\rate{2}-\corr{1}(\frac{\rate{2}}{\rate{1}}-1))}
\end{align*}
we have by~\cref{lemma:broderick12_6.1} and~\cref{eq:poissonized_error}
\begin{align*}
    \E \news{\zerovec}{(0, \samplesize{2})}\ge \lowerbound((0, \samplesize{2}))\sim \constantlowerwithrateprime{1} \samplesize{2}^{\rate{2}-\corr{1}(\frac{\rate{2}}{\rate{1}}-1)} 
\end{align*}

Alternatively we observe that $(\variantfreqperpop{1}+\variantfreqperpop{2}^{\rate{2}/\rate{1}})^{-\rate{1}}\ge (\variantfreqperpop{1}+\variantfreqperpop{2})^{-\rate{1}}$ in the range of $\variantfreqperpop{2}\ge \variantfreqperpop{1}$ and 
\begin{align*}
    \frac{(\variantfreqperpop{1}+\variantfreqperpop{2}^{\rate{2}/\rate{1}})^{-\rate{1}}}{(\variantfreqperpop{1}+\variantfreqperpop{2})^{\corr{1}+\corr{2}}}&=(\variantfreqperpop{1}+\variantfreqperpop{2}^{\rate{2}/\rate{1}})^{-\rate{1}}(\variantfreqperpop{1}+\variantfreqperpop{2})^{-\corr{1}-\corr{2}}\\
    &\ge (\variantfreqperpop{1}+\variantfreqperpop{2})^{-\corr{1}-\corr{2}-\rate{1}}\bm{1}_{1/2 \ge \variantfreqperpop{2} \ge \variantfreqperpop{1}}\\
    &\ge \variantfreqperpop{2}^{-\corr{1}-\corr{2}-\rate{1}}\bm{1}_{1/2 \ge \variantfreqperpop{2} \ge \variantfreqperpop{1}}
\end{align*}
Similarly we use that when $\variantfreqperpop{1}\le 1/2$ $,(1-\variantfreqperpop{1})^{\conc{1}-1}\ge \min\{1, 2^{1-\conc{1}}\}$, combine the bounds, we have 
\begin{align*}
    \ratemeasproposed(\de\vecvariantfreq{})\ge\mu'_{s,2}(\de\vecvariantfreq{}):= \frac{\mass 2^{-\rate{1}-\corr{1}-\corr{2}}\min\{1, 2^{1-\conc{1}}\}}{B(\corr{1}, \conc{1})B(\corr{2}, \conc{2})}\variantfreqperpop{2}^{-\rate{1}-\corr{1}-\corr{2}} \variantfreqperpop{2}^{\corr{2}-1}\variantfreqperpop{1}^{\corr{1}-1}(1-\variantfreqperpop{2})^{\conc{2}-1} \bm{1}_{1/2 \ge \variantfreqperpop{2} \ge \variantfreqperpop{1}}
\end{align*}
We check the conditions in~\cref{lemma:broderick12_6.1}
\begin{align*}
    &\int_0^x \variantfreqperpop{2} \int_0^1 \mu'_{s,2}(\de\vecvariantfreq{})\\
    =& \frac{\mass 2^{-\rate{1}-\corr{1}-\corr{2}}\min\{1, 2^{1-\conc{1}}\}}{B(\corr{1}, \conc{1})B(\corr{2}, \conc{2})}\int_0^x \variantfreqperpop{2} \int_0^{\variantfreqperpop{2}} \variantfreqperpop{2}^{-\rate{1}-\corr{1}-\corr{2}} \variantfreqperpop{2}^{\corr{2}-1}\variantfreqperpop{1}^{\corr{1}-1}(1-\variantfreqperpop{2})^{\conc{2}-1}\de\vecvariantfreq{}\\
    =& \frac{\mass 2^{-\rate{1}-\corr{1}-\corr{2}}\min\{1, 2^{1-\conc{1}}\}}{B(\corr{1}, \conc{1})B(\corr{2}, \conc{2})}\int_0^x \variantfreqperpop{2}^{-\rate{1}-\corr{1}-\corr{2}} \variantfreqperpop{2}^{\corr{2}} (1-\variantfreqperpop{2})^{\conc{2}-1} \int_0^{\variantfreqperpop{2}} \variantfreqperpop{1}^{\corr{1}-1}\de\vecvariantfreq{}\\
    =&\frac{\mass 2^{-\rate{1}-\corr{1}-\corr{2}}\min\{1, 2^{1-\conc{1}}\}}{B(\corr{1}, \conc{1})B(\corr{2}, \conc{2})\corr{1}}\int_0^x \variantfreqperpop{2}^{-\rate{1}-\corr{1}-\corr{2}} \variantfreqperpop{2}^{\corr{2}} (1-\variantfreqperpop{2})^{\conc{2}-1} \variantfreqperpop{2}^{\corr{1}}\de\variantfreqperpop{2}\\
    \sim & \frac{\mass 2^{-\rate{1}-\corr{1}-\corr{2}}\min\{1, 2^{1-\conc{1}}\}}{B(\corr{1}, \conc{1})B(\corr{2}, \conc{2})\corr{1}}\int_0^x \variantfreqperpop{2}^{-\rate{1}} \de\variantfreqperpop{2}\\
    =& \frac{\mass 2^{-\rate{1}-\corr{1}-\corr{2}}\min\{1, 2^{1-\conc{1}}\}}{B(\corr{1}, \conc{1})B(\corr{2}, \conc{2})\corr{1}} \frac{1}{1-\rate{1}}x^{1-\rate{1}}
\end{align*}
we have by~\cref{lemma:broderick12_6.1} and~\cref{eq:poissonized_error}
\begin{align*}
    \E\left[\news{\zerovec}{(0, \samplesize{2})} \right] \ge \lowerbound((0, \samplesize{2}))\sim \constantlowerwithotherrate{1} \samplesize{2}^{\rate{1}} 
\end{align*}
Combining the two bounds we can also conclude that 
\begin{align*}
   \E \left[\news{\zerovec}{(0, \samplesize{2})}\right]\ge \lowerbound((0, \samplesize{2},0))\sim \max\{\constantlowerwithotherrate{1}\samplesize{2}^{\rate{1}}, \constantlowerwithrateprime{1}\samplesize{2}^{\rate{2}'} \}
\end{align*}

When $\rate{2}/\rate{1}\le 1$ the proof essentially switched the role of population 1 and population 2 in previous proofs. 

We start with $\poissonize((t_1,0))$. 

Consider the range $1/2 \ge \variantfreqperpop{1} \ge \variantfreqperpop{2}^{\rate{2}/\rate{1}}\ge \variantfreqperpop{2}  $, where $\variantfreqperpop{1}+\variantfreqperpop{2}^{\rate{2}/\rate{1}}\le 2\variantfreqperpop{1}$ and $\variantfreqperpop{1}+\variantfreqperpop{2}\le 2\variantfreqperpop{1}$

We have 
\begin{align*}
    \frac{(\variantfreqperpop{1}+\variantfreqperpop{2}^{\rate{2}/\rate{1}})^{-\rate{1}}}{(\variantfreqperpop{1}+\variantfreqperpop{2})^{\corr{1}+\corr{2}}}&=(\variantfreqperpop{1}+\variantfreqperpop{2}^{\rate{2}/\rate{1}})^{-\rate{1}}(\variantfreqperpop{1}+\variantfreqperpop{2})^{-\corr{1}-\corr{2}}\\
    &\ge \left\{ (2\variantfreqperpop{1} )^{-\rate{1}} \right\} \left\{(2\variantfreqperpop{1})^{-\corr{1} - \corr{2}} \right\}\bm{1}_{1/2 \ge \variantfreqperpop{1} \ge \variantfreqperpop{2}^{\rate{2}/\rate{1}}\ge \variantfreqperpop{2}  } \\
    &= 2^{-\rate{1}-\corr{1}-\corr{2}} \variantfreqperpop{1}^{-\rate{1}-\corr{1}-\corr{2}}\bm{1}_{1/2 \ge \variantfreqperpop{1} \ge \variantfreqperpop{2}^{\rate{2}/\rate{1}}\ge \variantfreqperpop{2}  }.
\end{align*}
Combining with  $(1-\variantfreqperpop{2})^{\conc{2}-1}\le \min\{1, 2^{1-\conc{2}}\}$ when $\variantfreqperpop{2}\le 1/2$ we have
\begin{align*}
    \ratemeasproposed(\de\vecvariantfreq{})\ge\mu_{s,1}(\de\vecvariantfreq{}):= \frac{\mass 2^{-\rate{1}-\corr{1}-\corr{2}}\min\{1, 2^{1-\conc{2}}\}}{B(\corr{1}, \conc{1})B(\corr{2}, \conc{2})}\variantfreqperpop{1}^{-\rate{1}-\corr{1}-\corr{2}} \variantfreqperpop{2}^{\corr{2}-1}\variantfreqperpop{1}^{\corr{1}-1}(1-\variantfreqperpop{1})^{\conc{1}-1} \bm{1}_{1/2 \ge \variantfreqperpop{1} \ge \variantfreqperpop{2}^{\rate{2}/\rate{1}}\ge \variantfreqperpop{2}}
\end{align*}
Use the same arguments as before, it remains to check whether $\mu_{s,1}(\de\vecvariantfreq{})$ satisfy conditions in~\cref{lemma:broderick12_6.1}. We can focus on $x\le 1/2$ since we will want $x\to 0$. 
\begin{align*}
    &\int_0^x \variantfreqperpop{1} \int_0^1 \mu_{s,2}(\de\vecvariantfreq{})\\
    =& \frac{\mass 2^{-\rate{1}-\corr{1}-\corr{2}}\min\{1, 2^{1-\conc{2}}\}}{B(\corr{1}, \conc{1})B(\corr{2}, \conc{2})}\int_0^x \variantfreqperpop{1} \int_0^{\variantfreqperpop{1}^{\rate{1}/\rate{2}}} \variantfreqperpop{1}^{-\rate{1}-\corr{1}-\corr{2}} \variantfreqperpop{1}^{\corr{2}-1}\variantfreqperpop{2}^{\corr{1}-1}(1-\variantfreqperpop{1})^{\conc{1}-1}\de\vecvariantfreq{}\\
    =& \frac{\mass 2^{-\rate{1}-\corr{1}-\corr{2}}\min\{1, 2^{1-\conc{2}}\}}{B(\corr{1}, \conc{1})B(\corr{2}, \conc{2})}\int_0^x \variantfreqperpop{1}^{-\rate{1}-\corr{1}-\corr{2}} \variantfreqperpop{1}^{\corr{2}} (1-\variantfreqperpop{1})^{\conc{1}-1} \int_0^{\variantfreqperpop{1}^{\rate{1}/\rate{2}}} \variantfreqperpop{2}^{\corr{2}-1}\de\vecvariantfreq{}\\
    =&\frac{\mass 2^{-\rate{1}-\corr{1}-\corr{2}}\min\{1, 2^{1-\conc{2}}\}}{B(\corr{1}, \conc{1})B(\corr{2}, \conc{2})\corr{2}}\int_0^x \variantfreqperpop{1}^{-\rate{1}-\corr{1}-\corr{2}} \variantfreqperpop{1}^{\corr{2}} (1-\variantfreqperpop{1})^{\conc{1}-1} \variantfreqperpop{1}^{\corr{2}\frac{\rate{1}}{\rate{2}}}\de\variantfreqperpop{1}\\
    \sim & \frac{\mass 2^{-\rate{1}-\corr{1}-\corr{2}}\min\{1, 2^{1-\conc{2}}\}}{B(\corr{1}, \conc{1})B(\corr{2}, \conc{2})\corr{2}}\int_0^x \variantfreqperpop{1}^{-(\rate{1}-\corr{2}(\frac{\rate{1}}{\rate{2}}-1))} \de\variantfreqperpop{1}\\
    =& \frac{\mass 2^{-\rate{1}-\corr{1}-\corr{2}}\min\{1, 2^{1-\conc{2}}\}}{B(\corr{1}, \conc{1})B(\corr{2}, \conc{2})\corr{2}} \frac{1}{1-(\rate{1}-\corr{2}(\frac{\rate{1}}{\rate{2}}-1))}x^{1-(\rate{1}-\corr{2}(\frac{\rate{1}}{\rate{2}}-1))}
\end{align*}
we have by~\cref{lemma:broderick12_6.1} and~\cref{eq:poissonized_error}
\begin{align*}
    \E \left[\news{\zerovec}{(\samplesize{1}, 0)}\right]\ge \lowerbound((\samplesize{1},0))\sim \constantlowerwithrateprime{2} \samplesize{1}^{\rate{1}-\corr{2}(\frac{\rate{1}}{\rate{2}}-1)} 
\end{align*}

Alternatively we observe that $(\variantfreqperpop{1}+\variantfreqperpop{2}^{\rate{2}/\rate{1}})^{-\rate{1}}\ge (\variantfreqperpop{1}^{\rate{2}/\rate{1}}+\variantfreqperpop{2}^{\rate{2}/\rate{1}})^{-\rate{1}}$ in the range of $1/2 \ge \variantfreqperpop{1}\ge \variantfreqperpop{2}$ and $\variantfreqperpop{1}^{\rate{2}/\rate{1}}+\variantfreqperpop{2}^{\rate{2}/\rate{1}} \le 2 \variantfreqperpop{1}^{\rate{2}/\rate{1}}$

We have 
\begin{align*}
    \frac{(\variantfreqperpop{1}+\variantfreqperpop{2}^{\rate{2}/\rate{1}})^{-\rate{1}}}{(\variantfreqperpop{1}+\variantfreqperpop{2})^{\corr{1}+\corr{2}}}&=(\variantfreqperpop{1}+\variantfreqperpop{2}^{\rate{2}/\rate{1}})^{-\rate{1}}(\variantfreqperpop{1}+\variantfreqperpop{2})^{-\corr{1}-\corr{2}}\\
    &\ge (\variantfreqperpop{1}^{\rate{2}/\rate{1}}+\variantfreqperpop{2}^{\rate{2}/\rate{1}})^{-\rate{1}}(\variantfreqperpop{1}+\variantfreqperpop{2})^{-\corr{1}-\corr{2}}\bm{1}_{1/2 \ge \variantfreqperpop{1}\ge \variantfreqperpop{2} }\\
    &\ge \left\{ (2\variantfreqperpop{1}^{\rate{2}/\rate{1}} )^{-\rate{1}} \right\} \left\{(2\variantfreqperpop{1})^{-\corr{1} - \corr{2}} \right\}\bm{1}_{1/2 \ge \variantfreqperpop{1}\ge \variantfreqperpop{2} } \\
    &= 2^{-\rate{1}-\corr{1}-\corr{2}} \variantfreqperpop{1}^{-\rate{2}-\corr{1}-\corr{2}}\bm{1}_{1/2 \ge \variantfreqperpop{1}\ge \variantfreqperpop{2} }.
\end{align*}
Combining with $(1-\variantfreqperpop{2})^{\conc{2}-1}\le \min\{1, 2^{1-\conc{2}}\}$ as $\variantfreqperpop{2}\le 1/2$, we have 
\begin{align*}
    \ratemeasproposed(\de\vecvariantfreq{})\ge\mu'_{s,1}(\de\vecvariantfreq{}):= \frac{\mass 2^{-\rate{2}-\corr{1}-\corr{2}}\min\{1, 2^{1-\conc{2}}\}}{B(\corr{1}, \conc{1})B(\corr{2}, \conc{2})}\variantfreqperpop{1}^{-\rate{2}-\corr{1}-\corr{2}} \variantfreqperpop{1}^{\corr{2}-1}\variantfreqperpop{2}^{\corr{2}-1}(1-\variantfreqperpop{1})^{\conc{1}-1} \bm{1}_{1/2 \ge \variantfreqperpop{1} \ge \variantfreqperpop{2}}
\end{align*} 

We check the conditions in~\cref{lemma:broderick12_6.1}
\begin{align*}
    &\int_0^x \variantfreqperpop{2} \int_0^1 \mu'_{s,1}(\de\vecvariantfreq{})\\
    =& \frac{\mass 2^{-\rate{2}-\corr{1}-\corr{2}}\min\{1, 2^{1-\conc{2}}\}}{B(\corr{1}, \conc{1})B(\corr{2}, \conc{2})}\int_0^x \variantfreqperpop{1} \int_0^{\variantfreqperpop{1}} \variantfreqperpop{1}^{-\rate{2}-\corr{1}-\corr{2}} \variantfreqperpop{1}^{\corr{1}-1}\variantfreqperpop{2}^{\corr{2}-1}(1-\variantfreqperpop{1})^{\conc{1}-1}\de\vecvariantfreq{}\\
    =& \frac{\mass 2^{-\rate{2}-\corr{1}-\corr{2}}\min\{1, 2^{1-\conc{2}}\}}{B(\corr{1}, \conc{1})B(\corr{2}, \conc{2})}\int_0^x \variantfreqperpop{1}^{-\rate{2}-\corr{1}-\corr{2}} \variantfreqperpop{1}^{\corr{1}} (1-\variantfreqperpop{1})^{\conc{1}-1} \int_0^{\variantfreqperpop{1}} \variantfreqperpop{2}^{\corr{2}-1}\de\vecvariantfreq{}\\
    =&\frac{\mass 2^{-\rate{2}-\corr{1}-\corr{2}}\min\{1, 2^{1-\conc{2}}\}}{B(\corr{1}, \conc{1})B(\corr{2}, \conc{2})\corr{2}}\int_0^x \variantfreqperpop{1}^{-\rate{2}-\corr{1}-\corr{2}} \variantfreqperpop{1}^{\corr{1}} (1-\variantfreqperpop{1})^{\conc{1}-1} \variantfreqperpop{1}^{\corr{2}}\de\variantfreqperpop{1}\\
    \sim & \frac{\mass 2^{-\rate{2}-\corr{1}-\corr{2}}\min\{1, 2^{1-\conc{2}}\}}{B(\corr{1}, \conc{1})B(\corr{2}, \conc{2})\corr{2}}\int_0^x \variantfreqperpop{1}^{-\rate{2}} \de\variantfreqperpop{1}\\
    =& \frac{\mass 2^{-\rate{2}-\corr{1}-\corr{2}}\min\{1, 2^{1-\conc{2}}\}}{B(\corr{1}, \conc{1})B(\corr{2}, \conc{2})\corr{2}} \frac{1}{1-\rate{2}}x^{1-\rate{2}}
\end{align*}

we have by~\cref{lemma:broderick12_6.1} and~\cref{eq:poissonized_error}
\begin{align*}
    \E\left[\news{\zerovec}{(\samplesize{1}, 0)}\right]\ge \altlowerbound(( \samplesize{1},0))\sim \constantlowerwithotherrate{2} \samplesize{1}^{\rate{2}} 
\end{align*}

Combine the two bounds we can also conclude that 
\begin{align*}
    \E\left[\news{\zerovec}{(\samplesize{1}, 0)}\right]\ge \lowerbound(( \samplesize{1},0))\sim \max\{\constantlowerwithotherrate{2}\samplesize{1}^{\rate{2}}, \constantlowerwithrateprime{2}\samplesize{1}^{\rate{1}'} \}
\end{align*}

For $\poissonize((0, t_2))$, we consider the range $1/2 \ge \variantfreqperpop{2}\ge \variantfreqperpop{1}$, where we also have $\variantfreqperpop{2}^{\rate{2}/\rate{1}}\ge \variantfreqperpop{1}$, therefore we have 
\begin{align*}
    \frac{(\variantfreqperpop{1}+\variantfreqperpop{2}^{\rate{2}/\rate{1}})^{-\rate{1}}}{(\variantfreqperpop{1}+\variantfreqperpop{2})^{\corr{1}+\corr{2}}}&=(\variantfreqperpop{1}+\variantfreqperpop{2}^{\rate{2}/\rate{1}})^{-\rate{1}}(\variantfreqperpop{1}+\variantfreqperpop{2})^{-\corr{1}-\corr{2}}\\
    &\ge \left\{ (2\variantfreqperpop{2}^{\rate{2}/\rate{1}} )^{-\rate{1}} \right\} \left\{(2\variantfreqperpop{2})^{-\corr{1} - \corr{2}} \right\}\bm{1}_{1/2 \ge \variantfreqperpop{2}\ge \variantfreqperpop{1} } \\
    &= 2^{-\rate{1}-\corr{1}-\corr{2}} \variantfreqperpop{2}^{-\rate{2}-\corr{1}-\corr{2}}\bm{1}_{1/2 \ge \variantfreqperpop{2}\ge \variantfreqperpop{1} }.
\end{align*}

Combining with $(1-\variantfreqperpop{1})^{\conc{1}}\le \min\{1, 2^{1-\conc{1}}\}$ when $\variantfreqperpop{1}\le 1/2$ We have
\begin{align*}
    \ratemeasproposed(\de\vecvariantfreq{})\ge\mu_{s,1}(\de\vecvariantfreq{}):= \frac{\mass 2^{-\rate{1}-\corr{1}-\corr{2}}\min\{1, 2^{1-\conc{1}}\}}{B(\corr{1}, \conc{1})B(\corr{2}, \conc{2})}\variantfreqperpop{2}^{-\rate{2}-\corr{1}-\corr{2}} \variantfreqperpop{2}^{\corr{2}-1}\variantfreqperpop{1}^{\corr{1}-1}(1-\variantfreqperpop{2})^{\conc{2}-1} \bm{1}_{1/2 \ge \variantfreqperpop{2}\ge \variantfreqperpop{1}}
\end{align*}

We then check the condition of \cref{lemma:broderick12_6.1}. We can focus on $x\le 1/2$ since we will want $x\to 0$. 
\begin{align*}
    &\int_0^x \variantfreqperpop{1} \int_0^1 \mu_{s,1}(\de\vecvariantfreq{})\\
    =& \frac{\mass 2^{-\rate{1}-\corr{1}-\corr{2}}\min\{1, 2^{1-\conc{1}}\}}{B(\corr{1}, \conc{1})B(\corr{2}, \conc{2})}\int_0^x \variantfreqperpop{2} \int_0^{\variantfreqperpop{2}} \variantfreqperpop{2}^{-\rate{2}-\corr{1}-\corr{2}} \variantfreqperpop{2}^{\corr{2}-1}\variantfreqperpop{1}^{\corr{1}-1}(1-\variantfreqperpop{2})^{\conc{2}-1}\de\vecvariantfreq{}\\
    =& \frac{\mass 2^{-\rate{1}-\corr{1}-\corr{2}}\min\{1, 2^{1-\conc{1}}\}}{B(\corr{1}, \conc{1})B(\corr{2}, \conc{2})}\int_0^x \variantfreqperpop{2}^{-\rate{2}-\corr{1}-\corr{2}} \variantfreqperpop{2}^{\corr{2}} (1-\variantfreqperpop{2})^{\conc{2}-1} \int_0^{\variantfreqperpop{2}} \variantfreqperpop{1}^{\corr{1}-1}\de\vecvariantfreq{}\\
    =&\frac{\mass 2^{-\rate{1}-\corr{1}-\corr{2}}\min\{1, 2^{1-\conc{1}}\}}{B(\corr{1}, \conc{1})B(\corr{2}, \conc{2})\corr{1}}\int_0^x \variantfreqperpop{2}^{-\rate{2}-\corr{1}-\corr{2}} \variantfreqperpop{2}^{\corr{2}} (1-\variantfreqperpop{2})^{\conc{2}-1} \variantfreqperpop{2}^{\corr{1}}\de\variantfreqperpop{2}\\
    \sim & \frac{\mass 2^{-\rate{1}-\corr{1}-\corr{2}}\min\{1, 2^{1-\conc{1}}\}}{B(\corr{1}, \conc{1})B(\corr{2}, \conc{2})\corr{1}}\int_0^x \variantfreqperpop{2}^{-\rate{2}} \de\variantfreqperpop{2}\\
    =& \frac{\mass 2^{-\rate{1}-\corr{1}-\corr{2}}\min\{1, 2^{1-\conc{1}}\}}{B(\corr{1}, \conc{1})B(\corr{2}, \conc{2})\corr{1}} \frac{1}{1-\rate{2}}x^{1-\rate{2}}
\end{align*}

we have by \cref{lemma:broderick12_6.1} and \cref{eq:poissonized_error}
\begin{align*}
    \E \left[ \news{\zerovec}{(0,\samplesize{2})} \right] \ge \lowerbound((0,\samplesize{2}))\sim \constantlower{2} \samplesize{2}^{\rate{2}}.
\end{align*}

Combine the two bounds, we get the expectation part of the theorem.

Law of large number results are established by directly apply Proposition 6.4 of~\citet{Broderick2012}. We have $\news{\zerovec}{(0,\samplesize{2})}\sim \E \left[ \news{\zerovec}{(0,\samplesize{2})} \right]\ a.s. $ and $\news{\zerovec}{(\samplesize{1},0)}\sim \E \left[ \news{\zerovec}{(\samplesize{1},0)} \right]\ a.s. $.

\end{proof}

\subsection{Power law bounds in proportional scheme}
\label{sec:proportion_scheme_proof}
Here we give the proof of~\Cref{coro:proportional}. The key observation is that total number of variants seen can be upper bounded by the sum of variants seen in each population (the ``double counting'') and lower bounded by variants seen in either population. 

\begin{proof}
For a proportional sample $(\samplesize{1},\samplesize{2})=(\samplesize{},\lceil\propt\samplesize{}\rceil)$ with $\propt>0$ and $\samplesize{\popidx}\to\infty$ being integers. We first observe that by assumptions on sample size, we have $\samplesize{2}/\samplesize{1}\to\propt$

We can bound the number of new variants $\news{\bm{0}}\samplesize{}$ by
\begin{align*}
    \max\{\news{\bm{0}}{(\samplesize{},0)},\news{\bm{0}}{0,\frac{\samplesize{2}}{\samplesize{1}}\samplesize{}}\}\le \news{\bm{0}}{(\samplesize{},a\samplesize{})}\le \news{\bm{0}}{(\samplesize{},0)}+\news{\bm{0}}{(0,a\samplesize{})}\le 2\max\{\news{\bm{0}}{(\samplesize{},0)},\news{\bm{0}}{0,a\samplesize{}}\}
\end{align*}
Use the result in projection scheme. For $\rate{2}/\rate{1}\ge 1$, we have
\begin{align*}
    \E \left[\news{\bm{0}}{(\samplesize{},\frac{\samplesize{2}}{\samplesize{1}}\samplesize{})}\right]& \ge\max\{\lowerbound((\samplesize{1},0)),\lowerbound((0,\samplesize{2}))\} \\
    &\sim \max\{\constantlower{1}\samplesize{1}^{\rate{1}}, \constantlowerwithotherrate{1}\samplesize{2}^{\rate{2}-\corr{1}(\frac{\rate{2}}{\rate{1}}-1)}, \constantlowerwithotherrate{1}\samplesize{2}^{\rate{1}}\} \\
    &=\max\{\constantlower{1}(1+\frac{\samplesize{2}}{\samplesize{1}})^{\rate{1}}(\samplesize{1}+\samplesize{2})^{\rate{1}}, \\
    &~~~~~~\constantlowerwithotherrate{1}[(1+\frac{\samplesize{2}}{\samplesize{1}})/\frac{\samplesize{2}}{\samplesize{1}}]^{\rate{2}-\corr{1}(\frac{\rate{2}}{\rate{1}}-1)}(\samplesize{1}+\samplesize{2})^{\rate{2}-\corr{1}(\frac{\rate{2}}{\rate{1}}-1)}, \\
    &~~~~~~\constantlowerwithotherrate{1}[(1+\frac{\samplesize{2}}{\samplesize{1}})/\frac{\samplesize{2}}{\samplesize{1}}]^{\rate{1}}(\samplesize{1}+\samplesize{2})^{\rate{1}}\} \\
    &\ge \min\{\constantlower{1}(1+\frac{\samplesize{2}}{\samplesize{1}})^{\rate{1}}, \constantlowerwithotherrate{1}[(1+\frac{\samplesize{2}}{\samplesize{1}})/\frac{\samplesize{2}}{\samplesize{1}}]^{\rate{2}-\corr{1}(\frac{\rate{2}}{\rate{1}}-1)}, \constantlowerwithotherrate{1}[(1+\frac{\samplesize{2}}{\samplesize{1}})/\frac{\samplesize{2}}{\samplesize{1}}]^{\rate{1}}\} \times \\
    &~~~~~~(\samplesize{1}+\samplesize{2})^{\max\{\rate{1}, \rate{2}-\corr{1}(\frac{\rate{2}}{\rate{1}}-1) \}}\\
    &\to \min\{\constantlower{1}(1+\propt)^{\rate{1}}, \constantlowerwithotherrate{1}[(1+\propt)/\propt]^{\rate{2}-\corr{1}(\frac{\rate{2}}{\rate{1}}-1)}, \constantlowerwithotherrate{1}[(1+\propt)/\propt]^{\rate{1}}\} \times \\
    &~~~~~~(\samplesize{1}+\samplesize{2})^{\max\{\rate{1}, \rate{2}-\corr{1}(\frac{\rate{2}}{\rate{1}}-1) \}}
\end{align*}

For $\rate{2}/\rate{1}\le 1$, we have
\begin{align*}
    \E\left[\news{\bm{0}}{(\samplesize{},\frac{\samplesize{2}}{\samplesize{1}}\samplesize{})}\right]& \ge\max\{\lowerbound((\samplesize{1},0)),\lowerbound((0,\samplesize{2}))\} \\
    &\sim \max\{\constantlowerwithrateprime{2}\samplesize{1}^{\rate{2}}, \constantlower{2}\samplesize{2}^{\rate{2}}, \constantlowerwithotherrate{2}\samplesize{1}^{\rate{2}}\} \\
    &=\max\{\constantlowerwithrateprime{2}(1+\frac{\samplesize{2}}{\samplesize{1}})^{\rate{2}}(\samplesize{1}+\samplesize{2})^{\rate{2}}, \\
    &~~~~~~\constantlower{2}[(1+\frac{\samplesize{2}}{\samplesize{1}})/\frac{\samplesize{2}}{\samplesize{1}}]^{\rate{2}}(\samplesize{1}+\samplesize{2})^{\rate{2}}, \\
    &~~~~~~\constantlowerwithotherrate{2}(1+\frac{\samplesize{2}}{\samplesize{1}})^{\rate{1}-\corr{2}(\frac{\rate{1}}{\rate{2}}-1)}(\samplesize{1}+\samplesize{2})^{\rate{1}-\corr{2}(\frac{\rate{1}}{\rate{2}}-1)}\} \\
    &\ge \min\{\constantlowerwithrateprime{2}(1+\frac{\samplesize{2}}{\samplesize{1}})^{\rate{2}}, \constantlower{2}[(1+\frac{\samplesize{2}}{\samplesize{1}})/\frac{\samplesize{2}}{\samplesize{1}}]^{\rate{2}}, \constantlowerwithotherrate{2}(1+\frac{\samplesize{2}}{\samplesize{1}})^{\rate{1}-\corr{2}(\frac{\rate{1}}{\rate{2}}-1)}\} \times \\
    &~~~~~~(\samplesize{1}+\samplesize{2})^{\max\{\rate{2}, \rate{1}-\corr{2}(\frac{\rate{1}}{\rate{2}}-1) \}}\\
    &\to \min\{\constantlowerwithrateprime{2}(1+\propt)^{\rate{2}}, \constantlower{2}[(1+\propt)/\propt]^{\rate{2}}, \constantlowerwithotherrate{2}(1+\propt)^{\rate{1}-\corr{2}(\frac{\rate{1}}{\rate{2}}-1)}\} \times \\
    &~~~~~~(\samplesize{1}+\samplesize{2})^{\max\{\rate{2}, \rate{1}-\corr{2}(\frac{\rate{1}}{\rate{2}}-1) \}}
\end{align*}

Now we can have the upper bound, for any $\smallpositive>0$, we have
\begin{align*}
    \E\left[\news{\bm{0}}{(\samplesize{},\frac{\samplesize{2}}{\samplesize{1}}\samplesize{})}\right]& \le2\max\{\upperbound((\samplesize{1},0)),\lowerbound((0,\samplesize{2}))\} \\
    &\sim 2\max\{ \constantupper{1} \samplesize{1}^{\rate{1}+\smallpositive}, \constantupper{2}\samplesize{2}^{\rate{2}+\smallpositive}  \}\\
    &=2\max\{  \constantupper{1}(1+\frac{\samplesize{2}}{\samplesize{1}})^{\rate{1}+\smallpositive} (\samplesize{1}+\samplesize{2})^{\rate{1}+\smallpositive}, \constantupper{2}[(1+\frac{\samplesize{2}}{\samplesize{1}})/\frac{\samplesize{2}}{\samplesize{1}}]^{\rate{2}+\smallpositive}(\samplesize{1}+\samplesize{2})^{\rate{2}+\smallpositive} \}\\
    &\le 2\max\{ \constantupper{1}(1+\frac{\samplesize{2}}{\samplesize{1}})^{\rate{1}+\smallpositive}, \constantupper{2}[(1+\frac{\samplesize{2}}{\samplesize{1}})/\frac{\samplesize{2}}{\samplesize{1}}]^{\rate{2}+\smallpositive} \}(\samplesize{1}+\samplesize{2})^{\max\{\rate{1},\rate{2}\}+\smallpositive}\\
    &\to 2\max\{ \constantupper{1}(1+\propt)^{\rate{1}+\smallpositive}, \constantupper{2}[(1+\propt)/\propt]^{\rate{2}+\smallpositive} \}(\samplesize{1}+\samplesize{2})^{\max\{\rate{1},\rate{2}\}+\smallpositive}
\end{align*}

Law of large number results are established by same method of \citet{Broderick2012} and \citet{Gnedin2007}, by applying the large deviation bounds in \citet[pp.911, Theorem 4]{Freedman1973}, i.e., we have for any $\epsilon>0$, there exists an $\epsilon'$ only depends on $\epsilon$ and some $c$ such that
\begin{align*}
    P\left( \left| \frac{\news{\bm{0}}{(\samplesize{},\frac{\samplesize{2}}{\samplesize{1}}\samplesize{})}}{\E\left[\news{\bm{0}}{(\samplesize{},\frac{\samplesize{2}}{\samplesize{1}}\samplesize{})}\right]}-1\right|\ge \epsilon \right)<ce^{-\epsilon'\E\left[\news{\bm{0}}{(\samplesize{},\frac{\samplesize{2}}{\samplesize{1}}\samplesize{})}\right]}
\end{align*}
Using the derived upper bound on $\E\left[\news{\bm{0}}{(\samplesize{},\frac{\samplesize{2}}{\samplesize{1}}\samplesize{})}\right]$, the right hand side is summable, thus by Borrel Cantelli we have $\news{\bm{0}}{(\samplesize{},\frac{\samplesize{2}}{\samplesize{1}}\samplesize{})}\sim \E\left[\news{\bm{0}}{(\samplesize{},\frac{\samplesize{2}}{\samplesize{1}}\samplesize{})}\right]$ a.s..

\end{proof}

\subsection{Power law bounds when $\rate{1}=\rate{2}=\rate{}$}
\label{sec:two_rate_the_same}
We give proof of~\Cref{coro:proportional_power-law-rate-align} when $\rate{1}=\rate{2}=\rate{}$.

\begin{proof}
    We use the same bound from proportional scheme with $\rate{1}=\rate{2}=\rate{}$
    \begin{align*}
        \max\{ \E \news{\bm{0}}{\samplesize{1},0},\E \news{\bm{0}}{0,\samplesize{2}}\} \le \E \news{\bm{0}}{\samplesize{1},\samplesize{2}}\le \E\news{\bm{0}}{\samplesize{1},0}+\E\news{\bm{0}}{0,\samplesize{2}}\le  2\max\{ \E\news{\bm{0}}{\samplesize{1},0},\E\news{\bm{0}}{0,\samplesize{2}}\}
    \end{align*}
    We can lower bound the left hand side, with any $\min\{\corr{1},\corr{2}\}>\smallpositive>0$ by
    \begin{align*}
        \max\{ \E\news{\bm{0}}{\samplesize{1},0},\E\news{\bm{0}}{0,\samplesize{2}}\}&\ge \max\{ \constantlower{1}\samplesize{1}^{\rate{}}, \constantlowerwithotherrate{1}\samplesize{2}^{\rate{}} \}\\
        &\ge \min\{\constantlower{1},\constantlowerwithotherrate{1}\}\max\{\samplesize{1},\samplesize{2}\}^{\rate{}}\\
        &\ge 2^{-\rate{}}\min\{\constantlower{1},\constantlowerwithotherrate{1}\}(\samplesize{1}+\samplesize{2})^{\rate{}}
    \end{align*}
    And the right hand side can be upper bounded in a similar manner, for any $\smallpositive>0$:
    \begin{align*}
        2\max\{ \E\news{\bm{0}}{\samplesize{1},0},\E\news{\bm{0}}{0,\samplesize{2}}\}&\le 2\max\{\constantupper{1},\constantupper{2}\}\max\{\samplesize{1},\samplesize{2}\}^\rate{}\\
        &\le 2\max\{\constantupper{1},\constantupper{2}\}(\samplesize{1}+\samplesize{2})^{\rate{}+\smallpositive}
    \end{align*}

\end{proof}

\begin{figure}
    \centering
    \includegraphics[width = 0.6\textwidth]{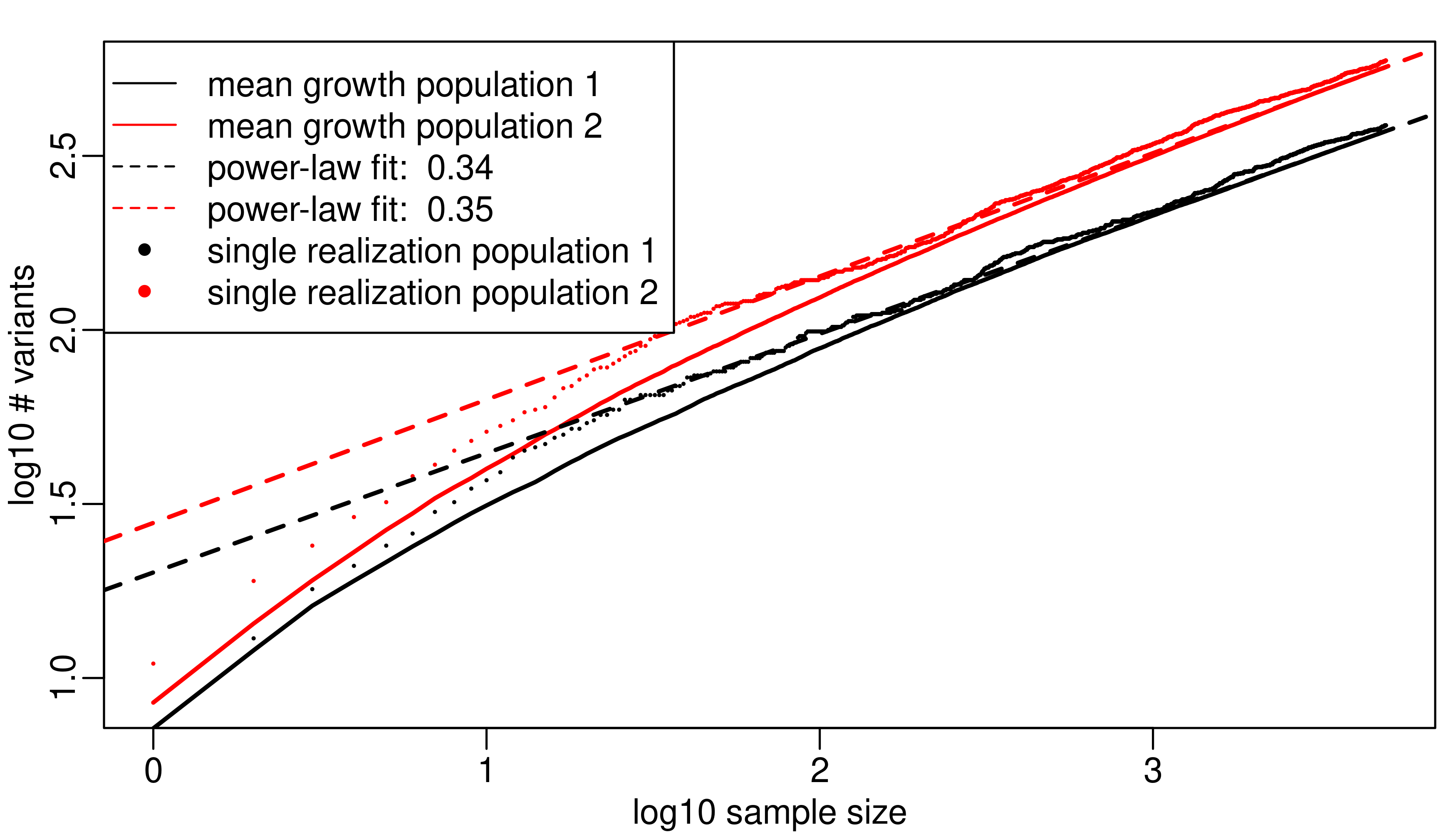}
    \caption{Simulated sampling curve with $\mass = 10, \rate{1}=0.3, \rate{2}=0.6, \conc{1}=\conc{2}=1, \corr{i}=\corr{2}=0.5$ in the projection scheme, i.e. sampling population 1 or population 2 only. The mean growth curves (shown in lines) were calculated using 100 independent realizations and a single realization is shown as dots. The bound for population 2 is 0.3 and 0.6 and the predicted power law rate for population 1 is 0.3. We fit the line using the last 1/3 of the mean growth curve. The fitted powerlaw align well with our theoretical prediction.}
\end{figure}

%\newpage
%\bibliography{refs}
\end{document}